%% file: sph_sym_master.tex
\documentclass[10pt, a4paper]{article}

\input{general.tex}

\title{Kasner-like description of spacelike singularities in spherically symmetric spacetimes with scalar matter}
\author{Warren Li}
\affil{\small Princeton University, Department of Mathematics, Fine Hall, Washington Road, Princeton, NJ 08544, USA}

\begin{document}

\maketitle

\begin{abstract}
    We study the properties of spacelike singularities in spherically symmetric spacetimes obeying the Einstein equations, in the presence of matter. \blue{Building upon previous work of An--Zhang \cite{AnZhang}}, we consider matter described by a scalar field, both in the presence of an electromagnetic field and without. We prove that, if a spacelike singularity obeying several reasonable assumptions is formed, then the Hawking mass, the Kretschmann scalar, and the matter fields have inverse polynomial blow-up rates near the singularity that may be described precisely.

    Furthermore, one may view the resulting spacetime in the context of the BKL heuristics regarding spacelike singularities in relativistic cosmology. In particular, near any point $p$ on the singular boundary in our spherically symmetric spacetime, we obtain a leading order BKL-type expansion, including a description of Kasner exponents associated to $p$. \blue{This confirms heuristics of Buonanno--Damour--Veneziano \cite{BuonannoDamourVeneziano}.} 

    \blue{As a result, we provide a } rigorous description of a detailed, \blue{quantitative} correspondence between Kasner-like singularities most often associated to the cosmological setting, and the singularities observed in (spherically symmetric) gravitational collapse. Moreover, we outline a program concerning the study of the stability and instability of spacelike singularities in the latter picture, both outside of spherical symmetry and within (where the electromagnetic field acts as a proxy for angular momentum). 
\end{abstract}

\tableofcontents

\input{introductionv2.tex}

\input{preliminaries.tex}

\input{statement.tex}

\input{upperbounds.tex}

\input{scalarfield.tex}

\input{bkl.tex}




\input{examples.tex}

\bibliography{bibliography_master.bib}
\bibliographystyle{abbrvnat_mod}

\end{document}

%% file: general.tex
\usepackage[german,english]{babel} 
\usepackage{booktabs} 
\usepackage{comment} 
\usepackage[utf8]{inputenc} 
\usepackage[T1]{fontenc} 
\usepackage{caption, soul}
\usepackage{authblk}

\usepackage{bigints,bbm,slashed,mathtools,amssymb,amsmath,amsfonts,amsthm} 
\usepackage{physics}
\usepackage{cancel}
\usepackage{url}
\usepackage{mathrsfs}
\usepackage{enumerate}
\usepackage{tikz} 
\usepackage{tikz-cd}
\usepackage{pgfplots}
\pgfplotsset{compat=1.16}
\usetikzlibrary{babel}
\usetikzlibrary{positioning,arrows} 
\usetikzlibrary{decorations.pathreplacing} 
\usetikzlibrary{patterns}
\usepackage{tikzsymbols} 
\usepackage{blindtext} 

\usepackage{geometry} 

\usepackage{setspace} 
\usepackage[T1]{fontenc} 
\usepackage[utf8]{inputenc} 

\usepackage{XCharter} 

\usepackage{fancyhdr} 
\usepackage[hang,flushmargin]{footmisc}

\usepackage[square, numbers]{natbib} 
\usepackage{har2nat} 
\usepackage{hyperref}

\numberwithin{equation}{section}
\theoremstyle{definition}
\newtheorem{definition}{Definition}
\theoremstyle{plain}
\newtheorem{theorem}{Theorem}[section]
\newtheorem{proposition}[theorem]{Proposition}
\newtheorem{lemma}[theorem]{Lemma}
\newtheorem{corollary}[theorem]{Corollary}

\newtheorem{innercustomgeneric}{\customgenericname}
\providecommand{\customgenericname}{}
\newcommand{\newcustomtheorem}[2]{%
  \newenvironment{#1}[1]
  {%
   \renewcommand\customgenericname{#2}%
   \renewcommand\theinnercustomgeneric{##1}%
   \innercustomgeneric
  }
  {\endinnercustomgeneric}
}

\newcustomtheorem{customthm}{Theorem}
\newcustomtheorem{customlemma}{Lemma}
\newcustomtheorem{customcorollary}{Corollary}

\theoremstyle{remark}
\newtheorem*{remark}{Remark}

\newcommand{\blue}[1]{{{#1}}}
\newcommand{\R}{\mathbb{R}}

%% file: introductionv2.tex
\section{Introduction} \label{intro}

\subsection{Background} \label{intro.background}

One of the monumental predictions of Einstein's theory of general relativity is that of ``singularities''. Indeed, celebrated theorems of Penrose \cite{penrose} and Hawking \cite{hawking67} \blue{prove} that globally hyperbolic spacetimes $(\mathcal{M}, g)$ satisfying the Einstein equations
\begin{equation} \label{eq:einstein}
    G_{\mu\nu}[g] \coloneqq \mathrm{Ric}_{\mu\nu}[g] - \frac{1}{2} R[g] g_{\mu\nu} = 2 \, T_{\mu\nu},
\end{equation}
as well as several \blue{physically reasonable} conditions regarding the spatial topology and the matter fields, will contain ``singularities'', so long as either of the following are present: (a) a codimension $2$ trapped surface, or (b) a Cauchy hypersurface whose mean curvature is bounded strictly away from $0$.

In both cases, however, the phenomenon captured by these theorems is not actually singularity formation as this term is usually understood -- suggesting some type of blow-up -- but rather geodesic incompleteness of the spacetime $(\mathcal{M}, g)$. 
Indeed, even in the absence of matter, there are known examples of globally hyperbolic spacetimes, for instance the Kerr family or the Taub-NUT spacetimes (see, e.g.~\cite{hawking_ellis_1973}), satisfying the assumptions of the Hawking-Penrose theorems, yet exhibiting no coordinate-independent blowup. 

In fact, \blue{the above examples} can be smoothly and isometrically embedded into a larger spacetime $(\tilde{\mathcal{M}}, \tilde{g})$ in such a way that all previously incomplete geodesics can be extended into $(\tilde{\mathcal{M}}, \tilde{g})$. Hence, from the point of view of observers travelling along such geodesics, there is nothing `singular' about the predictions of these theorems, and their geodesic incompleteness is instead due to the presence of so-called Cauchy horizons. In view of the strong cosmic censorship conjecture, one could ask about the non-genericity of such Cauchy horizons, and whether the presence of singularities is a true prediction of Einstein's equation (\ref{eq:einstein}). 

There are, of course, known solutions to the Einstein equations where genuine singularities are present. The celebrated Schwarzschild family of spacetimes are vacuum solutions to (\ref{eq:einstein}) depending on a single parameter $M > 0$. They each contain a black hole interior region, described by the metric $g_{Schw}$ on $(0, 2M) \times \R \times \mathbb{S}^2$:
\begin{equation} \label{eq:schwarzschild}
    g_{Schw} = - \left( \frac{2M}{r} - 1 \right)^{-1} \, dr^2 + \left( \frac{2M}{r} - 1 \right) \, dt^2 + r^2 \left( d \theta^2 + \sin^2 \theta \, d \varphi^2 \right).
\end{equation}

Treating $r$ as a (past-directed) timelike coordinate, the interior region is bounded to the future by a \textit{spacelike singularity} at $r = 0$. The truly singular nature of this boundary is reflected in the fact that the \textit{Kretschmann scalar} can be computed to be
\begin{equation*}
    \mathrm{Riem}[g_{Schw}]_{\alpha \beta \gamma \delta} \,
    \mathrm{Riem}[g_{Schw}]^{\alpha \beta \gamma \delta} 
    = \frac{48 M^2}{r^6}
\end{equation*}
and thus blows up towards the singularity at $r = 0$. In fact, there does not even exist a $C^0$ extension of \eqref{eq:schwarzschild} across the $ r = 0 $ boundary \cite{SbierskiC0}. 

Another example of a singular vacuum solution to (\ref{eq:einstein}) is the Kasner solution \cite{Kasner}, defined on the manifold $\mathcal{M} = (0, + \infty) \times \mathbb{T}^3$ and given by the metric $g_{Kas}$:
\begin{equation} \label{eq:kasner}
    g_{Kas} = - dt^2 + \sum_{i = 1}^3 t^{2p_i} (dx^i)^2.
\end{equation}
In order for (\ref{eq:kasner}) to satisfy Einstein's equations (\ref{eq:einstein}) with $T_{\mu\nu} = 0$, the \textit{Kasner exponents} $p_i \in \R$ must be chosen to satisfy the two relations $\sum_{i=1}^3 p_i = 1$, $\sum_{i=1}^3 p_i^2 = 1$. These relations imply that $p_1, p_2, p_3$ are constrained to lie on the so-called \textit{Kasner circle}, parameterized by $u \in (- \infty , + \infty]$ as follows:
\begin{equation} \label{eq:kasnercircle}
    p_1(u) = - \frac{u}{1 + u + u^2}, 
    \quad p_2(u) = \frac{1 + u}{1 + u + u^2},
    \quad p_3(u) = \frac{u(1 + u)}{1 + u + u^2}.
\end{equation}

Like Schwarzschild, the Kasner spacetimes possess a singular boundary at $t = 0$, and the singular nature of the boundary is manifested by the following blow-up of the Kretschmann scalar as $t \to 0^+$:
\begin{equation*}
    \mathrm{Riem}[g_{Kas}]_{\alpha \beta \gamma \delta} \,
    \mathrm{Riem}[g_{Kas}]^{\alpha \beta \gamma \delta} 
    = \frac{4}{t^4} \cdot \left( \sum_{i=1}^3 p_i^2 (p_i - 1)^2 + \sum_{i < j} p_i^2 p_j^2 \right)
    = \frac{16 u^2 (1 + u)^2}{(1 + u + u^2)^3} \cdot \frac{1}{t^4}.
\end{equation*}
Here $u \in \{-1, 0, + \infty\}$, which correspond to $(p_1, p_2, p_3) = (1, 0, 0)$ and permutations thereof, give rise to exceptional Kasner solutions that are flat and thus locally isometric to Minkowski space. These can be locally extended beyond the $t = 0$ boundary, and we exclude these solutions (so-called `special points' of the Kasner circle) from the subsequent discussion.
    
The metrics (\ref{eq:schwarzschild}) and (\ref{eq:kasner}) provide evidence that genuine singularities may arise in solutions to Einstein's equation (\ref{eq:einstein}). But both are highly symmetric solutions, and to argue that singularities are a genuine prediction of general relativity, one must argue that generic, non-symmetric initial data would also produce such singularities in evolution.

In a highly influential paper, Khalatnikov and Lifshitz \cite{kl63} propose an ansatz for the asymptotic form of a metric solving the vacuum Einstein equation near a candidate spacelike
singularity, of which (\ref{eq:schwarzschild}) and (\ref{eq:kasner}) are special examples. However, via a function counting argument, the authors of \cite{kl63} argue that their singular solutions are exceptional among solutions to Einstein's vacuum equations arising from general initial data, and further provide a heuristic instability mechanism; loosely speaking, the instability is sourced by spatial curvature terms. This will be elaborated upon in Section \ref{intro.bkl}.

In this article, we study spherically symmetric solutions to Einstein's equation (\ref{eq:einstein}). By Birkhoff's theorem \cite{BirkhoffRelativityAM}, spherically symmetric solutions to the Einstein \underline{vacuum} equations are non-dynamical and are (locally) isometric to some region of the Schwarzschild solution 
with mass $M \in \R$, where $M= 0$ corresponds to Minkowski space. We thus introduce dynamics into the problem via the addition of matter, namely a (massless and uncharged) scalar field $\phi$, whose energy-momentum tensor is given by
\begin{equation} \label{eq:energymomentum_scalar}
    T_{\mu \nu}[\phi] = \nabla_{\mu} \phi \nabla_{\nu} \phi - \frac{1}{2} g_{\mu \nu} \nabla^{\rho} \phi \nabla_{\rho} \phi.
\end{equation}

The scalar field $\phi$ evolves according to the scalar wave equation on $(\mathcal{M}, g)$:
\begin{equation} \label{eq:scalar_wave}
    \square_g \phi = g^{\mu \nu} \nabla_{\mu} \nabla_{\nu} \phi = 0.
\end{equation}

In studying the coupled system (\ref{eq:einstein}), (\ref{eq:energymomentum_scalar}), (\ref{eq:scalar_wave}) in spherical symmetry, the aforementioned instability mechanism of \cite{kl63} will be suppressed in two different ways. Firstly, we observe (see Section~\ref{intro.bkl}) that spherical symmetry causes the spatial curvature term sourcing the instability to vanish by assumption. Secondly, as observed by Belinskii and Khalatnikov in \cite{bk72}, the presence of scalar field matter allows for a more general class of asymptotic ansatz near the singularity, including so-called \textit{subcritical} regimes where the spatial curvature no longer provides an instability.

For these reasons, it is expected that the formation and asymptotics of general spacelike singularities can be understood in this setting. 
This is the objective of the present article, where we prove that spacelike singularities in the spherically symmetric Einstein-scalar field system can be described precisely, in the sense that in the neighborhood of a point $p$ on the singular boundary, we can describe the leading order blow-up rates of all geometric quantities (including the Kretschmann scalar and the Hawking mass) and matter quantities such as $\phi$ and its derivatives, in terms of the area-radius $r$ (see Section~\ref{intro.model} for a definition).

We further provide a precise correspondence between such spacelike singularities and the heuristics of \cite{bk72} -- in particular we associate to each point $p$ on the singular boundary a triplet of Kasner exponents. This can be viewed as a rigorous version of the correspondence previously found in the physics literature \cite{BuonannoDamourVeneziano}. See Section~\ref{sub:intro_rough_1} for a rough version of our main theorem, or Section~\ref{statement} for the detailed version.

Our work builds upon that of \cite{AnZhang}, where the authors provide upper bound estimates (corresponding exactly to those in Section~\ref{upperbounds}) for various geometric and matter quantities. The improvement in the present paper is that such upper bounds are upgraded to precise leading-order asymptotics, thus allowing us to make the correspondence to Kasner-like singularities. 

Another novel aspect of our work is that in order to provide an avenue within which one can study stable (i.e.~subcritical) and unstable regimes, we allow also an electromagnetic field. This additional matter field is represented by a $2$-form $F \in \Omega^2(\mathcal{M})$, and has the energy-momentum tensor:
\begin{equation} \label{eq:energymomentum_em}
    T_{\mu \nu}[F] = F_{\mu \rho}  F_{\nu}^{\phantom{\nu} \rho} - \frac{1}{4} g_{\mu \nu} F^{\rho \sigma} F_{\rho \sigma}.
\end{equation}
The equations of motion associated to $F$ are given by Maxwell's equations:
\begin{equation} \label{eq:maxwell}
    \nabla^{\nu} F_{\mu \nu} = 0, \qquad \nabla_{[\lambda} F_{\mu \nu]} = 0.
\end{equation}

As we discuss in Section~\ref{intro.em}, an electromagnetic field sources a similar instability to gravitational perturbations outside of symmetry, and the study of spacelike singularity formation becomes heavily linked to the subcritical regimes of \cite{bk72}. It is of great interest to extend the ideas of \cite{BuonannoDamourVeneziano, AnZhang} to the Einstein-Maxwell-scalar field model. In particular, the addition of the electromagnetic field will be a crucial component of future work, to be introduced in Section~\ref{intro.stab}. 

We now outline the remainder of the introduction. In Section~\ref{intro.model}, we provide a more thorough overview of the spherically symmetric system under investigation, with the aim of being able to state a rough version of our main theorems in Section~\ref{sub:intro_rough_1}. In the following Sections~\ref{intro.bkl} and \ref{intro.em}, we relate our main theorem to the heuristics of BKL \cite{bk72, bk77, bkl71, bkl82, kl63}, where Section \ref{intro.em} emphasizes in particular the role played by the electromagnetic field. In Section~\ref{intro.proof}, we outline several of the important features of the proof, while in the remaining Sections~\ref{intro.stab} and \ref{intro.related} we provide mathematical context and propose several further directions of study, in particular we mention the expected stability properties of our singular spacetimes in Section~\ref{intro.stab}.


\subsection{Einstein-scalar field in spherical symmetry} \label{intro.model}

We now introduce the spherically symmetric Einstein-scalar field model in question. In a series of works \cite{Christodoulou_formation, Christodoulou_BV, Christodoulou_wcc}, Christodoulou studied this system (without electromagnetism) as a toy model for understanding the weak and strong cosmic censorship conjectures of Penrose \cite{penrose_cc} (or see \cite{Christodoulou_cc} for a more modern treatment). After defining the key concepts associated to spherically symmetric spacetimes, we give a brief overview of Christodoulou's conclusions. 

A spherically symmetric spacetime $(\mathcal{M}, g)$ is such that there exists a group of isometries acting on $(\mathcal{M}, g)$ isomorphic to $SO(3)$, whose orbits are each copies of $\mathbb{S}^2$ smoothly embedded in $\mathcal{M}$, outside of some orbits given by single points whose union we call the \textit{center} $\Gamma$. 

Defining $\mathcal{Q} = \mathcal{M} / SO(3)$ to be the quotient manifold, we can write $\mathcal{M} = \mathcal{Q} \times \mathbb{S}^2$ where $\mathcal{Q}$ is a $1+1$-dimensional manifold with (possibly empty) boundary $\Gamma$, and $g$ may be written as the warped product:
\begin{equation*}
    g = g_{\mathcal{Q}} + r^2 d \sigma_{\mathbb{S}^2} = g_{\mathcal{Q}} + r^2 ( d \theta^2 + \sin^2 \theta \, d \blue{\varphi^2}).
\end{equation*}
Here $r \geq 0$ is a function invariant under the $SO(3)$ action known as the \textit{area-radius}, vanishing precisely at $\Gamma$, while $g_{\mathcal{Q}}$ is a $\blue{(1+1)}$-dimensional Lorentzian metric on $\mathcal{Q}$.

Such a Lorentzian manifold $\mathcal{Q}$ has local coordinate charts given by two null coordinates $(u, v)$, with respect to which the four dimensional metric $g$ may be written as 
\begin{equation} \label{eq:metric}
    g = - \Omega^2(u, v) \, du dv + r^2(u, v) \, d \sigma_{\mathbb{S}^2}.
\end{equation}
We assume (as is the case for asymptotically flat one- or two-ended spacetimes) that this coordinate system on $\mathcal{Q}$ is global. Hence, the geometry is described by two functions of $u$ and $v$, the area-radius $r(u, v)$ and the \textit{null lapse} $\Omega^2(u, v)$.

There is gauge freedom regarding the choice of the null coordinates $u$ and $v$; given smooth and monotonically increasing $U = U(u)$ and $V = V(v)$, we may alternatively write
\begin{equation*}
    g = - \tilde{\Omega}^2(U, V) \, dU dV + r^2(U, V) \, d \sigma_{\mathbb{S}^2},
\end{equation*}
so long as $\tilde{\Omega}^2(U(u), V(v)) \cdot \frac{d U}{du} \frac{d V}{dv} = \Omega^2(u, v)$. It will often be useful to consider quantities that are \textit{gauge-invariant}, such as the area-radius $r(u, v)$ or the \textit{Hawking mass} $m(u, v)$, given by
\begin{equation} \label{eq:hawkingmass}
    1 - \frac{2 m}{r} \coloneqq g(\nabla r, \nabla r) = - 4 \Omega^{-2} \partial_u r \, \partial_v r.
\end{equation}

We study the coupled system (\ref{eq:einstein}) and (\ref{eq:energymomentum_scalar}) for spherically symmetric metrics of the form (\ref{eq:metric}). The system of evolution equations for $r(u,v)$, $\Omega^2(u,v)$ and the scalar field $\phi(u, v)$ with respect to these coordinates will be described in Section \ref{setup.emsf}. 

In the work \cite{Christodoulou_formation, Christodoulou_BV, Christodoulou_wcc} of Christodoulou, it is established that the maximal globally hyperbolic future development of \textit{generic} one-ended initial data possesses a future-complete null infinity $\mathcal{I}^+$, and that the spacetime either disperses as in Figure \ref{fig:christodoulou_dispersive}, or possesses a black hole region bounded to the future by a spacelike boundary, at which $r$ tends to $0$ and the Kretschmann scalar blows up, see Figure~\ref{fig:christodoulou_collapse}.

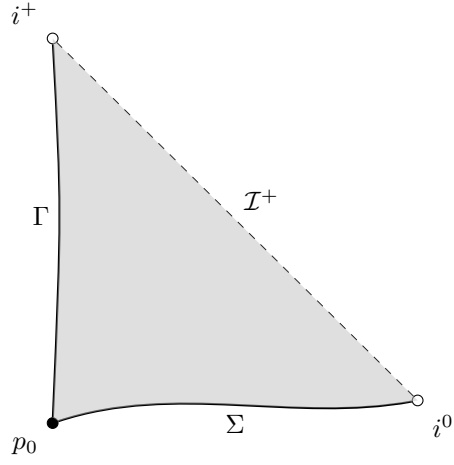
\begin{figure}[ht] 
    \centering
    \begin{tikzpicture}[scale=0.6]
        \node (s) at (0,-0.5) [circle, draw, inner sep=0.5mm, fill=black] {};
        \node [below left=0.2mm of s] {$p_0$};
        \node (i0) at (8, 0) [circle, draw, inner sep=0.5mm] {};
        \node [below right=0.2mm of i0] {$i^0$};
        \node (i+) at (0, 8) [circle, inner sep=0.5mm, draw] {};
        \node [above left=0.2mm of i+] {$i^+$};

        \draw [thick] (s) .. controls (2.7, 0.4) and (5.3, -0.5) ..  (i0)
            node (sigma) [midway, below] {$\Sigma$};
        \draw [dashed] (i0) -- (i+) node (nullinf) [midway, above right] {$\mathcal{I}^+$};
        \draw [thick] (i+) .. controls (0.2, 4.2) .. (s) node [midway, left] {$\Gamma$};

        \path[fill=lightgray, opacity=0.5] (s) .. controls (0.2, 4.2) .. (0, 8) -- (8, 0) 
            .. controls (5.3, -0.5) and (2.7, 0.4) ..  (s);
    \end{tikzpicture}
    \captionsetup{justification = centering}
    \caption{Penrose diagram representation of a dispersive solution to the Einstein-scalar field equations. This is a future causal geodesically complete spacetime with a complete future null infinity $\mathcal{I}^+$.}
    \label{fig:christodoulou_dispersive}
\end{figure}

In the latter case, \cite{Christodoulou_formation, Christodoulou_BV} applies soft monotonicity arguments based on the fact that $\partial_u \partial_v r^2 < 0$ (see already \eqref{eq:wave_r_u} in the case $Q \equiv 0$) in order to establish the spacelike character of the singular $\{r = 0\}$ boundary, as well as the blow-up of the Hawking mass and Kretschmann scalar. In order to make a comparison to the BKL ansatz for the spacelike singularity proposed in \cite{kl63, bkl71}, however, we shall require a more precise description of how such geometric quantities and the scalar field $\phi$ blow up.

\begin{figure}[ht] 
    \centering
    \begin{tikzpicture}[scale=0.8]
        \node (s) at (0,-0.5) [circle, draw, inner sep=0.5mm, fill=black] {};
        \node [below left=0.2mm of s] {$p_0$};
        \node (i0) at (8, 0) [circle, draw, inner sep=0.5mm] {};
        \node [below right=0.2mm of i0] {$i^0$};
        \node (i+) at (3, 5) [circle, inner sep=0.5mm, draw] {};
        \node [above right=0.2mm of i+] {$i^+$};
        \node (ss) at (0, 4.5) [circle, inner sep=0.5mm, draw] {};
        \node [above left=0.2mm of ss] {$b_0$};

        \path[fill=lightgray, opacity=0.5] (s) .. controls (2.7, 0.4) and (5.3, -0.5) .. (8, 0)
            -- (3, 5) .. controls (2, 5.5) and (1, 4.0) .. (0, 4.5) 
            .. controls (0.1, 2) .. (s);
        \path[fill=lightgray] (i+) .. controls (1.5, 4) and (0.5, 3.8)
            .. (0, 4.5)
            .. controls (1, 4.0) and (2, 5.5) .. (i+);

        \draw [thick] (s) .. controls (2.7, 0.4) and (5.3, -0.5) ..  (i0)
            node (sigma) [midway, below] {$\Sigma$};
        \draw [dashed] (i0) -- (i+) node (nullinf) [midway, above right] {$\mathcal{I}^+$};
        \draw [thick] (ss) .. controls (0.1, 2) .. (s) node [midway, left] {$\Gamma$};
        \draw (i+) -- (0.1, 2.1) node [midway, below right] {$\mathcal{H}^+$};
        \draw (i+) .. controls (1.5, 4) and (0.5, 3.8) .. (0, 4.5)
            node [midway, below] {$\mathcal{A}$};
        \draw [dashed] (i+) .. controls (2, 5.5) and (1, 4.0) .. (ss)
            node [midway, above] {$\mathcal{S}$};
    \end{tikzpicture}

    \captionsetup{justification = centering}
    \caption{Penrose diagram representation of a gravitational collapse solution to the Einstein-scalar field equations. The solution possesses a complete future null infinity $\mathcal{I}^+$ as well as a black hole region bounded to the past by the \textit{event horizon} $\mathcal{H}^+$. The black hole interior contains an \textit{apparent horizon} $\mathcal{A}$ at which $\partial_v r = 0$, whose future is a \textit{trapped region} $\mathcal{T}$ in which $\partial_v r < 0$ and which culminates at a spacelike boundary $\mathcal{S} = \{ r = 0 \}$ connecting $i^+$ with the first singularity at the center, denoted $b_0$.}
    \label{fig:christodoulou_collapse}
\end{figure}
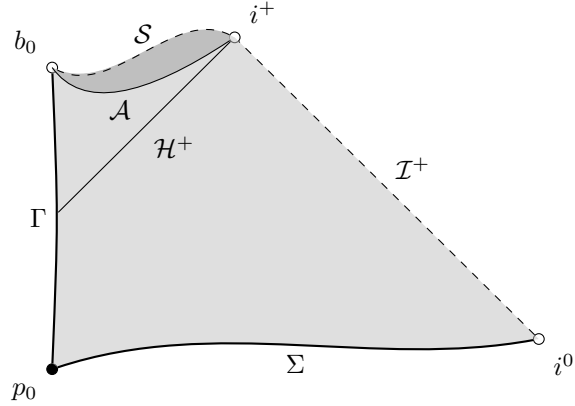

Progress to this effect is achieved in \blue{\cite{AnZhang, AnGajic, AnChenHe}}, where the authors derive \textit{upper bounds} for the quantities $r^2 \partial_u \phi$ and $r^2 \partial_v \phi$ near $\{ r = 0 \}$. Subsequently, these upper bounds on the scalar field are used to show inverse polynomial upper bounds for $\Omega^2$ and the Hawking mass $m$. Note that \cite{AnZhang} produces such estimates in a local neighborhood of a small portion of the spacelike singularity, while \cite{AnGajic} uses decay rates along the event horizon $\mathcal{H}^+$, originating from the heuristics of Price \cite{Price} and rigorously proved in \cite{dr_priceslaw}, to prove improved bounds near future spacelike infinity $i^+$, in line with the expectation that the singularity ``approaches Schwarzschild'' near $i^+$. \blue{Finally, \cite{AnChenHe} extends \cite{AnZhang, AnGajic} to problems involving a positive cosmological constant.}

In the present article, we build upon the upper bound estimates of \cite{AnZhang}, showing that in a regular double null gauge, the quantities $r^2 \partial_u \phi$ and $r^2 \partial_v \phi$ can be extended in a H\"older continuous fashion all the way to the singular boundary at $r = 0$. These asymptotics will in turn allow us to deduce exact rates for the polynomial blow-up of quantities such as $\Omega^2$, $m$, and the Kretschmann scalar. For example, we identify a quantity $\Psi(u, v)$, bounded and H\"older continuous up to the $r = 0$ boundary, such that
\begin{equation} \label{eq:omega_example}
    \Omega^2(u, v) \cdot r^{1 - \Psi^2(u, v)}(u, v) = C(u) + O(r).
\end{equation}

Here $C(u)$ should be thought of as a function `at the $r = 0$ singularity'. Asymptotics such as (\ref{eq:omega_example}) are then used to compare such spacelike singularities in the spherically symmetric Einstein-scalar field system to those predicted in the heuristics of BKL in \cite{kl63, bk72} -- to be \blue{introduced} in greater depth in Section \ref{intro.bkl}. 

Another new feature of the present article is the introduction of electromagnetism to the study of spacelike singularities in spherical symmetry. As mentioned previously, an electromagnetic field is represented by its Maxwell tensor, the $2$-form $F_{\mu\nu}$, and an energy-momentum tensor given by \eqref{eq:energymomentum_em}.
One then couples electromagnetism to gravity and scalar matter by studying the system \eqref{eq:einstein} with $T_{\mu \nu} = T_{\mu\nu}[\phi] + T_{\mu\nu}[F]$, alongside the matter evolution equations \eqref{eq:scalar_wave} and \eqref{eq:maxwell}.

The spherically symmetric Einstein-Maxwell-scalar field model has received much attention in the mathematical community, see \cite{dafermos05, LukOh1, LukOh2, kommemi}. These results, however, are largely concerned with the existence and properties of the Cauchy horizon in the interior of black holes. Indeed, the groundbreaking result of \cite{LukOh1, LukOh2} proves the $C^2$ formulation of the strong cosmic censorship conjecture in the context of the spherically symmetric Einstein-Maxwell-scalar field model. It is expected that similar results will also hold in the vacuum setting, see for instance \cite{KerrStab}.

On the other hand, the objective of this present article is to study the behaviour of the scalar field and geometry when the singularity is no longer \textit{null}, as in the case of the Cauchy horizon, but instead \textit{spacelike}. This is more subtle than the model without electromagnetism, for reasons to be explained in Section \ref{intro.em}. 

\subsection{Precise asymptotics -- rough version of our main theorems} \label{sub:intro_rough_1}

We now present a rough version of our main theorems. The setup is as follows: we pose characteristic initial data for the spherically symmetric Einstein-Maxwell-scalar field equations \eqref{eq:einstein}, \eqref{eq:scalar_wave}, \eqref{eq:maxwell} on the bifurcate null hypersurface $C_0 \cup \underline{C}_0$, where $C_0$ and $\underline{C}_0$ are respectively outgoing and ingoing null cones. 

By standard theory, we obtain a unique solution $(r, \Omega^2, \phi, F)$ to the above system of equations on a \textit{maximal future domain of development} $\mathcal{D} \subset \R^2$, attaining \blue{the} prescribed data on $C_0 \cup \underline{C}_0$, and such that $r, r^{-1}, \log \Omega^2, \phi$ and $F$ remain bounded and regular in compact subregions of $\mathcal{D}$. The domain $\mathcal{D}$ is \textit{globally hyperbolic} in the sense that any inextendible timelike curve in $\mathcal{D}$ has a past endpoint on the initial data hypersurface $C_0 \cup \underline{C}_0$. See already the Penrose diagram in Figure~\ref{fig:char_ivp0}.

Since we are interested in the \underline{local} properties of spacelike singularities, we assume already that $C_0$ and $\underline{C}_0$ are foliated by trapped $2$-spheres i.e.~that $\partial_u r < 0$ and $\partial_v r < 0$ on $C_0 \cup \underline{C}_0$. By Raychaudhuri's equations, this means that the maximal future development $\mathcal{D}$ will be entirely foliated by trapped $2$-spheres.

For now, we suppose that $C_0$ and $\underline{C}_0$ are \textit{regular} and \textit{compact}; \blue{in particular, the area-radius $r(u, v)$ restricted to $C_0 \cup \underline{C}_0$ remains bounded away from $0$.}
We will, however, assume the initial data on $C_0 \cup \underline{C}_0$ to be such that the maximal future development \blue{$\mathcal{D}$} will possess a future boundary which is an $\{ r = 0 \}$ singularity. \blue{That is,} the entire future boundary will consist of the union of two regular null hypersurfaces (emanating from the future endpoints of $C_0$ and $\underline{C}_0$) and a spacelike singularity $\mathcal{S}$ at which $r$ extends continuously to $0$.
We call such a maximal future development a \textit{strongly singular spacetime}. This is shown pictorially in Figure~\ref{fig:char_ivp0}, and will be described in greater detail in Section~\ref{setup.sing}.

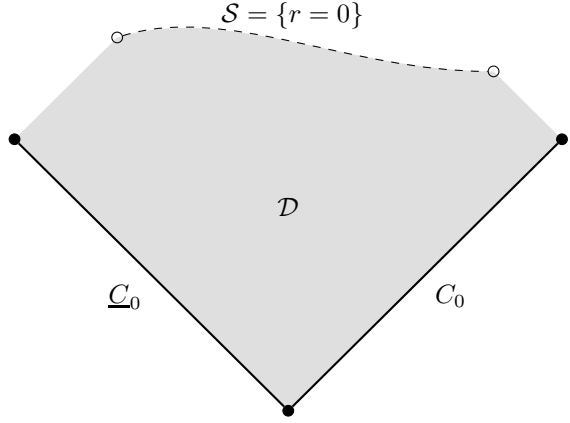
\begin{figure}[ht]
    \centering
    \begin{tikzpicture}[scale=0.9]
        \path[fill=lightgray, opacity=0.5] (0, -4) -- (-4, 0) -- (-2.5, 1.5)
            .. controls (-0.9, 2) and (0.9, 1) .. (3, 1)
            -- (4, 0) -- (0, -4);

        \node (p) at (0, -4) [circle, draw, inner sep=0.5mm, fill=black] {};
        \node (r) at (4, 0) [circle, draw, inner sep=0.5mm, fill=black] {};
        \node (l) at (-4, 0) [circle, draw, inner sep=0.5mm, fill=black] {};
        \node (rs) at (3, 1) [circle, draw, inner sep=0.5mm] {};
        \node (ls) at (-2.5, 1.5) [circle, draw, inner sep=0.5mm] {};

        \node at (0, -1) {$\mathcal{D}$};

        \draw [thick] (p) -- (r)
            node [midway, below right] {$C_0$};
        \draw [thick] (p) -- (l)
            node [midway, below left] {$\underline{C}_0$};
        \draw [dashed] (ls) .. controls (-0.9, 2) and (0.9, 1) .. (rs)
            node [midway, above=0.5mm] {$\mathcal{S} = \{ r = 0 \}$};
    \end{tikzpicture}

    \captionsetup{justification = centering}
    \caption{A strongly singular spacetime to which Theorem~\ref{roughthm:asymp} applies. The region $\mathcal{D}$ is foliated by trapped $2$-spheres, and possesses a future boundary $\mathcal{S}$ at which $r$ extends continuously to $0$. Examples of such strongly singular spacetimes exist both in the case $F_{\mu\nu} = 0$, and in the case that $F_{\mu \nu} \neq 0$ where one must additionally assume \eqref{eq:QTS} and \eqref{eq:SKE}.}
    \label{fig:char_ivp0}
\end{figure}

In the presence of a non-trivial electromagnetic field, we \blue{will have to} make the following \blue{two} additional assumptions on the strongly singular maximal development $\mathcal{D}$, where $R_0 > 0$ and $\alpha \in (0, 1)$ are fixed constants.
\begin{equation} \label{eq:QTS} \tag{QTS}
    \inf_{(u, v) \in \mathcal{D}} \min \{ - r \partial_v r(u, v), - r \partial_u r(u, v) \} \geq R_0,
\end{equation}
\begin{equation} \label{eq:SKE} \tag{SKE}
    \inf_{(u, v) \in \mathcal{D}} \min \left \{ \left| r \frac{\partial_v  \phi}{- \partial_v r}\right| (u, v), \left| r \frac{\partial_u \phi}{- \partial_u r}\right|(u, v) \right \} \geq \sqrt{1 + \alpha}.
\end{equation}
\eqref{eq:QTS} and \eqref{eq:SKE} are acronyms for \textit{quantitatively trapped spheres} and \textit{subcritical Kasner exponents} respectively. \blue{Here,} \eqref{eq:QTS}, which \blue{always holds} when $F_{\mu\nu} \equiv 0$ (since $\partial_u \partial_v r^2 < 0$ in this case), provides a mechanism to reach an $r = 0$ singularity in the first place, while we defer the discussion of the second assumption \eqref{eq:SKE} to Sections~\ref{intro.bkl} and \ref{intro.em}.

We briefly mention that the strongly singular assumption is not vacuous. In the $F_{\mu\nu} = 0$ case, there is a simple algebraic condition on the data at $C_0 \cup \underline{C}_0$ that guarantees that $\mathcal{D}$ is strongly singular, see already Lemma~\ref{lem:setup_esfss} and the subsequent remark. Alternatively, it is straightforward to find strongly singular subregions of the trapped region $\mathcal{T}$ in gravitational collapse spacetimes, such as that depicted by Figure~\ref{fig:christodoulou_collapse}. 

For $F_{\mu\nu} \neq 0$, the existence of strongly singular spacetimes is a little more subtle, but we can for example refer the reader to Section~\ref{examples.just} to justify that there exists at least a one-parameter family of strongly singular examples. \blue{We admit, however, that understanding conditions, purely on the initial data, which guarantee \eqref{eq:QTS} and \eqref{eq:SKE}, is a highly nontrivial problem, see the discussion in Section~\ref{intro.stab}. One way to resolve this}, anticipating already the stability result Theorem~\ref{roughthm:stab}, \blue{is by considering} small \blue{but \emph{charged}} \emph{perturbations} of any strongly singular spacetime obeying \eqref{eq:SKE}; \blue{note the \emph{unperturbed spacetime} may have $F_{\mu\nu} = 0$.}

We now state Theorem~\ref{roughthm:asymp}, which is a rough version of our main Theorems~\ref{thm:esfss} and \ref{thm:emsfss}.

\begin{customthm}{I}[Asymptotics -- rough version] \label{roughthm:asymp}
    Consider spherically symmetric characteristic initial data for the Einstein-Maxwell-scalar field system \eqref{eq:einstein}, \eqref{eq:scalar_wave}, \eqref{eq:maxwell} on the bifurcate null hypersurface $C_0 \cup \underline{C}_0$, and suppose further that the maximal future development \blue{$\mathcal{D}$} is strongly singular, as described above.

    If either of the following holds:
    \begin{enumerate}[(A)]
        \itemsep -0.2em
        \item
            the Maxwell field vanishes, i.e.~$F_{\mu\nu} = 0$, or
        \item
            the Maxwell field is non-trivial, i.e.~$F_{\mu\nu} \neq 0$, but we assume a priori that \eqref{eq:QTS} and \eqref{eq:SKE} hold,
    \end{enumerate}
    then the dynamical variables $r(u, v)$, $\Omega^2(u, v)$ and $\phi(u, v)$ obey the following asymptotics in $\mathcal{D}$, for some $0 < \beta < 1$.

    \begin{enumerate}[(I)]
        \item \label{roughI}
            \underline{Asymptotics for $r$:} 
            $r^2(u, v)$ is a $C^{1, \beta}$ function up to and including the singular boundary $\mathcal{S}$. Its first derivatives $-r \partial_u r(u, v), - r \partial_v r(u, v)$ are strictly positive $C^{0, \beta}$ functions up to $\mathcal{S} = \{ (u, v): r(u, v) = 0 \}$, which is a $C^{1, \beta}$ curve in the $(u, v)$--plane.
        \item \label{roughII}
            \underline{Asymptotics for $\phi$:}
            There exist functions $\Psi(u, v)$ and $\Xi(u, v)$, bounded and H\"older continuous up to and including the singular boundary $\mathcal{S}$, such that $\phi(u, v)$ can be written as
            \begin{equation*}
                \phi(u, v) = \Psi(u, v) \cdot \log \left(\frac{r_0}{r(u, v)}\right) + \Xi(u, v).
            \end{equation*}
            (Here $r_0$ is chosen to be the value of $r(u, v)$ at the initial data bifurcation sphere $C_0 \cap \underline{C}_0$ -- however we note that in reality, $r_0$ is simply a scaling parameter and can be chosen to be any fixed constant.) 
        \item \label{roughIII}
            \underline{Asymptotics for $\Omega^2$:} 
            There exists a function $\mathfrak{M}(u, v)$, with $\log \mathfrak{M}$ bounded and H\"older continuous up to and including the singular boundary $\mathcal{S}$, such that 
            \begin{equation*}
                4 \Omega^{-2} \partial_u r \partial_v r (u, v) = \mathfrak{M} \cdot \left( \frac{r}{r_0} \right)^{- (\Psi^2 + 1)}.
            \end{equation*}
        \item \label{roughIV}
            \ul{Estimates localized at a TIP\footnote{%
                    TIP stands for \textit{terminal indecomposable past}, see Chapter 6 of \cite{hawking_ellis_1973}. Strictly speaking, our use of the term TIP is an abuse of terminology as every $p \in \mathcal{S}$ actually represents a $2$-sphere of TIPs.
            }:}
            Let $p = (u_p, v_p) \in \mathcal{S}$. By (\ref{roughI})--(\ref{roughIII}), the quantities $- r \partial_u r$, $- r \partial_v r$, $\Psi$, $\Xi$ and $\mathfrak{M}$ may all be defined at $p$, and we denote their values at $p$ as $C_u$, $C_v$, $\Psi_{\infty}$, $\Xi_{\infty}$ and $\mathfrak{M}_{\infty}$ respectively.

            For such $p$, consider its causal past $J^-(p) = \{ (u, v) \in \mathcal{D}: u \leq u_p, v \leq v_p \}$. Then for $(u, v) \in J^-(p)$, one has, for $\alpha > 0$ (coinciding with the $\alpha$ appearing in \eqref{eq:SKE} when $F_{\mu\nu} \not\equiv 0$, and $\alpha = 1$ otherwise):
            \begin{gather*}
                - r \partial_u r (u, v) = C_u + O(r^{\alpha}), \quad - r \partial_v r (u, v) = C_v + O(r^{\alpha}), \\[0.5em]
                r^2 \partial_u \phi (u, v) = C_u \cdot \Psi_{\infty} + O(r^{\alpha}), \quad r^2 \partial_v \phi (u, v) = C_v \cdot \Psi_{\infty} + O(r^{\alpha}), \\[0.5em]
                \phi(u, v) = \Psi_{\infty} \cdot \log \left( \frac{r_0}{r} \right) + \Xi_{\infty} + O(r^{\alpha} \log r^{-1}), \\[0.5em]
                \Omega^2(u, v) = \frac{4 C_u C_v}{r_0^2} \cdot \mathfrak{M}_{\infty}^{-1} \cdot \left( \frac{r}{r_0} \right)^{\Psi_{\infty}^2 - 1} + O(r^{\Psi_{\infty}^2 - 1 + \alpha} \log r^{-1}).
            \end{gather*}
        \item \label{roughV}
            \underline{Precise blow-up estimates:}
            For $p\in \mathcal{S}$ and $J^-(p)$ as in (\ref{roughIV}), the Hawking mass, the Kretschmann scalar and the scalar field blow up at the following rates:
            \begin{gather*}
                2m(u, v) = \mathfrak{M}_{\infty} r_0 \cdot \left( \frac{r_0}{r} \right)^{\Psi_{\infty}^2} + O(r^{- \Psi^2 + \alpha} \log r^{-1}), \\[0.5em]
                \mathrm{Riem}_{\alpha \beta \gamma \delta} \mathrm{Riem}^{\alpha \beta \gamma \delta} = \mathfrak{M}_{\infty}^2 \cdot \left( \frac{r_0}{r} \right)^{2 (\Psi_{\infty}^2 + 3)} \cdot \frac{ 4 ( 3 - 2 \Psi_{\infty}^2 + 2 \Psi_{\infty}^4)}{r_0^4} + O(r^{-2 \Psi_{\infty}^2 - 6 + \alpha} \log r^{-1} ), \\[0.5em]
                \nabla_{\alpha} \phi \nabla^{\alpha} \phi = - \frac{\mathfrak{M}_{\infty} \Psi_{\infty}^2}{r_0^2} \cdot \left( \frac{r_0}{r} \right)^{\Psi_{\infty}^2 + 3} + O(r^{- \Psi_{\infty}^2 - 3 + \alpha} \log r^{-1}).
            \end{gather*}
    \end{enumerate}
\end{customthm}

We refer \blue{the reader} to Theorem~\ref{thm:esfss} for the precise version of Theorem~\ref{roughthm:asymp} when $F_{\mu\nu} \equiv 0$, and Theorem~\ref{thm:emsfss} for the precise statement when $F_{\mu\nu} \neq 0$. The various H\"older norms involved will be defined in Section~\ref{statement.norms}. We now make several brief comments about Theorem~\ref{roughthm:asymp}.

\begin{enumerate}
    \item
        Theorem~\ref{roughthm:asymp} gives the leading order term of an asymptotic expansion for $\phi$, $\Omega^2$, and other related quantities in terms of $r$. It would be interesting to compute the next order terms in such expansions, as in the heuristic work \cite{BuonannoDamourVeneziano}, though we do not pursue that here.
    \item In light of (\ref{roughIII}), it is tempting to introduce the following (past-directed) time function:
        \begin{equation} \label{eq:proper_time}
            \tau = \mathfrak{M}^{-1/2} \cdot \frac{2 r_0}{3 + \Psi^2} \cdot \left( \frac{r}{r_0} \right)^{\frac{3 + \Psi^2}{3}}.
        \end{equation}
        \blue{This} $\tau$ is chosen such that $g(\nabla \tau, \nabla \tau) = - 1 + O(r^{\alpha} \log r^{-1})$, and \blue{is our} candidate time function for an asymptotically Gaussian or asymptotically CMC \blue{foliation} of our spacetime, see Corollary~\ref{cor:bkl}.
    \item
        \blue{Using a foliation by constant $\tau$ slices}, we associate the following \blue{three} Kasner exponents to each $p \in \mathcal{S}$:
        \begin{equation} \label{eq:Psikasner}
            p_1(p) = \frac{\blue{\Psi^2(p)} - 1}{\blue{\Psi^2(p)} + 3}, \qquad p_2(p) = p_3(p) = \frac{2}{\blue{\Psi^2(p)} + 3}.
        \end{equation}
        The fact that the exponents are permitted to vary along $\mathcal{S}$ will be justified in Section~\ref{intro.bkl}.

\end{enumerate}

\blue{
\noindent
Stemming from the latter two comments, we state a corollary claiming that the spacetimes of Theorem~\ref{roughthm:asymp} are \emph{Kasner-like}, see \cite{RingstromSilentGeometry} for further discussion of the expansion-normalized Weingarten map $\mathcal{K}_i^{\phantom{i}j}$ and its relevance to Kasner exponents, as well as Corollary~\ref{cor:bkl} for the precise version of the following.
\begin{customcorollary}{II}[Kasner-like behaviour, rough version] 
    Let $(\mathcal{M}, g)$ be a strongly singular, spherically symmetric spacetime as described by Theorem~\ref{roughthm:asymp}. Then, for $r$ sufficiently small, $\tau$ defined in \eqref{eq:proper_time} is timelike, and one may define $k_{ij}$ to be the second fundamental form of the foliation by constant $\tau$-hypersurfaces.

    Then, for $(u, v) \in J^-(p)$ and $r(u, v)$ sufficiently small, the expansion-normalized Weingarten map $\mathcal{K}_i^{\phantom{i}j} \coloneqq (\tr k)^{-1} \, k^{i}_{\phantom{i}j}$  is such that the eigenvalues of $\mathcal{K}^i_{\phantom{i}j}$ are given by the $p_i(p)$ in \eqref{eq:Psikasner}, plus an error of order $O(r^{\alpha} \log r^{-1})$.
\end{customcorollary}
}

\subsection{The BKL ansatz} \label{intro.bkl}

We now briefly describe the ansatz for the near-singularity solutions of Einstein's equations, proposed by Khalatnikov and Lifshitz in \cite{kl63} and developed in the celebrated work of BKL \cite{bkl71, bkl82}. Starting with vacuum, \cite{kl63} considers the following $1+3$-decomposition of spacetime:
\begin{equation} \label{eq:adm}
    \mathcal{M} = (0, T) \times \Sigma, \qquad g = - dt^2 + {}^{(3)}h(t).
\end{equation}
Here ${}^{(3)}h(t)$ is a time-dependent smooth Riemannian metric on the $3$-manifold $\Sigma$.

Influenced by the Kasner metric (\ref{eq:kasner}), \cite{kl63} proposes the following ansatz for $g$ in the neighborhood of a singularity at $\{ t = 0 \}$: for $(t, x) \in \mathcal{M}$ they write
\begin{equation} \label{eq:bkl}
    g(t, x) \approx - dt^2 + \sum_{I = 1}^3 t^{2 p_I(x)} \omega^{I}(x) \otimes \omega^{I}(x).
\end{equation}
In \eqref{eq:bkl}, $\omega^I(x)$ are smooth $1$-forms on $\Sigma$ spanning the cotangent space $T_x^*\Sigma$, while the exponents $p_I(x)$ are smooth functions on $\Sigma$ satisfying the two \emph{Kasner relations}:
\begin{gather} \label{eq:kasnerrelation1}
    \sum_{I = 1}^3 p_I(x) = 1, \qquad
    \sum_{I = 1}^3 p_I^2(x) = 1.
\end{gather}
The time coordinate $t$ is chosen such that the singularity is synchronized at $t = 0$.

Assume for now that $p_1(x) < p_2(x) < p_3(x)$. We define $e_I(x)$ to be the vector fields dual to $\omega^I(x)$, and use the following notational convention: for any $(p, q)$-tensor field $T$ defined on $\Sigma$ we use lower case indices $T_{i_1 \cdots i_p}^{\phantom{i_1 \cdots i_p}j_1 \cdots j_q}$ to denote the tensor field evaluated using the coordinate vector fields $\frac{\partial}{\partial x^i}$ and their respective coordinate $1$-forms $dx^i$, while upper case indices $T_{I_1 \cdots I_p}^{\phantom{I_1 \cdots I_p} J_1 \cdots J_q}$ denote the tensor field evaluated using the frame fields $e_I(x), \omega^I(x)$. Indices are raised and lowered with respect to the Riemannian metric ${}^{(3)}h = {}^{(3)}h(t)$.

In the gauge (\ref{eq:adm}), the evolutionary components of Einstein's equation (\ref{eq:einstein}) yield the following:
\begin{gather}
    \frac{\partial}{\partial t} {}^{(3)}h_{ij} = - 2 k_{ij}, \label{eq:adm1}\\[1em]
    \frac{1}{\sqrt{\det ({}^{(3)}h_{pq})}} \frac{\partial}{\partial t} \left( \sqrt{\det({}^{(3)} h_{pq})} k_i^{\phantom{i}j}\right)  = {}^{(3)} \mathrm{Ric}[{}^{(3)} h]_i^{\phantom{i}j} - 2 T_i^{\phantom{i}j} + \delta_i^{\phantom{i}j} \left( T_0^{\phantom{0}0} + T_k^{\phantom{k}k} \right). \label{eq:adm2}
\end{gather}
The metric on the right hand side of (\ref{eq:bkl}) solves (\ref{eq:adm1}) and (\ref{eq:adm2}), if one assumes the right hand side of (\ref{eq:adm2}) to be absent. Hence a necessary condition for (\ref{eq:bkl}) to be self-consistent in vacuum (with $T_{\mu\nu} = 0$) is that the expression ${}^{(3)} \mathrm{Ric} [{}^{(3)}h]$ is suitably small. In \cite{kl63}, the authors argue that the required smallness is that for some $\epsilon > 0$, 
\begin{equation} \label{eq:bkl_smallness}
    |{}^{(3)}\mathrm{Ric}[{}^{(3)} h]_{IJ}| = O(t^{p_I + p_J - 2 + \epsilon}).
\end{equation}

On the other hand, a computation of spatial Ricci curvature with respect to the (approximate) orthonormal frame $t^{- p_I(x)} e_I(x)$ allows one to obtain the following asymptotics for ${}^{(3)}\mathrm{Ric}[{}^{(3)} h]_{IJ}$, where $\lambda_{IJ}^K$ are the structure coefficients of the frame, $d \omega^K = \lambda_{IJ}^K \, \omega^I \wedge \omega^J, \;\lambda^{K}_{IJ} = - \lambda^K_{JI}$.
\begin{equation} \label{eq:bkl_ricci}
    \left| t^{- p_I - p_J + 2}\cdot {}^{(3)}\mathrm{Ric}[{}^{(3)} h]_{IJ} \right| = 
    (\lambda^1_{23})^2 \cdot t^{2(1 + p_1 - p_2 - p_3)} \cdot \delta_{IJ} + O(t^{\epsilon})
\end{equation}

Since $1 + p_1 - p_2 - p_3 = 2 p_1 < 0$ (see for instance \eqref{eq:kasnercircle} with $u \geq 1$), ${}^{(3)}\mathrm{Ric}[{}^{(3)}h]$ obeys (\ref{eq:bkl_smallness}) if and only if the structure coefficient $\lambda^1_{23}(x)$ vanishes identically on $\Sigma$. \cite{kl63} then argues that this prohibits the ansatz \eqref{eq:bkl} from having the correct number of physical degrees of freedom. Note that $\lambda^1_{23}(x) = 0$ holds if and only if the one-form $\omega^1$ is hypersurface orthogonal, and in particular holds in spacetimes with certain symmetries, including spherical symmetry.

In \cite{bkl71}, BKL further provide heuristics arguing that if the term (\ref{eq:bkl_ricci}) is present, then there exists some critical time $t_c$, such that the metric will dynamically ``bounce'' from the approximate form (\ref{eq:bkl}) for times $t \gg t_c$, to another approximate form (\ref{eq:bkl}) for times $t \ll t_c$, but with modified Kasner exponents
\begin{equation} \label{eq:bounce}
    \tilde{p}_1(x) = - \frac{p_1(x)}{1 + 2 p_1(x)}, \quad \tilde{p}_2(x) = \frac{p_2(x) + 2 p_1(x)}{1 + 2 p_1(x)}, \quad \tilde{p}_3(x) = \frac{p_3(x) + 2 p_1(x)}{1 + 2 p_1(x)}.
\end{equation}

The new exponents $\tilde{p}_I$ still obey \eqref{eq:kasnerrelation1}, hence this new Kasner regime is subject to a similar inconsistency. \cite{bkl71, bkl82} then propose that the generic eventual approach to a spacelike singularity involves an infinite cascade of such \textit{Kasner bounces}, and is known in the literature as the \textit{BKL oscillatory approach to singularity}.

Rigorous work concerning the proposed oscillatory dynamics will, in practise, involve controlling wildly \blue{behaved} error terms. It is thus desirable to find regimes, possibly involving matter, where the instability mechanism can be suppressed. One such regime appears in \cite{bk72}, where the authors find that matter described by a massless scalar field indeed has this property.

In \cite{bk72}, the authors propose the same near-singularity ansatz (\ref{eq:bkl}) for the metric, along with the following ansatz for the scalar field $\phi$:
\begin{equation} \label{eq:bkl_scalar}
    \phi(t, x) \approx A(x) \log t + B(x).
\end{equation}
In order for (\ref{eq:bkl}) and (\ref{eq:bkl_scalar}) to solve the Einstein-scalar field equations asymptotically, the Kasner relations (\ref{eq:kasnerrelation1}) are modified to
\begin{equation} \label{eq:kasner2_scalar}
    \sum_{I=1}^3 p_I(x) = 1, \qquad \sum_{I= 1}^3 p_I(x)^2 + 2 A(x)^2 = 1.
\end{equation}

The key observation made in \cite{bk72} is that the modified Kasner relations (\ref{eq:kasner2_scalar}) allow for exponents with $0 < p_1(x) < p_2(x) < p_3(x)$. In such regimes, $1 + p_1 - p_2 - p_3 = 2 p_1$ is now positive, and the previously problematic term (\ref{eq:bkl_ricci}) can be treated as an error.

It can also be shown that the energy-momentum tensor terms $T_i^{\phantom{i}j}$ and $T_0^{\phantom{0}0}$ appearing on the right hand side of (\ref{eq:adm2}) combine to produce an error which can similarly be treated as negligible. The conclusion is that (\ref{eq:bkl}) and (\ref{eq:bkl_scalar}) produce self-consistent asymptotics near a $t = 0$ singularity, so long as the exponents lie in the so-called \textit{subcritical regime} where $\min p_I(x) > 0$. In the context of Theorem~\ref{roughthm:asymp} and \eqref{eq:Psikasner}, this corresponds to $\inf \Psi^2(u, v) > 1$.

As remarked earlier, within spherical symmetry, the subcriticality assumption is not necessary with or without a scalar field. In order to observe such instabilities, and the necessity of studying such subcritical regimes, we add another matter source, namely an electromagnetic field.

\subsection{The role of electromagnetism} \label{intro.em}


Recall once again that an electromagnetic field is described by a $1$-form $A_{\mu}$ and its corresponding Maxwell field $F = dA$. 
Being a $2$-form, the Maxwell field $F_{\mu\nu}$ can be considered as an anisotropic matter source -- unlike the scalar field $\phi$ -- and could potentially mimic the effects of gravitational perturbations outside of symmetry. So one could ask: how exactly does this arise in our setting?

In the context of the BKL ansatz (\ref{eq:bkl}), electromagnetic matter was studied by Belinski and Khalatnikov \cite{bk77} (see also Chapter 4 of \cite{BelinskiHenneaux}), where alongside the ansatz 
(\ref{eq:bkl}) for the metric, the authors suggest the following asymptotics for the electromagnetic field:
\begin{equation} \label{eq:bkl_em}
    E^I(x, t) \approx \frac{1}{t} \Phi_{E}^{I} (x), \qquad B^I(x, t) \approx \frac{1}{t} \Phi_{B}^I (x).
\end{equation}
Note that here we define
$E_I(x, t) \coloneqq F \left( {\partial_t}, e_I \right)$ and $B_I(x, t) \coloneqq *F \left( \partial_t , e_I \right)$
to be the usual $1+3$ electric-magnetic decomposition of $F_{\mu\nu}$, evaluated with respect to the spatial frame $\{ e_I \}$.

However, even \eqref{eq:bkl_em} turns out to be inconsistent outside the subcritical regime $\min \{p_I(x)\} > 0$. Assuming $p_1(x) < p_2(x) < p_3(x)$, the $e_1(x)$ projections of Maxwell's equations \eqref{eq:maxwell} in the BKL ansatz yield:
\begin{equation} \label{eq:bkl_e1b1}
    \frac{\partial}{\partial t} (t E^1(x)) \approx - (\lambda^1_{23})^2 t^{2 p_1- 1} \cdot t  B^1(x), \quad
    \frac{\partial}{\partial t} (t B^1(x)) \approx + (\lambda^1_{23})^2 t^{2 p_1 - 1}\cdot t E^1(x).
\end{equation}
Hence if $p_1 < 0$ the right-hand side would appear to be non-integrable.

Nonetheless, the system \eqref{eq:bkl_e1b1} does suggest that the expression
\begin{equation} \label{eq:bkl_omega}
    \omega^2(t, x) \coloneqq t^2 [ E^1(x)^2 + B^1(x)^2 ] 
\end{equation}
is (approximately) constant in time, and, abusing notation slightly, we set $\omega^2(x) = \omega^2(t, x)$ in the following discussion. 

As well as ${}^{(3)}\mathrm{Ric}[{}^{(3)}h]$, we must now also estimate the matter term on the right hand side of (\ref{eq:adm2}). Writing $T_{\mu\nu}^{EM}$ in terms of $E^i$ and $B^i$, one finds the following expression:
\begin{equation*}
    - 2 T_{ij} + {}^{(3)}h_{ij} (T_{0}^{\phantom{0}0} + T_k^{\phantom{k}k}) = 2 (E_i E_j + B_i B_j) - {}^{(3)}h_{ij} (E_k E^k + B_k B^k),
\end{equation*}
which in light of (\ref{eq:bkl_omega}) -- the contributions of $E_I$ and $B_I$ for $I = 2, 3$ are part of the error -- yields
\begin{equation} \label{eq:bkl_emem}
    \left|t^{- p_I - p_J + 2}\cdot [- 2 T_{IJ} + {}^{(3)}h_{IJ} (T_0^{\phantom{0}0} + T_k^{\phantom{k}k} )] \right| = 
    \omega^2 \cdot t^{2 p_1} \delta_{IJ} + O(t^{\epsilon}).
\end{equation}

The term \eqref{eq:bkl_emem} with $\omega \neq 0$ sources a similar instability to that of \eqref{eq:bkl_ricci} with $\lambda^1_{23} \neq 0$. Perhaps more remarkably, the analysis of \cite{bk77} suggests that as long as either $\omega$ or $\lambda^1_{23}$ fails to vanish, we obtain a Kasner bounce with the transition map \eqref{eq:bounce}. This concludes our brief explanation of why the electromagnetic field acts as a proxy for spatial curvature in the BKL setting.

To see how these ideas are manifested in the spherically symmetric setting, we note that in a regular double null gauge $(u, v)$, one of the components of Einstein's equation (see already (\ref{eq:wave_r_u})) is
\begin{equation} \label{eq:wave_r_u_rough}
    \partial_u ( - r \partial_v r ) \approx -\frac{Q^2 \Omega^2}{4 r^2},
\end{equation}
at least in the vicinity of $\{ r = 0 \}$, where $Q$ is a constant quantifying the strength of the Maxwell field. 

From (\ref{roughI}) of Theorem~\ref{roughthm:asymp}, we expect that $- r \partial_u r$ and $- r \partial_v r$ are positive and bounded away from $0$. Further, recalling (\ref{eq:omega_example}), we integrate (\ref{eq:wave_r_u_rough}) towards $r = 0$ and change variables, finding
\begin{align*}
    - r \partial_v r 
    &\approx \text{ data } - \int \frac{Q^2 \Omega^2}{4 r^2} \, du, \\[0.5em]
    &\approx \text{ data } - \int Q^2 \cdot C(u) \cdot r^{\Psi^2 - 3} \, du \\[0.5em]
    &\approx \text{ data } - \int \frac{Q^2 \cdot C(u)}{ - r \partial_u r} \cdot r^{\Psi^2 - 2} \, dr.
\end{align*}

Now, the expression preceding $r^{\Psi^2 - 2}$ in the integrand is bounded and tends to a positive constant as $r \to 0$, so it suffices to consider the integrability of $r^{\Psi^2 - 2}$. The key observation is that while $r^{\Psi^2 - 2}$ is indeed integrable near $r = 0$ for $\Psi^2 > 1$, this fails to be true when instead $\Psi^2 \leq 1$. Furthermore, this computation would suggest that in the latter case, $- r \partial_v r$ would eventually become positive, \blue{which contradicts} the fact that we lie in a trapped region!

Therefore, when $Q \neq 0$, spacelike singularity formation is only expected when the quantity $\Psi$ obeys $\Psi^2 > 1$. By \eqref{eq:Psikasner}, this indeed corresponds to the subcritical regime of Kasner exponents as expected. In the present article, we impose $\Psi^2 > 1$ whenever $Q \neq 0$ via the additional assumption \eqref{eq:SKE}.  

\subsection{Ideas of the proof} \label{intro.proof}

We now briefly describe several of the key steps in the proof of Theorem~\ref{roughthm:asymp}. For simplicity, we focus on a single point $p \in \mathcal{S}$ and its causal past $J^-(p)$. The essential claims will be that for any $(u, v) \in J^-(p)$, there exists some $\alpha > 0$ such that:
\begin{enumerate}[(I)]
    \item \label{roughproofI}
        There exist positive constants $C_u, C_v > 0$ such that in $J^-(p)$, one has
        \begin{equation*}
            - r \partial_u r = C_u + O(r^{\alpha}), \qquad - r \partial_v r = C_v + O(r^{\alpha}).
        \end{equation*}
    \item \label{roughproofII}
        There exist constants $\Psi_{\infty}, \Xi_{\infty} \in \R$ such that in $J^-(p)$, one has
        \begin{equation*}
            r^2 \partial_u \phi = C_u \cdot \Psi_{\infty} + O(r^\alpha), \qquad 
            r^2 \partial_v \phi = C_v \cdot \Psi_{\infty} + O(r^\alpha),
        \end{equation*}
        \begin{equation*}
            \phi = \Psi_{\infty} \log r^{-1} + \Xi_{\infty} + O(r^{\alpha}).
        \end{equation*}
    \item \label{roughproofIII}
        There exists some $C_{\Omega} > 0$ such that in $J^-(p)$, one has
        \begin{equation*}
            \Omega^2 = C_{\Omega} \cdot r^{\Psi_{\infty}^2 - 1} + O(r^{\Psi_{\infty}^2 - 1 + \alpha}).
        \end{equation*}
\end{enumerate}

In the full proof, each of these must be shown simultaneously, but for the purpose of illustrating the ideas we simply sketch a proof of each of (\ref{roughproofI})--(\ref{roughproofIII}) assuming the others. 

We start with (\ref{roughproofIII}), which is proven in Section~\ref{bklasymp.geometry}, particularly Proposition~\ref{prop:bklasymp_lapse}. The idea is to rewrite the Raychaudhuri equation (see already \eqref{eq:raych_v}) as:
\begin{equation*}
    \partial_v \log \left( \frac{\Omega^2}{- \partial_v r} \right) = \frac{\partial_v r}{r} \cdot \left( \frac{r^2 \partial_v \phi}{- r \partial_v r} \right)^2.
\end{equation*}
Assuming (\ref{roughproofI}) and (\ref{roughproofII}), the rightmost expression can be written as $\Psi_{\infty}^2 + O(r^{\alpha})$. Therefore, subtracting $\Psi_{\infty}^2 \cdot \partial_v (\log r) = \Psi_{\infty}^2 \cdot \frac{\partial_v r}{r}$ from both sides, one finds
\begin{equation*}
    \partial_v \log \left(\frac{\Omega^2}{- \partial_v r} \cdot r^{- \Psi_{\infty}^2 }\right) = \partial_v r \cdot O(r^{\alpha-1}).
\end{equation*}
The right hand side is now integrable towards $r = 0$, hence $\frac{\Omega^2}{ - \partial_v r}\cdot r^{- \Psi_{\infty}^2}$ tends to some positive constant as $r \to 0$. Combining again with (\ref{roughproofI}), and performing similar integrations in the $u$-direction, we indeed find (\ref{roughproofIII}).

Next, we move onto (\ref{roughproofI}), see already \eqref{eq:rurll} in Section~\ref{bklasymp.proof} (note though that in the course of our actual proof we largely circumvent the need for (\ref{roughproofI}) by using instead the gauge-invariant derivatives \eqref{eq:llbar}). For this, we use the following wave-like evolution equation for $r^2(u, v)$:
\begin{equation*}
    \partial_u ( - r \partial_v r) = \partial_v ( - r \partial_u r) = \frac{\Omega^2}{4} \left( 1 - \frac{Q^2}{r^2} \right).
\end{equation*}
But by (\ref{roughproofIII}), we may bound $\Omega^2$ by a power of $r$, giving
\begin{equation*}
    \partial_u ( - r \partial_v r) = \begin{cases}
        O(r^{\Psi_{\infty}^2 - 1}) & \text{ if } Q \equiv 0,\\
        O(r^{\Psi_{\infty}^2 - 3}) & \text{ if } Q \not \equiv 0.
    \end{cases}
\end{equation*}
To integrate this, we change variables from $u$ to $r$. This is where we use \eqref{eq:QTS}, which tells us that $du = O(r) dr$. Furthermore, when $Q \neq 0$, we require $\Psi_{\infty}^2 > 1$ in order for the integral of $O(r^{\Psi^2 - 3})$ in $u$ not to diverge as $r \to 0$ -- justifying why we need the assumption \eqref{eq:SKE}. If either $Q \equiv 0$, or $Q \neq 0$ but \eqref{eq:QTS} and \eqref{eq:SKE} hold, then we get (\ref{roughproofI}) as claimed.

We remark here that the $O(r^{\alpha})$ error term in (\ref{roughproofI}) is not optimal, and that the true first-order correction will be $O(r^\gamma)$ with $\gamma$ depending on $\Psi_{\infty}^2$. However, for our goal of deriving the leading order asymptotics, our $O(r^{\alpha})$ errors are sufficient.

Finally, we derive (\ref{roughproofII}), giving asymptotics for the main dynamical quantity $\phi(u, v)$. This is the main objective of Section~\ref{scalarfield}, particularly Proposition~\ref{prop:scalarfield_asymptotics}. There, we use the wave equation in the following form:
\begin{equation*}
    \partial_u ( r^2 \partial_v \phi ) = - r \partial_v r \cdot \partial_u \phi + r \partial_u r \cdot \partial_v \phi = \tilde{X} \phi,
\end{equation*}
where $\tilde{X}$ is defined as $\tilde{X} = - r \partial_v r \partial_u + r \partial_u r \partial_v$. 
The procedure will be to first estimate $|\tilde{X} \phi|$ and then integrate this equation towards $r = 0$. 

The fact $\tilde{X} \phi$ appears is crucial. If the vector field\footnote{The vector field $\tilde{X}$ is also a rescaled version of the Kodama vector field \cite{Kodama}, most often used to find energy estimates in the exterior region, where $r$ is instead a spacelike variable and $\tilde{X}$ is timelike.} $\tilde{X}$ were to be replaced by, say, the null vector field $\partial_v$, then anticipating (\ref{roughproofI}) the modified right-hand side would obey $|\partial_v \phi| = O(r^{-2})$.
However, $r^{-2}$ is \underline{not} integrable in $u$ towards $r = 0$, recalling that $du = O(r) dr$ but one still sees the logarithmic divergence $\int r^{-2} du = \int O(r^{-1}) dr = O(\log r^{-1})$. Fortunately, $\tilde{X}$ will be a ``better'' derivative than $\partial_v$. 

The reason for this is that since $\tilde{X} r = 0$, $\tilde{X}$ is a ``spatial'' derivative tangent to the spacelike hypersurfaces of constant $r$. In line with the \textit{asymptotically velocity dominated} expectation in relativistic cosmology that spatial derivatives better than timelike derivatives near $r = 0$, it is unsurprising that $\tilde{X} \phi$ is better behaved than, say, $\partial_v \phi$ or $\partial_u \phi$.

Indeed, we shall obtain $|\tilde{X} \phi| \lesssim r^{-2 + \alpha}$, the proof of which is the main focus of Section~\ref{scalarfield.xphi}, particularly Proposition~\ref{prop:scalarfield_xphi}. Using this, we now integrate the above equation and find that $r^2 \partial_v \phi = C_v \cdot \Psi_{\infty} + O(r^{\alpha})$, for some $\Psi_{\infty} \in \R$. The remaining claims of (\ref{roughproofII}) are then straightforward -- noting that we use $|\tilde{X} \phi| \lesssim r^{-2 + \alpha}$ again to show that the $\Psi_{\infty}$ arising in the leading order terms of $r^2\partial_v \phi$ and $r^2 \partial_u \phi$ coincide.

This concludes our sketch of the proof of Theorem~\ref{roughthm:asymp}. In the remainder of the introduction, we discuss potential extensions of our Theorem~\ref{roughthm:asymp}, and consider further its place in the existing mathematical literature.

\subsection{Stability and instability of \texorpdfstring{$\mathcal{S}$}{S}} \label{intro.stab}

Theorem~\ref{roughthm:asymp} provides leading order asymptotics for all geometric and matter field quantities near a spacelike singularity $\mathcal{S} = \{ r = 0 \}$ \underline{provided we already know it exists}. 
Returning to our initial discussion regarding the generality of such singularities, it is prudent to ask to what extent our results model singularities arising from a general class of initial data.

For instance, one may consider stability properties of the strongly singular spacetimes described in Theorem~\ref{roughthm:asymp}. Even if we only consider perturbations in spherical symmetry, this is a subtle problem, since the whereabouts of the eventual singular boundary is only determined a posteriori, and one must track the location of the new singularity $\hat{\mathcal{S}}$ arising from perturbed data, in any choice of coordinate gauge.

In particular, to make quantitative comparisons between the perturbed spacetime and the background spacetime, we require a further gauge transformation $(u, v) \mapsto (\hat{u}(u), \hat{v}(v))$ such that the position of $\hat{\mathcal{S}}$ in the $(\hat{u}, \hat{v})$-plane coincides with that of $\mathcal{S}$ in the $(u, v)$-plane. We state the following theorem, to be proved in the follow-up article \cite{me_kasner_stab}.

\begin{customthm}{III}[Stability -- rough version] \label{roughthm:stab}
    Let $(r, \Omega^2, \phi, Q)$ be a spherically symmetric and strongly singular solution to the Einstein-Maxwell-scalar field system \eqref{eq:einstein}, \eqref{eq:scalar_wave}, \eqref{eq:maxwell}, as described in Theorem~\ref{roughthm:asymp}, with characteristic initial data $(r_i, \Omega^2_i, \phi_i, Q_i)$ given on $C_0 \cup \underline{C}_0$. The domain of existence $\mathcal{D}$ and spacelike singularity $\mathcal{S}$ are given by the following subsets in the $(u, v)$-plane:
    \begin{gather*}
        \mathcal{D} = \{ (u, v) \in [u_0, u_1] \times [v_0, v_1] \subset \R^2 : v < v_*(u) \text{ when } v_*^{-1}(v_1) \leq u \leq u_1\}, \\[0.5em]
        \mathcal{S} = \{ (u, v_*(u)): v_*^{-1}(v_1) \leq u \leq u_1 \}.
    \end{gather*}
    By Theorem~\ref{roughthm:asymp}, $v_*(u)$ is a $C^{1, \beta}$ curve in the $(u, v)$-plane, see Section~\ref{setup.data} and Figure~\ref{fig:char_ivp1}.

    Consider perturbed characteristic initial data on $C_0 \cup \underline{C}_0$ given by $(\hat{r}_i, \hat{\Omega}^2_i, \hat{\phi}_i, \hat{Q}_i)$, with the perturbation being of size $\epsilon$ in a suitable norm, and suppose further that one of the following holds:
    \begin{enumerate}[(A)]
        \item $Q_i \equiv \hat{Q}_i \equiv 0$, or
        \item the subcriticality assumption (\ref{eq:SKE}) holds\footnote{By a Cauchy stability argument, it is only necessary for \eqref{eq:SKE} to hold near $\mathcal{S}$, which translates to $\inf_{p \in \mathcal{S}} \Psi_{\infty}^2(p) > 1 + \alpha$.} for the background solution $(r, \Omega^2, \phi, Q)$.
    \end{enumerate}
    Then for $\epsilon$ sufficiently small, the maximal future development $(\hat{r}, \hat{\Omega}^2, \hat{\phi}, \hat{Q})$ of the perturbed data, defined in a new domain of development $\hat{\mathcal{D}}$, is also strongly singular with spacelike singularity $\hat{\mathcal{S}}$ and satisfies the assumptions of Theorem~\ref{roughthm:asymp}. 

    Furthermore, after a gauge transformation $u \mapsto \hat{u}(u), v \mapsto \hat{v}(v)$, with $\| \hat{u}(u) - u \|_{C^{1, \beta}}, \| \hat{v}(v) - v \|_{C^{1, \beta}} \lesssim \epsilon$, the locations of $\hat{\mathcal{D}}$ and $\hat{\mathcal{S}}$ in the $(\hat{u}, \hat{v})$-plane coincide with that of $\mathcal{D}$ and $\mathcal{S}$ in the $(u, v)$-plane. 
    Abusing notation somewhat, for $(u, v) \in \mathcal{D}$ as above we define $(f - \hat{f})\, (u, v)$ to be $f(u, v) - \hat{f}(\hat{u}^{-1}(u), \hat{v}^{-1}(v))$. If $\hat{\Psi}$, $\hat{\Xi}$ and $\hat{\mathfrak{M}}$ are the quantities arising in (\ref{roughI})--(\ref{roughIII}) of Theorem~\ref{roughthm:asymp} corresponding to the perturbed spacetime, one finally has the quantitative stability estimate
    \begin{equation*}
        \| r^2 - \hat{r}^2 \|_{C^{1, \beta}} + \| \Psi - \hat{\Psi} \|_{C^0} 
        + \| \Xi - \hat{\Xi} \|_{C^0} + \| \mathfrak{M} - \hat{\mathfrak{M}} \|_{C^0} \lesssim \epsilon.
    \end{equation*}
\end{customthm}

This stability theorem is consistent with the expectations of BKL, since either of (A) or (B) suggests that the right hand side of \eqref{eq:adm2} is suitably small, see again Sections~\ref{intro.bkl} and \ref{intro.em}.
It is more interesting to consider perturbations that violate both (A) and (B). Concretely, this means taking a strongly singular solution with $Q \equiv 0$ and having $\Psi^2 < 1$ somewhere near $\mathcal{S}$, but then perturbing with $\hat{Q} \not\equiv 0$.

The subtle nature of this problem is illustrated by the following scenario (see Figure~\ref{fig:schwarzschild_reissnernordstrom}): take the unperturbed spacetime to be a strongly singular portion $\mathcal{D}$ of the Schwarzschild interior, then choose the perturbation to be such that $Q \equiv \epsilon \neq 0$ but $\phi$ still vanishes. This perturbed spacetime is therefore a portion of the Reissner-Nordstr\"om interior (see Chapter 5.5 of \cite{hawking_ellis_1973} for the explicit Reissner-Nordstr\"om metric). However, the Reissner-Nordstr\"om interior contains no singularities; it instead possesses a bifurcate null Cauchy horizon. Hence the perturbed maximal development $\hat{\mathcal{D}}$ arising from compact data contains no singularities, and is smoothly extendible in all directions, no matter the smallness of $\epsilon$. 

\begin{figure}[ht]
    \centering
    \begin{minipage}{.4\textwidth}
        \scalebox{0.7}{
    \begin{tikzpicture}[scale=0.7]
        \node (s) at (0, -6) [circle, draw, inner sep=0.5mm, fill=black] {};
        \node (il) at (-6, 0) [circle, draw, inner sep=0.5mm] {};
        \node (ir) at (+6, 0) [circle, draw, inner sep=0.5mm] {};
        \node (e) at (0, 6) {};
        \node (e+) at (0, 7) {};
        \node (r0) at (0, -4) [circle, draw, inner sep=0.5mm, fill=black] {};
        \node (cl) at (-2.5, -1.5) [circle, draw, inner sep=0.2mm, fill=black] {};
        \node (cr) at (+2.5, -1.5) [circle, draw, inner sep=0.2mm, fill=black] {};

        \path[fill=lightgray, opacity=0.5] (s) -- (-6, 0)
            .. controls (-2, 1) and (2, -1) .. (+6, 0) -- (s);
        \path[fill=lightgray] (r0) -- (-2.5, -1.5) -- (-0.9, 0.1)
            -- (1.15, -0.15) -- (2.5, -1.5) -- (r0);

        \node at (r0) [below] {\scriptsize $(u_0, v_0)$};
        \node at (cl) [above left] {\scriptsize $(u_1, v_0)$};
        \node at (cr) [above right] {\scriptsize $(u_0, v_1)$};
        \node at (0, -1.7) {\small $\mathcal{D}$};

        \draw [thick] (s) -- (il) node [midway, below left] {$\mathcal{H}^+$};
        \draw [thick] (s) -- (ir) node [midway, below right] {$\mathcal{H}^+$};
        \draw (r0) -- (cl) node [midway, below left=-1mm] {$C_0$};
        \draw (r0) -- (cr) node [midway, below right=-1mm] {$\underline{C}_0$};
        \draw [dashed] (il) -- (-7, -1) node [midway, above left] {$\mathcal{I}^+$};
        \draw [dashed] (ir) -- (7, -1) node [midway, above right] {$\mathcal{I}^+$};
        \draw [dashed] (il) .. controls (-2, 1) and (2, -1) .. (ir)
            node [midway, above=0.5mm] {$\mathcal{S} = \{ r = 0 \}$};
    \end{tikzpicture}}
    \end{minipage}
    \hspace{40pt}%
    \begin{minipage}{.4\textwidth}
        \scalebox{0.7}{
    \begin{tikzpicture}[scale=0.7]
        \node (s) at (0, -6) [circle, draw, inner sep=0.5mm, fill=black] {};
        \node (il) at (-6, 0) [circle, draw, inner sep=0.5mm] {};
        \node (ir) at (+6, 0) [circle, draw, inner sep=0.5mm] {};
        \node (e) at (0, 6) {};
        \node (e+) at (0, 7) {};
        \node (r0) at (0, -4) [circle, draw, inner sep=0.5mm, fill=black] {};
        \node (cl) at (-2.5, -1.5) [circle, draw, inner sep=0.2mm, fill=black] {};
        \node (cr) at (+2.5, -1.5) [circle, draw, inner sep=0.2mm, fill=black] {};

        \path[fill=lightgray, opacity=0.5] (s) -- (-6, 0)
            -- (0, 6) -- (6, 0) -- (s);
        \path[fill=lightgray] (r0) -- (-2.5, -1.5) -- (0, 1)  -- (2.5, -1.5) -- (r0);

        \node at (r0) [below] {\scriptsize $(u_0, v_0)$};
        \node at (cl) [above left] {\scriptsize $(u_1, v_0)$};
        \node at (cr) [above right] {\scriptsize $(u_0, v_1)$};
        \node at (0, 1) [above] {\scriptsize $(u_1, v_1)$};
        \node at (0, -1.7) {\small $\hat{\mathcal{D}}$};

        \draw [thick] (s) -- (il) node [midway, below left] {$\mathcal{H}^+$};
        \draw [thick] (s) -- (ir) node [midway, below right] {$\mathcal{H}^+$};
        \draw [dashed] (e) -- (il) node [midway, above left] {$\mathcal{CH}^+$};
        \draw [dashed] (e) -- (ir) node [midway, above right] {$\mathcal{CH}^+$};
        \draw (r0) -- (cl) node [midway, below left=-1mm] {$C_0$};
        \draw (r0) -- (cr) node [midway, below right=-1mm] {$\underline{C}_0$};
        \draw [dashed] (il) -- (-7, -1) node [midway, above left] {$\mathcal{I}^+$};
        \draw [dashed] (ir) -- (7, -1) node [midway, above right] {$\mathcal{I}^+$};
\end{tikzpicture}}
\end{minipage}
\captionsetup{justification = centering}
\caption{Penrose diagrams representing the Schwarzschild interior (left) and the Reissner-Nordstr\"om interior (right). The darker shaded region of the Schwarzschild interior represents a strongly singular spacetime in the sense of Theorem~\ref{roughthm:asymp}, whose singular boundary $\mathcal{S}$ is destroyed upon perturbations adding a non-trivial charge $Q \neq 0$. This effect is due to the presence of a (global) Cauchy horizon $\mathcal{CH}^+$ in the Reissner-Nordstr\"om interior.}
\label{fig:schwarzschild_reissnernordstrom}
\end{figure}
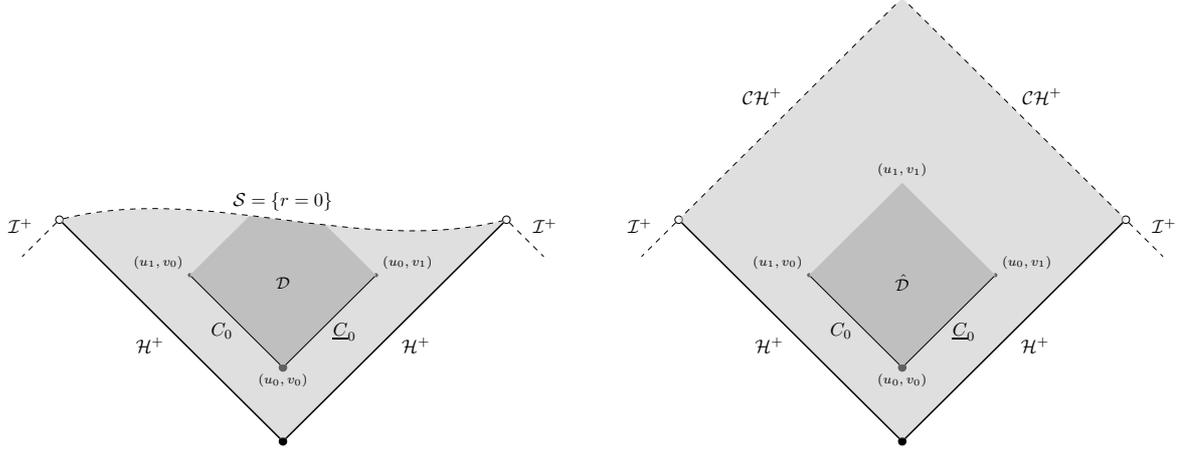

\blue{Note that, since $\phi \equiv 0$ in the Reissner-Nordstr\"om spacetime, the condition \eqref{eq:QTS} does not hold anywhere.} This example illustrates \blue{not only that verifying (A) and (B) in the $Q \neq 0$ setting is extremely subtle, but also} that \blue{\emph{charged}} perturbations of \blue{$Q = 0$ spacetimes} which violate (A) and (B) can, in fact, destroy the very presence of the spacelike singularity, both locally and globally. This can be viewed as one manifestation of the instability discussed in Sections~\ref{intro.bkl} and \ref{intro.em}.

In fact, using the methods of this paper, it is possible to qualitatively characterize the instability in our context. We argue that a charged perturbation of a strongly singular spacetime violating (\ref{eq:SKE}) will either fail to be singular, or contain a Kasner bounce such that the new singularity instead obeys the subcriticality assumption \eqref{eq:SKE}, so long as the double-null gauge remains reasonable in the sense that \eqref{eq:QTS} still holds. The proof of the following corollary will be deferred to \cite{me_kasner_stab}.

\begin{customcorollary}{IV}[Instability] \label{roughcor:instab}
    Let $(r, \Omega^2, \phi, Q)$ be a spherically symmetric and strongly singular solution to the spherically symmetric Einstein-Maxwell scalar field equations \eqref{eq:einstein}, \eqref{eq:scalar_wave}, \eqref{eq:maxwell}, such that $Q \equiv 0$. Suppose furthermore that (\ref{eq:SKE}) is violated near the singular boundary $\mathcal{S}$ -- in fact impose the stronger assumption that $\sup_{p \in \mathcal{S}} \Psi_{\infty}^2(p) < 1$.

    Given that the above spacetime arises from characteristic initial data $(r_i, \Omega_i^2, \phi_i, Q_i = 0)$ on $C_0 \cup \underline{C}_0$, consider now perturbed initial data $(\hat{r}_i, \hat{\Omega}^2_i, \hat{\phi}_i, \hat{Q}_i)$ with $\hat{Q}_i \neq 0$. Upon solving (\ref{eq:einstein}), (\ref{eq:scalar_wave}), (\ref{eq:maxwell}), the maximal future development of the perturbed data, given by functions $(\hat{r}, \hat{\Omega}^2, \hat{\phi}, \hat{Q})$ defined on a new domain of existence $\hat{\mathcal{D}}$, is such that one of the following holds:

    \begin{enumerate}[1.]
        \item \label{instab1} (Disappearance of singularity.)\, 
            The domain of existence is the entire causal rectangle
            $\hat{\mathcal{D}} = [u_0, u_1] \times [v_0, v_1],$
            so that $\inf_{\hat{\mathcal{D}}} r(u, v) > 0$ and each of $(\hat{r}, \hat{\Omega}^2, \hat{\phi}, \hat{Q})$ are bounded and regular in $\hat{\mathcal{D}}$. The singular boundary $\mathcal{S}$ disappears.
        \item \label{instab3} (Gauge degeneration.)\,
            The perturbed spacetime remains singular in the sense that $\inf_{\hat{\mathcal{D}}} r(u, v) = 0$, but the quantitative trapped sphere condition (\ref{eq:QTS}) now fails to hold, i.e.\ 
            $$\inf_{(u, v) \in \hat{\mathcal{D}}} \min \{ - r \partial_u r (u, v), - r \partial_v r (u, v) \} = 0.$$
        \item \label{instab2} (Bounce to subcritical regime.)\,
            The perturbed spacetime remains strongly singular and contains a new $\hat{r} = 0$ singular boundary $\hat{\mathcal{S}}$. The quantitative trapped sphere condition (\ref{eq:QTS}) still holds, however the new function $\hat{\Psi}$ (which is defined as in \eqref{eq:Psi} but is no longer necessarily continuous up to $\mathcal{S}$) departs quantitatively from that of the unperturbed spacetime, in the sense that
            $$\adjustlimits\liminf_{r_* \to 0} \sup_{r(p) \leq r_*} \hat{\Psi}^2(p) \geq 1.$$
            Furthermore, if $\hat{\Psi}$ can be shown to be continuously extendible to the boundary, and $\hat{\Psi}_{\infty}$ is its restriction to $\mathcal{S}$, then in fact $\inf_{p \in \mathcal{S}} \hat{\Psi}_{\infty}(p) > 1$. In particular, in this case (\ref{eq:SKE}) holds in the perturbed spacetime, once one restricts to a small neighborhood of the new singular boundary $\hat{\mathcal{S}}$.
    \end{enumerate}
    In each case of these three cases, some quantity will differ in a large sense from its value in the unperturbed spacetime, no matter how small $\epsilon \neq 0$ may be.
\end{customcorollary}

We now discuss circumstances in which each of the possibilities 1--3 may hold. An example of possibility~\ref{instab1} is the Schwarzschild to Reissner-Nordstr\"om perturbation described earlier. However, one could argue that this scenario is fine-tuned in the sense that $\phi$ is chosen to vanish identically in both the background spacetime and the perturbed spacetime, and therefore one cannot access the subcritical regime discussed in Section~\ref{intro.bkl}. If we chose a background spacetime other than exact Schwarzschild, or chose instead a perturbation with nontrivial $\phi$, it may be reasonable to expect that $\phi$ plays a significant role in the instability mechanism, and hence possibility~\ref{instab1} is conjecturally non-generic.

We turn next to possibility~\ref{instab3} in Corollary~\ref{roughcor:instab}. This is somewhat undesirable, as we cannot use our methods to describe the Kasner exponents at any point $p \in \hat{\mathcal{S}}$ at which (\ref{eq:QTS}) fails. Note for instance that the failure of (\ref{eq:QTS}) allows for portions of the $\hat{r} = 0$ singular boundary which are \textit{null} rather than spacelike. Unfortunately, it is difficult to rule out the failure of \eqref{eq:QTS}, and one expects to be able to find constructions where possibility~\ref{instab3} occurs. Nevertheless, one could hope that such instances are again non-generic, particularly in the context of Corollary~\ref{roughcor:instab} where the size $\epsilon$ of the perturbation is sufficiently small. 

This leaves possibility~\ref{instab2}, which suggests (at least in the case where $\hat{\Psi}$ remains bounded and continuously extendible to $\mathcal{S}$) the presence of Kasner bounces as described in Section~\ref{intro.bkl}. In the context of the Einstein-Maxwell scalar field system, this is perhaps the most interesting of the possibilities in Corollary~\ref{roughcor:instab} (and one could conjecture that this is the generic outcome when $\epsilon$ is small). Were possibility \ref{instab2} to occur, it must be due to the nonlinearities present in the Einstein equations, and is expected to be difficult to treat in any generality. 

We do mention, however, that the author and Van de Moortel \cite{MeVdM} have provided examples exhibiting a Kasner bounce. \cite{MeVdM} constructs spherically symmetric, spatially homogeneous solutions to the Einstein-Maxwell-\textit{charged} scalar field system, referred to as the interior of hairy black holes\footnote{Note, however, that there may not be any corresponding ``hairy black hole exterior'', and the term ``hairy'' is simply used to mean that there is no decay for the scalar field along the event horizon.}.
These spacetimes contain an $r = 0$ singular boundary $\mathcal{S}$, and the near-singularity region of the spacetimes features, in some cases, a ``Kasner inversion'' where a quantity morally equivalent to $\Psi$ ``inverts'' from $\Psi \approx \alpha \in (0, 1)$ (in a non-subcritical Kasner-like regime $\mathcal{K}_1$), to $\Psi \approx \alpha^{-1} > 1$ in the final subcritical Kasner-like regime $\mathcal{K}_2$ containing the singularity $\mathcal{S}$. This is depicted in Figure~\ref{fig:mevdm}.

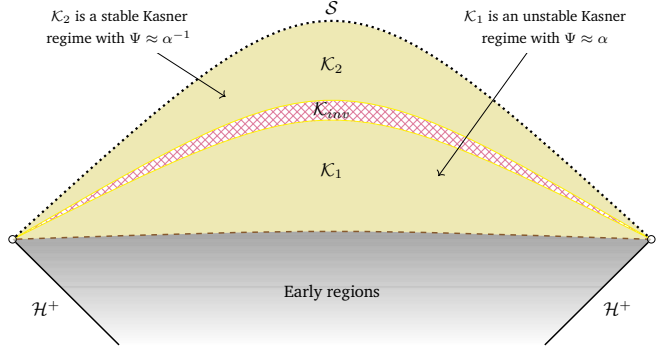
\begin{figure}[ht]
    \centering
    \scalebox{0.7}{
    \begin{tikzpicture}
        \path[shade, top color=gray, opacity=0.6] (-6, 0) .. controls (0, 0.2) .. (6, 0)
            -- (4, -2) -- (-4, -2) -- (-6, 0);
        \path[fill=yellow!80!black, opacity=0.4] (6, 0)
            .. controls (0, 0.2) .. (-6, 0)
            .. controls (0, 3) .. (6, 0);
        \path[pattern=crosshatch, pattern color=purple!50!white] (6, 0)
            .. controls (0, 3) .. (-6, 0)
            .. controls (0, 3.5) .. (6, 0);
        \path[fill=yellow!80!black, opacity=0.4] (6, 0)
            .. controls (0, 5.5) .. (-6, 0)
            .. controls (0, 3.5) .. (6, 0);

        \node (r) at (6, 0) [circle, draw, inner sep=0.5mm] {};
        \node (l) at (-6, 0) [circle, draw, inner sep=0.5mm] {};

        \draw [thick] (4, -2) -- (r)
            node [midway, below right] {$\mathcal{H}^+$};
        \draw [thick] (-4, -2) -- (l)
            node [midway, below left] {$\mathcal{H}^+$};
        \draw [brown!70!black, dashed, thick] (l) ..controls (0, +0.2) .. (r)
            node [midway, above=-0.5mm, black] {}; 
        \draw [yellow] (l) ..controls (0, 3) .. (r)
            node [midway, below, black] {};
        \draw [yellow] (l) ..controls (0, 3.5) .. (r)
            node [midway, above=-0.5mm, black] {};
        \draw [very thick, dotted] (l) .. controls (0, 5.5) .. (r)
            node [midway, above] {$\mathcal{S}$}; 

        \node at (0, -0.2) {}; 
        \node at (0, +1.3) {$\mathcal{K}_1$};
        \node at (0, +2.45) {$\mathcal{K}_{inv}$};
        \node at (0, +3.3) {$\mathcal{K}_2$};
        \node at (0, -1) {\small Early regions};
        \node at (0, -2.5) {};

        \draw [->] (4, 4) -- (2, 1.2);
        \node [align=center, fill=white] at (4, 4) {\footnotesize $\mathcal{K}_1$ is an unstable Kasner \\ \footnotesize regime with $\Psi \approx \alpha$};
        \draw [->] (-4, 4) -- (-2, 2.7);
        \node [align=center, fill=white] at (-4, 4) {\footnotesize $\mathcal{K}_2$ is a stable Kasner \\ \footnotesize regime with $\Psi \approx \alpha^{-1}$};
\end{tikzpicture}}
\captionsetup{justification = centering}
\caption{The near-singularity region of a spacetime constructed in \cite{MeVdM}. The region contains two Kasner-like regimes: $\mathcal{K}_1$ fails to satisfy an inequality of the form \eqref{eq:SKE} and is thus unstable. Before reaching the singularity $\mathcal{S}$, there is an ``inversion'' to the subcritical Kasner-like region $\mathcal{K}_2$. The (nonlinear) inversion procedure occurs in the transition region $\mathcal{K}_{inv}$.}
\label{fig:mevdm}
\end{figure}

The spatial homogeneity assumed in \cite{MeVdM} reduces the dynamics to a finite-dimensional ODE system. The inversion mechanism is then described via deriving a one-dimensional dynamical system for $\Psi$:
$$\frac{d\Psi}{dR} \approx - \Psi ( \Psi - \alpha ) (\Psi - \alpha^{-1}).$$
Here $R \coloneqq \log r^{-1}$ is such that $R \to + \infty$ as $r \to 0$, and the inversion from $\Psi \approx \alpha$ to $\Psi \approx \alpha^{-1}$ is exactly the progression of $\Psi$ away from the unstable fixed point $\alpha$ and towards the stable fixed point $\alpha^{-1}$. We note, finally, that via the correspondence (\ref{eq:Psikasner}) between $\Psi$ and the Kasner exponents $p_1, p_2, p_3$, the inversion $\alpha \mapsto \alpha^{-1}$ corresponds exactly to the Kasner map (\ref{eq:bounce}).

Though the spacetimes of \cite{MeVdM} are exactly spatially homogeneous, the Kasner bounces observed here can be considered in a wider context, due to the following observation: consider small perturbations of the inversion spacetime in Figure~\ref{fig:mevdm} in the region $\mathcal{K}_1$; then using a Cauchy stability argument which allows us to apply the stability result of Theorem~\ref{roughthm:stab} in the subcritical region $\mathcal{K}_2$, this shows that we can find examples of spacetimes possessing Kasner bounces outside of exact spatial homogeneity.

However more precise and quantitative control of the instability possibility~\ref{instab2}, particularly in the case that neither the background nor the perturbation are spatially homogeneous
-- even if one stays in spherical symmetry -- remains interesting and open. One major difficulty here is the presence of spatial derivatives, forbidding one from staying in the realm of ODEs. \blue{We also refer the reader to a recent preprint \cite{MeSurfaceSymPaper} of the author, where we find a large class of spacetimes solving the Einstein--Maxwell--scalar field system -- albeit written in a different gauge -- exhibiting such bounces outside spatial homogeneity; furthermore, a subclass of these can be checked to dynamically verify the conditions \eqref{eq:QTS} and \eqref{eq:SKE} upon transforming back to double null gauge. See also the related \cite{MeGowdyPaper} regarding the Gowdy symmetric Einstein vacuum equations.}

So far, we have discussed stability and instability within the class of spherically symmetric (but possibly charged) perturbations. Another direction would be to consider instead perturbations outside of spherical symmetry, where one may expect stability and instability results corresponding to Theorem~\ref{roughthm:stab} and Corollary~\ref{roughcor:instab} respectively, in light of Sections~\ref{intro.bkl} and \ref{intro.em}. This will be elaborated upon in Section~\ref{intro.related}.

\subsection{Related works and future directions} \label{intro.related}


\subsubsection{Applications to spherically symmetric collapse}

Perhaps the most familiar setting under which the system (\ref{eq:einstein}), (\ref{eq:scalar_wave}), (\ref{eq:maxwell}) has been studied is that of one- or two-ended spherically gravitational collapse. We discuss first the $F_{\mu\nu} \equiv 0$ case, where Christodoulou \cite{Christodoulou_formation, Christodoulou_BV, Christodoulou_wcc} has shown that the generic Penrose diagram representing collapse of one-ended asymptotically flat data to a black hole is that of Figure~\ref{fig:christodoulou_collapse}.

Theorem~\ref{roughthm:asymp} shows that for such spacetimes, we can associate to any trapped $2$-sphere $(u, v) \in \mathcal{T}: \{ (u, v): \partial_u r (u, v) < 0, \partial_v r (u, v) < 0 \}$, the functions $\Psi$, $\Xi$ and $\mathfrak{M}$, which each extend in a continuous manner to functions $\Psi_{\infty}$, $\Xi_{\infty}$ and $\mathfrak{M}_{\infty}$ defined along $\mathcal{S}$. We focus on $\Psi_{\infty}$, which determines the Kasner exponents associated to any $p \in \mathcal{S}$, via (\ref{eq:Psikasner}).

There are two regions of the Penrose diagram which are of particular interest -- the region near spacelike infinity $i^+$, and the region near the first singularity at the center $b_0$; each of these is depicted in Figure~\ref{fig:i+b0_zoom}.

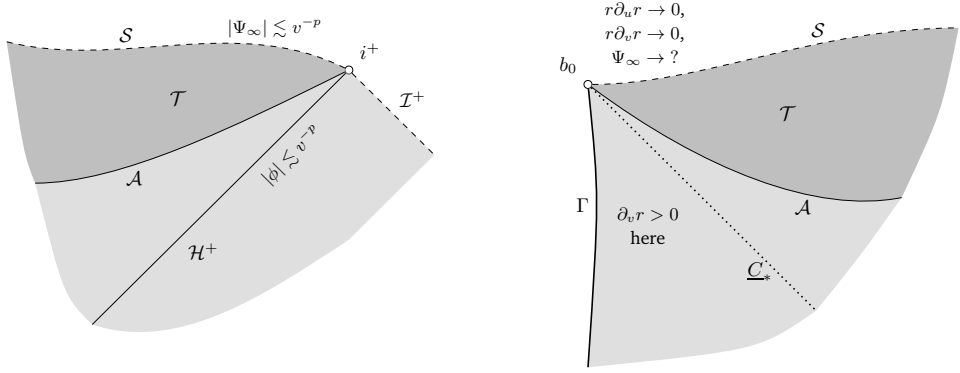
\begin{figure}[ht] 
    \centering
\begin{minipage}{.4\textwidth}
    \scalebox{0.75}{
    \begin{tikzpicture}
        \node (i+) at (6, 6) [circle, draw, inner sep=0.5mm] {};

        \path[fill=lightgray] (i+) .. controls (4, 7) and (2, 6) .. (0, 6.5)
            .. controls (0.3, 4.5) .. (0.5, 4) 
            .. controls (2, 4) and (4, 5) .. (i+);
        \path[fill=lightgray, opacity=0.5] (i+) .. controls (4, 5) and (2, 4) .. (0.5, 4)
            .. controls (1, 2) .. (1.5, 1.5) 
            .. controls (3, 1) and (4.5, 2) .. (6, 3) -- (7.5, 4.5) -- (i+);

        \node at (3, 5.5) {$\mathcal{T}$};
        \node [above right=0.2mm of i+] {$i^+$};
        \node [rotate=45] at (5, 4.5) {\small $|\phi| \lesssim v^{-p}$};
        \node at (4.7, 6.75) {\small $|\Psi_{\infty}| \lesssim v^{-p}$};
        
        \draw (i+) -- (1.5, 1.5) node [pos=0.65, below right] {$\mathcal{H}^+$};
        \draw [dashed] (i+) -- (7.5, 4.5) node [midway, above right] {$\mathcal{I}^+$};
        \draw [dashed] (i+) .. controls (4, 7) and (2, 6) .. (0, 6.5) node [pos=0.65, above] {$\mathcal{S}$};
        \draw (i+) .. controls (4, 5) and (2, 4) .. (0.5, 4) node [pos=0.65, below] {$\mathcal{A}$};
\end{tikzpicture}}
\end{minipage} \hspace{10pt}
\begin{minipage}{.4\textwidth}
    \scalebox{0.75}{
    \begin{tikzpicture}
        \node (b0) at (0, 5) [circle, draw, inner sep=0.5mm] {};

        \path[fill=lightgray] (b0) .. controls (2, 5) and (4, 6) .. (6.5, 6)
            .. controls (6.2, 4.5) .. (5.5, 3)
            .. controls (4, 2.7) and (2, 3.5) .. (b0);
        \path[fill=lightgray, opacity=0.5] (b0) .. controls (2, 3.5) and (4, 2.7) .. (5.5, 3)
            .. controls (4.8, 2) .. (4, 1)
            .. controls (3, 0.3) .. (0, 0)
            .. controls (0.2, 3) .. (b0);

        \node [above left=0.2mm of b0] {$b_0$};
        \node at (3.5, 4.5) {$\mathcal{T}$};
        \node [align=center, fill=white] at (1, 5.9) {\small $r \partial_u r \to 0$, \\ \small $ r \partial_v r \to 0$, \\ \small $\Psi_{\infty} \to \;?$};
        \node [align=center] at (1.05, 2.5) {\small $\partial_v r > 0$ \\ \small here};

        \draw [dashed] (b0) .. controls (2, 5) and (4, 6) .. (6.5, 6) node [pos=0.65, above] {$\mathcal{S}$};
        \draw (b0) .. controls (2, 3.5) and (4, 2.7) .. (5.5, 3) node [pos=0.65, below] {$\mathcal{A}$};
        \draw [thick, dotted] (b0) -- (4, 1) node [pos=0.75, below] {$\underline{C}_*$};
        \draw [thick] (b0) .. controls (0.2, 3) .. (0, 0) node[left, midway] {$\Gamma$};
\end{tikzpicture}}
\end{minipage}
\captionsetup{justification = centering}
\caption{Zooming in on the portions of the singular boundary $\mathcal{S}$ near $i^+$ (left) and $b_0$ (right), indicating the relevance of Theorem~\ref{roughthm:asymp} in these regions}
\label{fig:i+b0_zoom}
\end{figure}

For the purpose of studying the near-$i^+$ region, \cite{AnGajic} uses Price's law decay of the scalar field $\phi$ along the event horizon $\mathcal{H}^+$, namely that for an outgoing Eddington-Finkelstein type coordinate $v$ and some $p > 1$,
\begin{equation} \label{eq:priceslaw}
    |\phi|_{\mathcal{H}^+}(v)| \lesssim v^{-p}.
\end{equation}
Decay of the form \eqref{eq:priceslaw} was first proven in the spherically symmetric setting in \cite{dr_priceslaw} -- in fact \cite{dr_priceslaw} also allows for $F_{\mu\nu} \neq 0$ so long as the black hole is sub-extremal in the limit. 

Returning to $F_{\mu\nu} = 0$, \cite{AnGajic} propagates the decay \eqref{eq:priceslaw} towards $\mathcal{S}$, showing that for the same coordinate $v$, one has
    $|\Psi_{\infty}(v)| \lesssim v^{-p}$
for large $v$. Via Theorem~\ref{roughthm:asymp} it is thus reasonable to say that the singularity is ``almost-Schwarzschildian'' near $i^+$. Note, in particular, that (\ref{eq:SKE}) is necessarily violated in this region, and we expect instability should we turn on either electromagnetic or non-symmetric perturbations.

We turn now to the region near $b_0$. Though Theorem~\ref{roughthm:asymp} attributes a value for $\Psi_{\infty}$, and hence the Kasner exponents, for all $p \in \mathcal{S} \setminus \{ b_0 \}$, it is highly unclear what happens at the point $b_0$ itself. For $b_0$, one cannot apply the analysis presented in this article, since $b_0$ is not preceded by a trapped region -- in fact $J^-(b_0) \cap \mathcal{T} = \emptyset$. Hence new techniques are needed to understand, for instance, if $\Psi_{\infty}(p)$ has a limit as $p \to b_0$. This may be related to the subtle issue of naked singularities, see also \cite{Christodoulou_examples, Christodoulou_cc, Christodoulou_wcc}.


Moving on, we introduce a non-trivial Maxwell field $F_{\mu\nu} \neq 0$. The Einstein-Maxwell-scalar field model considered in this article forces us to consider only two-ended asymptotically flat black hole spacetimes. For this model, \cite{dafermos03, dafermos05} and later \cite{LukOh1, LukOh2} show that the possible Penrose diagrams are as in Figure~\ref{fig:twoended}, and the black hole contains a non-empty Cauchy horizon $\mathcal{CH}^+$ across which the metric is $C^0$-extendible.
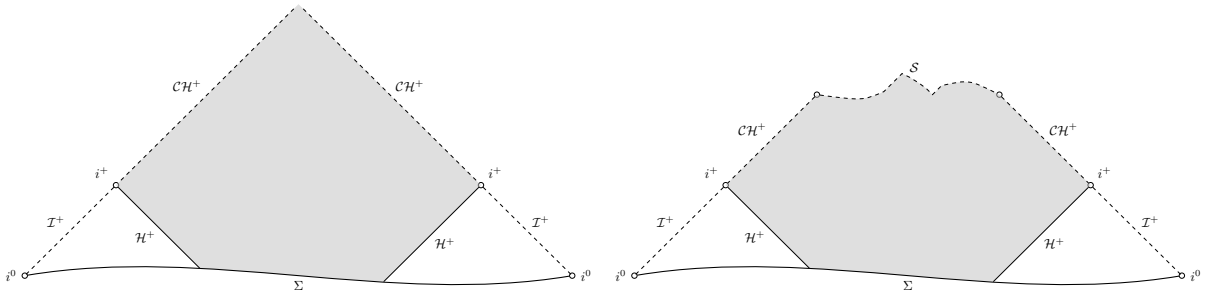
\begin{figure}[ht]
    \centering
    \begin{minipage}{.45\textwidth}
        \scalebox{0.5}{
    \begin{tikzpicture}[scale=0.8]
        \node (il) at (-6, 0) [circle, draw, inner sep=0.5mm] {};
        \node (ir) at (+6, 0) [circle, draw, inner sep=0.5mm] {};
        \node (e) at (0, 6) {};
        \node (i0l) at (-9, -3) [circle, draw, inner sep=0.5mm] {};
        \node (i0r) at (+9, -3) [circle, draw, inner sep=0.5mm] {};

        \node [above right=0.2mm of ir] {$i^+$}; 
        \node [above left=0.2mm of il] {$i^+$};
        \node [right=0.2mm of i0r] {$i^0$}; 
        \node [left=0.2mm of i0l] {$i^0$};
        \path[fill=lightgray, opacity=0.5] (0, 6) -- (-6, 0) -- (-3.25, -2.75)
            .. controls (0, -3) .. (2.8, -3.2) -- (6, 0) --  (0, 6);

        \draw (-3.25, -2.75) -- (il) node [midway, below left=-0.5mm] {$\mathcal{H}^+$};
        \draw (2.8, -3.2) -- (ir) node [midway, below right=-0.5mm] {$\mathcal{H}^+$};
        \draw [dashed] (e) -- (il) node [midway, above left] {$\mathcal{CH}^+$};
        \draw [dashed] (e) -- (ir) node [midway, above right] {$\mathcal{CH}^+$};
        \draw [dashed] (il) -- (-9, -3) node [midway, above left] {$\mathcal{I}^+$};
        \draw [dashed] (ir) -- (+9, -3) node [midway, above right] {$\mathcal{I}^+$};
        \draw [thick] (i0l) .. controls (-3, -2) and (3, -4)  .. (i0r)
            node [midway, below] {$\Sigma$};
    \end{tikzpicture}}
\end{minipage} \hspace{10pt}
\begin{minipage}{.45\textwidth}
    \scalebox{0.5}{
    \begin{tikzpicture}[scale=0.8]
        \node (il) at (-6, 0) [circle, draw, inner sep=0.5mm] {};
        \node (ir) at (+6, 0) [circle, draw, inner sep=0.5mm] {};
        \node (e) at (0, 6) {};
        \node (i0l) at (-9, -3) [circle, draw, inner sep=0.5mm] {};
        \node (i0r) at (+9, -3) [circle, draw, inner sep=0.5mm] {};
        \node (sl) at (-3, 3) [circle, draw, inner sep=0.5mm] {};
        \node (sr) at (+3, 3) [circle, draw, inner sep=0.5mm] {};

        \node [above right=0.2mm of ir] {$i^+$}; 
        \node [above left=0.2mm of il] {$i^+$};
        \node [right=0.2mm of i0r] {$i^0$}; 
        \node [left=0.2mm of i0l] {$i^0$};
        \path[fill=lightgray, opacity=0.5] (-6, 0) -- (-3, 3)
            .. controls (-1.5,2.8) .. (-0.8, 3.1) -- (-0.2, 3.7)
            .. controls (0.2, 3.5) .. (0.6, 3.2)
            -- (0.8, 3.0) -- (1.1, 3.3)
            .. controls (2, 3.5) .. (+3, +3) -- (6, 0) -- (2.8, -3.2)
            .. controls (0, -3) .. (-3.25, -2.75) -- (-6, 0);

        \draw (-3.25, -2.75) -- (il) node [midway, below left=-0.5mm] {$\mathcal{H}^+$};
        \draw (2.8, -3.2) -- (ir) node [midway, below right=-0.5mm] {$\mathcal{H}^+$};
        \draw [dashed] (sl) -- (il) node [midway, above left] {$\mathcal{CH}^+$};
        \draw [dashed] (sr) -- (ir) node [midway, above right] {$\mathcal{CH}^+$};
        \draw [dashed] (il) -- (-9, -3) node [midway, above left] {$\mathcal{I}^+$};
        \draw [dashed] (ir) -- (+9, -3) node [midway, above right] {$\mathcal{I}^+$};
        \draw [dashed] (sl) .. controls (-1.5,2.8) .. (-0.8, 3.1) -- (-0.2, 3.7)
            .. controls (0.2, 3.5) .. (0.6, 3.2)
            node [midway, above=0.8mm] {$\mathcal{S}$}
            -- (0.8, 3.0) -- (1.1, 3.3)
            .. controls (2, 3.5) .. (sr);
        \draw [thick] (i0l) .. controls (-3, -2) and (3, -4)  .. (i0r)
            node [midway, below] {$\Sigma$};
    \end{tikzpicture}}
\end{minipage}
\captionsetup{justification=centering}
\caption{Penrose diagrams representing two-ended spherically symmetric gravitational collapse spacetime. The shaded regions correspond to the black hole interiors ($\mathcal{A}$ and $\mathcal{T}$ are no longer drawn), which always possess a non-empty Cauchy-horizon $\mathcal{CH}^+$ and in some cases an $r = 0$ spacelike singularity $\mathcal{S}$. If $\mathcal{S} \neq \emptyset$ as in the rightmost picture, it is unknown whether $\mathcal{S}$ could contain null pieces.}
\label{fig:twoended}
\end{figure}

Since the Cauchy horizon is null and $r|_{\mathcal{CH}^+} \neq 0$, Theorem~\ref{roughthm:asymp} does not apply there\footnote{Even so, $\mathcal{CH}^+$ is generically singular in a $C^2$ sense, see \cite{dafermos05, LukOh1}. This is historically significant in showing that singularities in General Relativity are not always spacelike, in contrast to the original predictions of BKL.}. 
In fact, there are spacetimes e.g.~perturbations of Reissner-Nordstr\"om \cite{dafermos03} for which the Penrose diagram is as in the left of Figure~\ref{fig:twoended}. Nevertheless, for large data an $r = 0$ singularity $\mathcal{S}$ could arise, as in the right of Figure~\ref{fig:twoended}.

In such cases, Theorem~\ref{roughthm:asymp} \textit{could} apply, though we must be careful as (\ref{eq:QTS}) may not necessarily hold. In particular it is not clear a priori whether $\mathcal{S}$ could contain null segments. It is worth mentioning, however, that one could refer to the charged examples of Section~\ref{examples.just} and apply a straightforward gluing argument\footnote{See \cite{ChristophRyan} for a more detailed discussion of gluing for Einstein's equation in spherical symmetry.} to guarantee that at least part of $\mathcal{S}$ obeys the assumptions of Theorem~\ref{roughthm:asymp}.

One particularly worthwhile problem would be to consider the asymptotics at the future endpoint of the Cauchy horizon; from Section~\ref{intro.stab} it is expected that either we have positive Kasner exponents here, or the gauge otherwise breaks down. We hope to return to this problem in future work.

Finally, we mention the spherically symmetric Einstein-Maxwell-\textit{charged} scalar field model. This matter model allows for one-ended gravitational collapse, while retaining the possibility of a non-empty Cauchy horizon. Therefore all of the issues mentioned here could arise. We refer the reader to \cite{kommemi, Moortel18, MoortelC2, MoortelChristoph, MoortelBreakdown} for further discussions of the black hole interior in this model.

\subsubsection{Beyond spherical symmetry}

Outside of the spherically symmetric setting, there are many more works exhibiting examples of spacelike singularity formation (largely in the \textit{asymptotically velocity dominated} regime where spatial derivatives are dominated by timelike derivatives), as well as works studying the stability properties of such solutions. While by no means a complete list, we discuss several of these here:
\begin{itemize}
    \item \textbf{Gowdy symmetry: }
        In \cite{SCC_PolarizedGowdy}, the authors study polarized Gowdy\footnote{Gowdy symmetric spacetimes are such that there exist spacelike Killing vector fields $\mathbf{Y}^{\mu}$ and $\mathbf{Z}^{\mu}$, each with $\mathbb{S}^1$ orbits, such that their \textit{twist constants} $\epsilon_{\alpha \beta \gamma \delta} \mathbf{Y}^{\alpha} \mathbf{Z}^{\beta} \nabla^{\gamma} \mathbf{Y}^{\delta}$ and $\epsilon_{\alpha \beta \gamma \delta} \mathbf{Y}^{\alpha} \mathbf{Z}^{\beta} \nabla^{\gamma} \mathbf{Z}^{\delta}$ vanish, where $\epsilon_{\alpha \beta \gamma \delta}$ is the spacetime volume form. If $\mathbf{Y}^{\mu}$ and $\mathbf{Z}^{\mu}$ are moreover orthogonal, then the spacetime is polarized Gowdy.} solutions
        to the Einstein vacuum equations. They prove strong cosmic censorship in this setting, meaning that for some open and dense set of polarized Gowdy initial data, the maximal globally hyperbolic development of this data is inextendible, and with curvature blow-up in one of the directions.

        Ringstr\"om then generalized these results in \cite{SCC_T3Gowdy} to the $\mathbb{T}^3$-symmetric Gowdy setting, without polarization. In contrast to \cite{SCC_PolarizedGowdy}, singularities in unpolarized Gowdy spacetimes exhibit a phenomenon known as ``spikes'', where spatial derivatives become large. In \cite{SCC_T3Gowdy} it is shown that an open and dense set of initial data exhibits a Kasner-like singularity with curvature blow-up, except at a finite number of spikes. 

    \item \textbf{Fuchsian constructions: }
        \cite{AnderssonRendall} uses Fuchsian methods (introduced into the study of \eqref{eq:einstein} in \cite{KichenassamyRendall}) to construct asymptotically velocity dominated solutions to the Einstein-scalar field and Einstein-stiff fluid systems obeying the asymptotics of \eqref{eq:bkl}. These spacetimes are real analytic, and subcritical in the sense of (\ref{eq:bkl_smallness}). This was extended to more general regimes, including the Einstein-Maxwell-scalar field system as well as higher-dimensional vacuum, in \cite{DHRW}.

        Recently, \cite{FournodavlosLuk} constructed singular Kasner-like solutions to the Einstein-vacuum equations in $1+3$-dimensions based on the heuristics of \cite{kl63}. Being vacuum solutions, the asymptotic data prescribed at the singular boundary prescribed must obey the further functional constraint $\lambda^1_{23} \equiv 0$, see (\ref{eq:bkl_ricci}). 

    \item \textbf{Stable big bang results: }
        It is important to determine whether solutions exhibiting (subcritical) Kasner-like behaviour are stable to small perturbations; in particular are properties such as curvature blow-up stable? Such stability results have been shown for \underline{explicit} background solutions such as FLRW and exact Kasner spacetimes, with the state-of-the-art result due to Fournodavlos, Rodnianski and Speck \cite{FournodavlosRodnianskiSpeck}, which proves the (past) stability of exact generalized Kasner spacetimes solving the Einstein-scalar field equations with (constant) Kasner exponents in the subcritical range. We mention also \cite{RodnianskiSpeck1, RodnianskiSpeck2, SpeckS3, FajmanUrban, BeyerOliynyk}.

    \item \textbf{Perturbations of Schwarzschild: }
        The Schwarzschild singularity, having Kasner exponents of $p_1 = - \frac{1}{3}$ and $p_2 = p_3 = \frac{2}{3}$, is expected to be unstable in the general setting. However, as Theorem~\ref{roughthm:stab} implies in the spherically symmetric (but electromagnetism-free) setting, there are certain families of perturbations under which the Schwarzschild singularity will be stable. We mention \cite{AlexakisFournodavlos}, where the authors show stability of the Schwarzschild singularity with respect to $U(1)$ polarized axi-symmetric\footnote{This leads to spacetimes possessing a hypersurface orthogonal and spacelike Killing vector field $\mathbf{K}^{\mu}$ with $\mathbb{S}^1$ orbits.} perturbations in vacuum. Another important contribution is \cite{FournodavlosBackwards}, which yields the construction of a large class of spacetimes converging to a singular Schwarzschildian $2$-sphere.

    \item \textbf{Spatially homogeneous spacetimes: } 
        It remains hugely important to study singularities that exhibit chaotic and oscillatory behaviour such as that predicted by \cite{bkl71}. To date, this has only been done in the exactly spatially homogeneous context, where the dynamics reduce to a system of finite-dimensional ODEs. Using this reduction, there are several results \cite{ringstrom_bianchi, Weaver_bianchi} giving a detailed description of the possible near-singularity states of suitable spatially homogeneous spacetimes, including spacetimes which exhibit infinitely many Kasner bounces. 
        (We mention also \cite{MoortelViolent, MeVdM}, where certain spherically symmetric and spatially homogeneous solutions\footnote{Such spacetimes are Kantowski-Sachs cosmologies, with spatial topology $\mathbb{S}^2 \times \mathbb{R}$ or $\mathbb{S}^2 \times \mathbb{S}^1$.} to the Einstein-Maxwell-scalar field are analysed, but the scalar field is either charged or massive. In both cases, the Kasner-like dynamics near the $r = 0$ singularity can be described in detail, including Kasner inversions in the case of \cite{MeVdM} -- see Section~\ref{intro.stab}.)
\end{itemize}

We may consider our strongly singular spacetimes described in Theorem~\ref{roughthm:asymp} in the context of each of these topics. One problem of particular interest is the nonlinear stability of the spacetimes of Theorem~\ref{roughthm:asymp} with respect to perturbations outside of symmetry, at least for those satisfying the subcriticality property \eqref{eq:SKE}.

This introduces the following difficulties outside of those already present in, say, \cite{FournodavlosRodnianskiSpeck}. Firstly, while the background Kasner spacetime in \cite{FournodavlosRodnianskiSpeck} is explicit and moreover written in the Gaussian or CMC gauge often applied in the analysis of such singular spacetimes, the double-null gauge used in Theorem~\ref{roughthm:asymp} is perhaps inappropriate outside of spherical symmetry. 
One thus expects to have to rewrite the spacetimes of Theorem~\ref{roughthm:asymp} in a more appropriate gauge; the hope being that Corollary~\ref{cor:bkl} would be one step in such a process.

Secondly, in Theorem~\ref{roughthm:asymp}, $\Psi_{\infty}(p)$, and therefore the Kasner exponents, are allowed to vary along $\mathcal{S}$, which would introduce additional error terms into the analysis upon commuting with spatial derivatives. In particular one may need to use Theorem~\ref{roughthm:asymp} to find improved estimates on such spatial derivatives.
%

Finally, \cite{FournodavlosRodnianskiSpeck} is reliant on elliptic estimates on $\mathbb{T}^3$; in our case one expects to need to localize the stability analysis to suitable causal subdomains near $\mathcal{S}$, thus introducing additional boundary terms in the process. Note that \cite{BeyerOliynyk} suggests one approach to a localized stable Big Bang.

A full resolution of the nonlinear stability of the singular boundary $\mathcal{S}$ in Theorem~\ref{roughthm:asymp} would require dealing with all the above issues; nonetheless, we hope to return to this problem in the near future.

\subsubsection{The scattering problem for spacelike singularities}

As in \cite{AnderssonRendall, FournodavlosBackwards, FournodavlosLuk}, one could take the approach of starting with asymptotic data on the singularity, and ask whether one can construct a singular solution to Einstein's equation \eqref{eq:einstein} realizing such asymptotics. In light of Theorem~\ref{roughthm:asymp}, the appropriate asymptotic data should be the functions $\Psi_{\infty}(p)$, $\Xi_{\infty}(p)$ and $\mathfrak{M}_{\infty}(p)$.

Unfortunately, there is one major shortcoming of Theorem~\ref{roughthm:asymp} -- namely that to have any hope of solving (\ref{eq:einstein}) backwards, one needs to obey the following \textit{asymptotic momentum constraint} at $\mathcal{S}$:
\begin{equation} \label{eq:asymp_momentum}
    \tilde{\nabla} \log \mathfrak{M}_{\infty} (p) = 2 \Psi_{\infty}(p) \cdot \tilde{\nabla}\, \Xi_{\infty}(p),
\end{equation}
where $\tilde{\nabla}$ is a suitable spatial derivative on $\mathcal{S}$. See Section~5 of \cite{BuonannoDamourVeneziano} for heuristics deriving \eqref{eq:asymp_momentum} from the usual momentum constraint. However, Theorem~\ref{roughthm:asymp} gives only H\"older regularity for $\log \mathfrak{M}_{\infty}$ and $\Xi_{\infty}$, hence (\ref{eq:asymp_momentum}) cannot even be interpreted meaningfully. Finding a suitable gauge, and notion of spatial derivative, such that (\ref{eq:asymp_momentum}) can be verified remains an open problem in the present setting.

Nonetheless, one could begin with smooth quantities $\Psi_{\infty}(p)$, $\Xi_{\infty}(p)$ and $\log \mathfrak{M}_{\infty}(p)$, satisfying (\ref{eq:asymp_momentum}). Then one expects, using methods from \cite{AnderssonRendall, FournodavlosLuk}, that one could reconstruct the strongly singular spacetime solving (\ref{eq:einstein}) and realizing Theorem~\ref{roughthm:asymp} with $\Psi_{\infty}$, $\Xi_{\infty}$ and $\mathfrak{M}_{\infty}$ as given. We hope to return to this problem in the future. However, due to the above regularity issues, we are still far from a satisfactory scattering theory linking asymptotic data at the singularity to data on any spacelike or bifurcate null hypersurface.

\subsection{Outline of the paper} \label{intro.outline}

In Section~\ref{setup}, we explain the analytic setup of the problem, including a description of the initial data and the a priori assumptions regarding the strongly singular maximal future development $\mathcal{D}$. In Section~\ref{statement}, we state a precise version of Theorem~\ref{roughthm:asymp}, which will involve one theorem when $F_{\mu\nu} \equiv 0$ and another when $F_{\mu\nu} \not\equiv 0$. We also state Corollary~\ref{cor:bkl}, which provides a leading order BKL-like expansion in asymptotically CMC coordinates. 

Sections~\ref{upperbounds} to \ref{bklasymp} will be dedicated to the proof of these theorems, with Section~\ref{upperbounds} first providing (non-sharp) upper bounds for $r^2$, $\Omega^2$, $\phi$ and their derivatives, then Section~\ref{scalarfield} and Section~\ref{bklasymp} providing precise asympotics for the scalar field $\phi$ and the remaining geometric quantities respectively. We also include the proof of Corollary~\ref{cor:bkl} in Section~\ref{bklasymp}. Finally, Section~\ref{examples} will include various explicit examples of spherically symmetric spacetimes possessing a spacelike singularity, in order to illustrate how Theorem~\ref{roughthm:asymp} may be applied.

\subsection*{Acknowledgements}

The author would like to thank his advisor Mihalis Dafermos for enlightening discussions from which this project originated, and for his useful comments on the manuscript. The author would also like to express gratitude to Maxime Van de Moortel, Xinliang An and Grigorios Fournodavlos for numerous helpful discussions and proposals regarding future directions of study.

\vspace{1pt}
\noindent
Data sharing not applicable to this article as no datasets were generated or analysed during the current study.

%% file: preliminaries.tex
\section{Set-up of the problem} \label{setup}

\subsection{The Einstein-Maxwell-scalar field model} \label{setup.emsf}
We study the Einstein-Maxwell-scalar field system on a spacetime $(\mathcal{M}, g)$, with massless scalar field $\phi: \mathcal{M} \to \R$ and Maxwell field given by the $2$-form $F \in \Omega^2(\mathcal{M})$. The dynamical quantities $(g, \phi, F)$ evolve as follows:
\begin{gather} \label{eq:einstein_msf}
    \mbox{Ric}_{\mu \nu} [g] - \frac{1}{2} R[g] g_{\mu\nu} = 2 \,T_{\mu \nu} \coloneqq 2 ( T_{\mu \nu}^{SF} [\phi] + T_{\mu \nu}^{EM} [F]) \\[0.5em]
    \label{eq:energy_momentum_sf2}
    T_{\mu \nu}^{SF}[\phi] = \nabla_{\mu} \phi \nabla_{\nu} \phi - \frac{1}{2} g_{\mu\nu} \nabla_{\rho} \phi \nabla^{\rho} \phi, \\[0.5em]
    \label{eq:energy_momentum_em2}
    T_{\mu \nu}[F] = F_{\mu \rho}  F_{\nu}^{\phantom{\nu} \rho} - \frac{1}{4} g_{\mu \nu} F^{\rho \sigma} F_{\rho \sigma}, \\[0.5em]
    \label{eq:eom_scalar_field}
    \square_g \phi = \nabla^{\rho} \nabla_{\rho} \phi = 0, \\[0.5em]
    \label{eq:eom_maxwell}
    \nabla_{[\lambda} F_{\mu \nu]} = 0, \qquad \nabla^{\nu} F_{\mu \nu} = 0.
\end{gather}
One particular solution of these equations is given by the Schwarzschild metric (\ref{eq:schwarzschild}), with vanishing matter fields $\phi \equiv 0$ and $F \equiv 0$. Another exact solution is Reissner-Nordstr\"om, where $F \neq 0$ is non-trivial but $\phi$ still vanishes. We refer the reader to Section~\ref{examples} for examples of explicit solutions with non-vanishing scalar field.

\subsection{Spherically symmetric solutions} \label{setup.ss}

We investigate solutions of (\ref{eq:einstein_msf}), (\ref{eq:eom_scalar_field}), (\ref{eq:eom_maxwell}) which are spherically symmetric as defined in Section \ref{intro.model}. As explained there, we write the metric in a (global) double-null coordinate gauge:
\begin{equation} \label{eq:doublenull}
    g = - \Omega^2(u, v) \, du \, dv + r^2(u, v) \blue{(d \theta^2 + \sin^2 \theta \, d \varphi^2)} .
\end{equation}

The scalar field $\phi = \phi(u, v)$ now depends only on $(u, v)$, while the Maxwell field $F$ can now be written as
\begin{equation} \label{eq:em_doublenull}
    F = \frac{Q(u, v)}{2r^2(u, v)} \cdot \Omega^2(u, v) \, du \wedge dv,
\end{equation}
where, for now, we allow $Q$ to depend on $(u, v)$.

In this choice of gauge, we use the system (\ref{eq:einstein_msf}), (\ref{eq:eom_scalar_field}), (\ref{eq:eom_maxwell}) to write the following equations for the dynamical quantities $r(u, v)$, $\Omega^2(u, v)$, $\phi(u, v)$ and $Q(u, v)$.  The $uu-$ and $vv-$components of (\ref{eq:einstein_msf}) give:
\begin{equation} \label{eq:raych_u}
    \partial_u ( - \Omega^{-2} \partial_u r ) = \Omega^{-2} r (\partial_u \phi)^2,
\end{equation}
\begin{equation} \label{eq:raych_v}
    \partial_v ( - \Omega^{-2} \partial_v r ) = \Omega^{-2} r (\partial_v \phi)^2.
\end{equation}
These are known as the \textit{Raychaudhuri equations}, and may be interpreted as constraint equations for $\Omega^2$ along null hypersurfaces of constant $u$ or constant $v$.

The remaining components of (\ref{eq:einstein_msf}) imply the following propagation equations for $r^2$ and $\Omega^2$:
\begin{equation} \label{eq:wave_r}
    \partial_u \partial_v r = - \frac{\Omega^2}{4r} - \frac{\partial_u r \partial_v r}{r} + \frac{\Omega^2}{4r^3} Q^2,
\end{equation}
\begin{equation} \label{eq:wave_omega}
    \partial_u \partial_v \log(\Omega^2) = \frac{\Omega^2}{2r^2} + \frac{2 \partial_u r \partial_v r}{r^2} - \frac{\Omega^2}{r^4} Q^2 - 2 \partial_u \phi \partial_v  \phi.
\end{equation}
We shall often find it advantageous to rewrite (\ref{eq:wave_r}) in the following forms:
\begin{equation} \label{eq:wave_r_u}
    \partial_u (- r \partial_v r) = \frac{\Omega^2}{4} \left( 1 - \frac{Q^2}{r^2} \right),
\end{equation}
\begin{equation} \label{eq:wave_r_v}
    \partial_v (- r \partial_u r) = \frac{\Omega^2}{4} \left( 1 - \frac{Q^2}{r^2} \right).
\end{equation}

The wave equation (\ref{eq:eom_scalar_field}) becomes the following equation for $\phi(u, v)$:
\begin{equation} \label{eq:wave_phi}
    \partial_u \partial_v \phi = - \frac{\partial_u r \cdot \partial_v \phi}{r} - \frac{\partial_v r \cdot \partial_u \phi}{r},
\end{equation}
which we also rewrite in the following useful forms:
\begin{equation} \label{eq:wave_phi_u}
    \partial_u ( r \partial_v \phi) = - \partial_u r \cdot \partial_v \phi,
\end{equation}
\begin{equation} \label{eq:wave_phi_v}
    \partial_v ( r \partial_u \phi) = - \partial_v r \cdot \partial_u \phi.
\end{equation}

Finally, the $u-$ and $v-$ components of the second equation in (\ref{eq:eom_maxwell}) give
\begin{equation} \label{eq:q_uv}
    \partial_u Q = \partial_v Q = 0.
\end{equation}
In particular, for the matter model considered here, the Maxwell charge $Q$ is not dynamical, and we simply treat it as a constant of our spacetime.

We call the system (\ref{eq:raych_u})--(\ref{eq:wave_phi_v}) the \underline{E}instein-\underline{M}axwell \underline{s}calar \underline{f}ield system in \underline{s}pherical \underline{s}ymmetry (or EMSFSS for short). In the special case that $Q \equiv 0$, we instead denote the system simply as the \underline{E}instein-\underline{s}calar \underline{f}ield in \underline{s}pherical \underline{s}ymmetry (or ESFSS for short).


\begin{remark}
    \eqref{eq:q_uv} places topological restrictions on the spacetime $(M, g)$ when the Maxwell charge $Q$ is nontrivial. In particular, having $Q \not\equiv 0$ prohibits one-ended asymptotically flat data with a regular center $\Gamma$.

    To overcome this, one could generalize the model to the Einstein-Maxwell-\textit{charged} scalar field system, where the Maxwell field is nonlinearly coupled to the scalar field $\phi$. See for instance \cite{kommemi, Moortel18} -- note that these articles allow for the scalar field $\phi$ to be \textit{massive} as well as charged.

    Since the charge and mass of the scalar field are expected to play no role near the spacelike singularity $\mathcal{S}$, we expect that many of our results should generalize straightforwardly to the Einstein-Maxwell-charged scalar field model, though in order to simplify the exposition we do not pursue that here.
\end{remark}

\subsection{Initial data in trapped regions} \label{setup.data}

We pose characteristic initial data for $(r, \Omega^2, \phi)$ for the EMSFSS system in a \textit{trapped region}. More precisely, initial data is prescribed on two null hypersurfaces, an outgoing null hypersurface $C_0$ and an ingoing null hypersurface $\underline{C}_0$ defined as follows:
\begin{gather*}
    C_0 \coloneqq \{ (u_0, v): v_0 \leq v \leq v_1 \}, \qquad
    \underline{C}_0 \coloneqq \{ (u, v_0): u_0 \leq u \leq u_1 \}.
\end{gather*}

It is well-known that the system \eqref{eq:raych_u}--\eqref{eq:q_uv}, with sufficiently regular initial data for $(r, \Omega^2, \phi)$ on $C_0 \cup \underline{C}_0$ satisfying suitable constraints, admit a well-posed characteristic initial value problem, thereby producing a \textit{maximal future hyperbolic development} in a domain
\begin{equation*}
    \mathcal{D} \coloneqq \{ (u, v) \in [u_0, u_1] \times [ v_0, v_1] \subset \R^2: r(u, v), \Omega^2(u, v), \phi(u,v) \text{ bounded and regular} \}.
\end{equation*}

Though $\mathcal{D}$ is not necessarily the entire characteristic rectangle $[u_0, u_1] \times [v_0, v_1]$, the \textit{null condition} inherent in the EMSFSS system ensures that $\mathcal{D}$ contains at least an open neighborhood of the initial data surface $C_0 \cup \underline{C}_0$. Moreover, the results of \cite{kommemi} yield that any point on the boundary of $\mathcal{D}$ lying in the interior of this characteristic rectangle corresponds to an $r = 0$ singularity\footnote{Indeed, this holds for any \textit{strongly tame} matter system, as defined in \cite{kommemi}.}. See Figure \ref{fig:char_ivp1} for details.

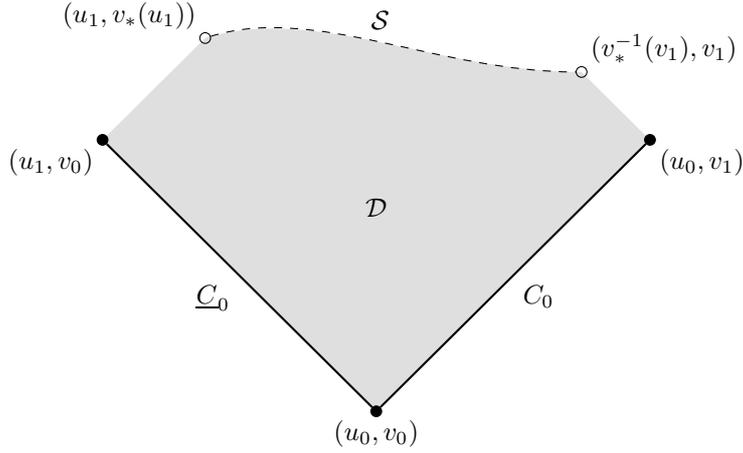
\begin{figure}[h]
    \centering
    \begin{tikzpicture}[scale=0.9]
        \path[fill=lightgray, opacity=0.5] (0, -4) -- (-4, 0) -- (-2.5, 1.5)
            .. controls (-0.9, 2) and (0.9, 1) .. (3, 1)
            -- (4, 0) -- (0, -4);

        \node (p) at (0, -4) [circle, draw, inner sep=0.5mm, fill=black] {};
        \node (r) at (4, 0) [circle, draw, inner sep=0.5mm, fill=black] {};
        \node (l) at (-4, 0) [circle, draw, inner sep=0.5mm, fill=black] {};
        \node (rs) at (3, 1) [circle, draw, inner sep=0.5mm] {};
        \node (ls) at (-2.5, 1.5) [circle, draw, inner sep=0.5mm] {};

        \node at (p) [below] {$(u_0, v_0)$};
        \node at (r) [below right] {$(u_0, v_1)$};
        \node at (l) [below left] {$(u_1, v_0)$};
        \node at (rs) [above right] {$(v_*^{-1}(v_1), v_1)$};
        \node at (ls) [above left] {$(u_1, v_*(u_1))$};
        \node at (0, -1) {$\mathcal{D}$};

        \draw [thick] (p) -- (r)
            node [midway, below right] {$C_0$};
        \draw [thick] (p) -- (l)
            node [midway, below left] {$\underline{C}_0$};
        \draw [dashed] (ls) .. controls (-0.9, 2) and (0.9, 1) .. (rs)
            node [midway, above=0.5mm] {$\mathcal{S}$};
    \end{tikzpicture}

    \captionsetup{justification = centering}
    \caption{The characteristic initial value problem considered in this article. Regular initial data is prescribed at $C_0 \cup \underline{C}_0$ such that $C_0 \cup \underline{C}_0$ is trapped, yielding a maximal future hyperbolic development $\mathcal{D}$ with a singular boundary $\mathcal{S}$ at which $r = 0$.}
    \label{fig:char_ivp1}
\end{figure}

Our analysis will be confined to the trapped region. As in \cite{Christodoulou_BV}, we take this to mean the region $\mathcal{T} \subset \mathcal{Q}$ with $\mathcal{T} \coloneqq \{ (u, v): \partial_u r(u, v) <0, \partial_v r (u,v) < 0 \}$. Due to the monotonicity of $\Omega^{-2} \partial_u r$ and $\Omega^{-2} \partial_v r$ from the Raychaudhuri equations (\ref{eq:raych_u}) and (\ref{eq:raych_v}), it suffices to have initial data in the trapped region, since
\begin{equation*}
    C_0 \cup \underline{C}_0 \subset \mathcal{T} \implies \mathcal{D} \subset \mathcal{T}.
\end{equation*}

We thus already assume that $C_0 \cup \underline{C}_0 \subset \mathcal{T}$, and we may fix the remaining gauge-freedom in our choice of double null coordinates by prescribing
\begin{equation} \tag{$uv-$gauge} \label{eq:gauge1}
    \partial_v r = - 1 \text{ on } C_0, \qquad \partial_u r = - 1 \text{ on } \underline{C}_0.
\end{equation}
Imposing (\ref{eq:gauge1}), the free data on $C_0 \cup \underline{C}_0$ will then be $\Omega^2 (u_0, v_0) \in \R^+$, $r (u_0, v_0) \in \R^+$, and the characteristic data for the scalar field, namely $\partial_v \phi$ on $C_0$ and $\partial_u \phi$ on $\underline{C}_0$ as well as $\phi(u_0, v_0)$. For our purposes, it will mostly suffice to take $\partial_v \phi|_{C_0}$ and $\partial_u \phi|_{\underline{C}_0}$ to be $C^1$.

It is often advantageous to consider quantities which are \textit{gauge-independent} with regard to the choice of $u$ and $v$. We may, for instance, consider the Hawking mass $m$, defined in (\ref{eq:hawkingmass}), in place of $\Omega^2$. We also make frequent use of the following gauge-independent vector fields:
\begin{equation} \label{eq:llbar}
    L = \frac{1}{- r \partial_v r} \partial_v, \qquad \underline{L} = \frac{1}{- r \partial_u r} \partial_u,
\end{equation}
which are normalized such that $g(L, \underline{L}) = - 2 (2 m r - r^2)^{-1}$. Observe that in (\ref{eq:gauge1}), one has
\begin{equation*}
    r L \phi = \partial_v \phi \text{ on } C_0, \qquad r \underline{L} \phi = \partial_u \phi \text{ on } \underline{C}_0.
\end{equation*}

\subsection{Strongly singular maximal future developments} \label{setup.sing}

We are interested in solutions of the characteristic initial value problem containing singularities, or more specifically spacelike singularities. By \cite{Christodoulou_formation, Christodoulou_BV, kommemi}, in particular the \textit{generalized extension principle} of \cite{kommemi}, such singularities in the ESFSS or EMSFSS system, at least those away from the center of symmetry $\Gamma$, are characterized 
 by the area-radius $r(u, v)$ tending to $0$. 

For this purpose, we focus on maximal globally hyperbolic future developments $\mathcal{D}$ possessing a non-empty singular boundary $\mathcal{S}$ at which $r(u, v) = 0$, which may be parameterized as $(u, v_*(u))$. 

In the ESFSS system (where $Q \equiv 0$), there is additional monotonicity in the system that simplifies the issue of finding $\mathcal{S}$. Note that (\ref{eq:wave_r_u}) and (\ref{eq:wave_r_v}) give:
\begin{equation*}
    \partial_u ( - r \partial_v r) = \partial_v (- r \partial_u r) = \frac{\Omega^2}{4} > 0.
\end{equation*}
Hence $- r \partial_u r$ and $- r \partial_v r$ are bounded below by their minimum value on $C_0 \cup \underline{C}_0$. Thus $r^2(u, v)$ is monotonically decreasing at a rate which is bounded below, in both the $u$ and $v$ directions, must therefore reach $0$ so long as the intervals $[u_0, u_1]$ and $[v_0, v_1]$ are large enough. We formalize this observation in the following lemma.

\begin{lemma} \label{lem:setup_esfss}
    Consider the ESFSS system with characteristic initial data prescribed on $C_0 \cup \underline{C}_0$ as in Section \ref{setup.data}. Then there exist constants $R_0, R_1 > 0$, depending only on data, such that for $(u, v) \in \mathcal{D}$,
    \begin{equation} \label{eq:esfss_cond}
        R_0 \leq - r \partial_u r (u, v) \leq R_1, \qquad R_0 \leq - r \partial_v r (u, v) \leq R_1.
    \end{equation}
    For the EMSFSS system with $Q \neq 0$, the upper bound of \eqref{eq:esfss_cond} remains true.

    If \eqref{eq:esfss_cond} holds and one supposes, moreover, that there exists $v_*(u_1) < v_1$ such that $(u_1, v) \in \mathcal{D}$ for $v_0 \leq v < v_*(u_1)$, but $r(u_0, v) \to 0$ as $v \to v_*(u_1)^-$. Then the region $\mathcal{D}$ has a singular boundary $\mathcal{S}$, parameterised by a continuous curve $v = v_*(u)$ for $v_*^{-1}(v_1) \leq u \leq u_1$, such that $r(u, v)$ extends continuously to $0$ on $\mathcal{S}$. The boundary $\mathcal{S}$ is \textit{spacelike} in the sense that $v_*(u)$ is strictly decreasing.
\end{lemma}

\begin{proof}
    When $Q \equiv 0$, the monotonicity of (\ref{eq:wave_r_u}) and (\ref{eq:wave_r_v}) mentioned above immediately implies the lower bound in (\ref{eq:esfss_cond}) for $R_0 = \min \{ \min_{(u_0, v) \in C_0} - r \partial_v r (u_0, v), \min_{(u, v_0) \in \underline{C}_0} - r \partial_u r (u, v_0) \} > 0$. Regardless of the value of $Q$, the upper bound is found by integrating (\ref{eq:wave_r_u}), (\ref{eq:wave_r_v}) and then applying (\ref{eq:raych_u}), (\ref{eq:raych_v}). For instance,
    \begin{align*}
        -r \partial_v r(u, v) 
        &\leq - r \partial_v r(u_0, v) + \int_{u_0}^u \frac{\Omega^2}{4}(\tilde{u}, v) \, d \tilde{u},
        \\[0.5em]
        &= - r \partial_v r(u_0, v) + \int_{r(u, v)}^{r(u_0, v)} \frac{\Omega^2}{- 4 \partial_u r} (\tilde{u}, v)\, d r (\tilde{u}, v),
        \\[0.5em]
        &\leq - r \partial_v r(u_0, v) + \int_{r(u, v)}^{r(u_0, v)} \frac{\Omega^2}{- 4 \partial_u r} (u_0, v) \, dr (\tilde{u}, v),
        \\[0.5em]
        &\leq \max_{(u_0, \tilde{v}) \in C_0} (- r \partial_v r (u_0, \tilde{v})) + r(u_0, v_0) \cdot \max_{(u_0, \tilde{v}) \in C_0} \frac{\Omega^2}{- 4 \partial_u r} (u_0, \tilde{v}).
    \end{align*}
    The inequality on the third line follows from the Raychaudhuri equation (\ref{eq:raych_u}), which implies that $\frac{\Omega^2}{- \partial_u r}$ is decreasing in $u$. A similar computation for $- r \partial_u r(u, v)$ then implies the upper bound in (\ref{eq:esfss_cond}).

    Assuming the existence of $v_*(u_1)$ as in the lemma, the parameterization of the singular boundary $\mathcal{S}$ as $(u, v_*(u))$ follows straightforwardly from \eqref{eq:esfss_cond}.
\end{proof}

\begin{remark}
    By the monotonicity $\partial_u \partial_v r^2 < 0$ in the case $Q \equiv 0$, a sufficient condition on data to guarantee the existence of such a $v_*(u_1) < v_1$ is the following inequality:
    \begin{equation} \label{eq:esfss_cond_ineq}
        r(u_0, v_0)^2 \geq r(u_1, v_0)^2 + r(u_0, v_1)^2.
    \end{equation}
    However this is not necessary, and need not restrict ourselves to scenarios where (\ref{eq:esfss_cond_ineq}) holds.
\end{remark}

We now consider the EMSFSS system with $Q \neq 0$. In this case, there is no monotonicity akin to $\partial_u \partial_v r^2 < 0$, and in order to proceed we instead make the \textit{a priori assumption} 
\begin{equation} \tag{QTS} \label{eq:cond_r}
    \inf_{(u, v) \in \mathcal{D}} \min \{ - r \partial_v r (u, v), - r \partial_u r (u, v) \} \geq R_0.
\end{equation}
Lemma~\ref{lem:setup_esfss} implies that we still have access to the full inequality \eqref{eq:esfss_cond} so \blue{long} as \eqref{eq:cond_r} holds.

When $Q \neq 0$, we also need to make a second a priori assumption: for some $\alpha \in (0, 1)$, we assume
\begin{equation} \tag{SKE} \label{eq:cond_phi}
    \inf_{(u, v) \in \mathcal{D}} \min \left \{ |r^2 L \phi (u, v)|^2, |r^2 \underline{L} \phi (u, v)|^2 \right \} \geq 1 + \alpha.
\end{equation}
Here the null vector fields $L$ and $\underline{L}$ are defined in (\ref{eq:llbar}). As discussed in Section~\ref{sub:intro_rough_1} the assumptions \eqref{eq:cond_r} and \eqref{eq:cond_phi} do not follow easily from any conditions on data, but nonetheless there is a non-empty set of spacetimes to which they apply. We are now ready to make the following definition.

\begin{definition}
    A \textit{strongly singular} future development $\mathcal{D}$ of (trapped) characteristic initial data on $C_0 \cup \underline{C}_0$ is a maximal future development $\mathcal{D}$ \blue{of characteristic initial data for the Einstein--Maxwell--scalar field system}, as described in Section \ref{setup.data}, with the two additional properties:
    \begin{enumerate}[(i)]
        \itemsep -0.2em
        \item
            There exists some continuous, strictly decreasing function $v_*(u)$, defined on $v_*^{-1}(v_1) \leq u \leq u_1$ for some $v_*^{-1}(v_1) \in (u_0, u_1)$, such that $\mathcal{S} = \{ (u, v_*(u)): v_*^{-1}(v_1) \leq u \leq u_1 \}$ is a future boundary of $\mathcal{D}$ such that the area-radius $r(u, v)$ extends continuously to $0$ at $\mathcal{S}$.
        \item
            Either $Q$ is identically zero, or the two \blue{additional conditions (\ref{eq:cond_r}) and (\ref{eq:cond_phi}) are satisfied within $\mathcal{D}$.}
    \end{enumerate}
\end{definition}


%% file: statement.tex
\section{Statements of the main Theorems} \label{statement}

\subsection{The H\"older regularity norms} \label{statement.norms}

We are now in a position to state the two main theorems. Before doing so, we recall what it means for a function to be $C^{k, \beta}$, and define our notion of logarithmically modified H\"older regularity, $C^{k, \beta, \log}$.

\begin{definition}
    For $\mathcal{D}$ bounded in $\R^2$ and $\beta \in (0, 1)$, we define the $C^{0,\beta}$ norm of a function $f: \mathcal{D} \to \R$ to be\footnote{As $\mathcal{D}$ is bounded in $\R^2$, it is straightforward to see that this is equivalent to the usual $C^{0,\beta}$ norm.}:
    \begin{equation} \label{eq:holder}
            \| f \|_{C^{0,\beta}} \coloneqq \sup_{(u, v) \in \mathcal{D}} |f(u, v)| 
            + \sup_{ \substack{
                \{u\} \times [\tilde{v}_0, \tilde{v}_1] \\ \subset \mathcal{D}
            }}\frac{|f(u, \tilde{v}_1) - f(u, \tilde{v}_0)|}{|\tilde{v}_1 -\tilde{v}_0|^{\beta}}
            + \sup_{ \substack{
                    [\tilde{u}_0, \tilde{u}_1] \times \{v\} \\ \subset \mathcal{D}
            }}\frac{|f(\tilde{u}_1, v) - f(\tilde{u}_0, v)|}{|\tilde{u}_1 -\tilde{u}_0|^{\beta}}.
    \end{equation}
    If $\|f\|_{C^{0, \beta}} < + \infty$ we say that $f$ is \textit{H\"older continuous} in $\mathcal{D}$, and can be defined up to and including the boundary $\partial D$ (in a unique H\"older continuous fashion).

    Similarly, define the logarithmically modified $C^{0, \beta, \log}$ norm of $f: \mathcal{D} \to \R$ to be
    \begin{multline} \label{eq:logholder}
        \| f \|_{C^{0,\beta, \log}} \coloneqq \sup_{(u, v) \in \mathcal{D}} |f(u, v)| 
        \\ + \sup_{ \substack{
            \{u\} \times [\tilde{v}_0, \tilde{v}_1] \\ \subset \mathcal{D}
        }}\frac{|f(u, \tilde{v}_1) - f(u, \tilde{v}_0)|}{|\tilde{v}_1 -\tilde{v}_0|^{\beta} (1 + |\log (\tilde{v}_1 - \tilde{v}_0)| )}
        + \sup_{ \substack{
                [\tilde{u}_0, \tilde{u}_1] \times \{v\} \\ \subset \mathcal{D}
        }}\frac{|f(\tilde{u}_1, v) - f(\tilde{u}_0, v)|}{|\tilde{u}_1 -\tilde{u}_0|^{\beta} (1 + |\log(\tilde{u}_1 -\tilde{u}_0 ) |)}.
    \end{multline}
    If $\|f\|_{C^{0, \beta, \log}} < + \infty$ we say that $f$ is \textit{log-H\"older continuous}. 

    Finally, we define $f$ to be $C^{k, \beta}$ (respectively $C^{k, \beta, \log}$) for some $k \in \mathbb{N}$ if $f$ is $C^k$ and its $k$th derivatives $\partial_u^{k_1} \partial_v^{k_2} f$ with $k_1 + k_2 = k$ are each $C^{0, \beta}$ (respectively $C^{0, \beta, \log})$.
\end{definition}

\begin{remark}
    Since $1 + \log x^{-1} \lesssim x^{-\gamma}$ for $x$ small, it holds that for $0 < \beta' < \beta < 1$, there exists $C(\beta, \beta', \mathcal{D})$ such that $\| f \|_{C^{0, \beta'}} \leq C \| f \|_{C^{0, \beta, \log}}$. Hence log-H\"older continuity implies H\"older continuity for any smaller exponent.
\end{remark}

%
%

\subsection{Theorems~\ref{thm:esfss} and \ref{thm:emsfss}: Asymptotics at singularity in double null coordinates} \label{statement.doublenull}

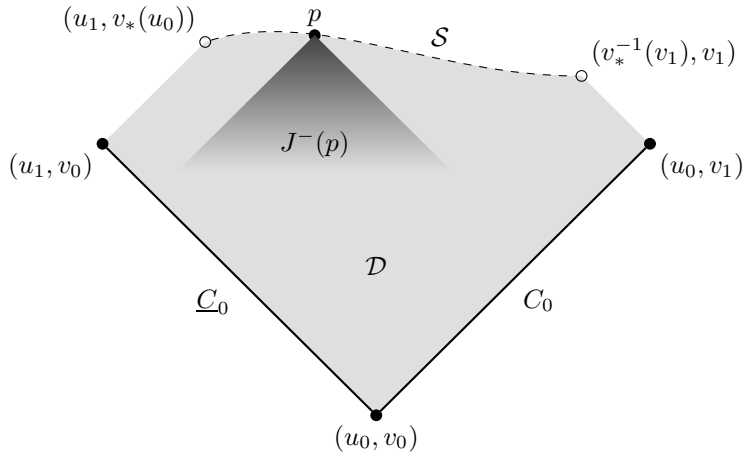
\begin{figure}[h]
    \centering
    \begin{tikzpicture}[scale=0.9]
        \path[fill=lightgray, opacity=0.5] (0, -4) -- (-4, 0) -- (-2.5, 1.5)
            .. controls (-0.9, 2) and (0.9, 1) .. (3, 1)
            -- (4, 0) -- (0, -4);

        \node (p) at (0, -4) [circle, draw, inner sep=0.5mm, fill=black] {};
        \node (r) at (4, 0) [circle, draw, inner sep=0.5mm, fill=black] {};
        \node (l) at (-4, 0) [circle, draw, inner sep=0.5mm, fill=black] {};
        \node (rs) at (3, 1) [circle, draw, inner sep=0.5mm] {};
        \node (ls) at (-2.5, 1.5) [circle, draw, inner sep=0.5mm] {};

        \node (sing) at (-0.9, 1.6) [circle, draw, inner sep=0.5mm, fill=black] {};

        \node at (p) [below] {$(u_0, v_0)$};
        \node at (r) [below right] {$(u_0, v_1)$};
        \node at (l) [below left] {$(u_1, v_0)$};
        \node at (rs) [above right] {$(v_*^{-1}(v_1), v_1)$};
        \node at (ls) [above left] {$(u_1, v_*(u_0))$};
        \node at (0, -1.8) {$\mathcal{D}$};
        \node at (sing) [above] {$p$};

        \shade [top color=darkgray, bottom color=lightgray!50!white] (sing) -- ++(-2, -2) -- ++(4, 0) -- ++(-2, 2);
        \node at (-0.9, 0.0) {$J^-(p)$};

        \draw [thick] (p) -- (r)
            node [midway, below right] {$C_0$};
        \draw [thick] (p) -- (l)
            node [midway, below left] {$\underline{C}_0$};
        \draw [dashed] (ls) .. controls (-0.9, 2) and (0.9, 1) .. (rs)
            node [pos=0.66, above=0.5mm] {$\mathcal{S}$};
    \end{tikzpicture}

    \captionsetup{justification = centering}
    \caption{Theorems \ref{thm:esfss} and \ref{thm:emsfss} give precise asymptotics for the dynamical quantities $r(u, v)$, $\Omega^2(u, v)$ and $\phi(u, v)$ inside the strongly singular maximal future development $\mathcal{D}$. Localizing to the TIP $J^-(p)$, one may use the values of $\Psi$, $\Xi$ and $\mathfrak{M}$ (defined in Theorems \ref{thm:esfss} and \ref{thm:emsfss}) at the singular point $p \in \mathcal{S}$ to describe the asymptotics in $J^-(p)$.}
    \label{fig:char_ivp_esfss}
\end{figure}

We now state Theorems~\ref{thm:esfss} and \ref{thm:emsfss}, which are the precise versions of Theorem~\ref{roughthm:asymp} for the ESFSS and the EMSFSS system respectively. We refer to Figure~\ref{fig:char_ivp_esfss} for the Penrose diagram representing both results.

\begin{theorem} \label{thm:esfss}
    Consider characteristic initial data for the ESFSS system on the bifurcate null hypersurface $C_0 \cup \underline{C}_0$, obeying the gauge condition (\ref{eq:gauge1}). Suppose that the maximal development $\mathcal{D}$ is strongly singular, as described in Section \ref{setup.sing}. Then:
    \begin{enumerate}[(I)]
        \item \label{item:esfss_thm_i}
            \underline{Asymptotics for $r$:} 
            The squared area-radius $r^2(u, v)$ is a $C^{1, 1/2}$ function in $\mathcal{D}$. There exist constants $R_0, R_1 > 0$, such that for $(u, v) \in \mathcal{D}$,
            \begin{equation} \label{eq:cond_pre_r}
                R_0 \leq - r \partial_u r (u, v) \leq R_1, \qquad R_0 \leq -r \partial_v r (u, v) \leq R_1.
            \end{equation}
            Consequently, we identify the $r = 0$ singularity as a $C^{1, 1/2}$ curve, denoted $\mathcal{S}$, which we parameterize as $\mathcal{S} = \{ (u, v_*(u)): u_0 \leq u \leq v_*^{-1}(v_1) \}$. We define $\bar{\mathcal{D}} = \mathcal{D} \cup \mathcal{S}$, and by $C_u(p), C_v(p)$ the continuous limits of $-r \partial_u r, - r \partial_v r$ respectively at $p \in \mathcal{S}$.
        \item \label{item:esfss_thm_ii}
            \underline{Asymptotics for $\phi$:}
           There exists a $C^{0, 1/2}$ function $\Psi(u, v)$ and a $C^{0, 1/2, \log}$ function $\Xi(u, v)$ such that for $(u, v) \in \mathcal{D}$,
            \begin{equation} \label{eq:esfss_thm_phi}
                \phi(u, v) = \Psi(u, v) \cdot \log \left( \frac{r_0}{r} \right) + \Xi(u, v).
            \end{equation}
            We may extend $\Psi$ and $\Xi$ continuously to the singular boundary $\mathcal{S}$, and denote by $\Psi_{\infty}(p)$ and $\Xi_{\infty}(p)$ their values at $p \in \mathcal{S}$.
        \item \label{item:esfss_thm_iii}
            \underline{Asymptotics for $\Omega^2$:} There exists a $C^{0, 1/2, \log}$ function $\mathfrak{M}(u, v)$, continuously extendible as a positive function on $\bar{\mathcal{D}}$, such that we have the following asymptotics for the gauge-invariant quantity $4 \Omega^{-2} \partial_u r \partial_v r$:
            \begin{equation} \label{eq:esfss_thm_lapse}
                4 \Omega^{-2} \partial_u r \partial_v r (u, v) = \mathfrak{M}(u, v) \cdot \left( \frac{r}{r_0} \right)^{-(\Psi^2 + 1)}.
            \end{equation}
            We denote by $\mathfrak{M}_{\infty}(p)$ the value of $\mathfrak{M}$ at $p \in \mathcal{S}$.
            
        \item \label{item:esfss_thm_iv}
            \underline{Estimates localized at a TIP}:
            Let $p = (u_p, v_*(u_p)) \in \mathcal{S}$, and consider its causal past $J^-(p) = \{ (u, v) \in \mathcal{D} : u \leq u_p, v \leq v_*(u_p) \}$. Then there exists a constant $D > 0$ such that for $(u, v) \in J^-(p)$:
            \begin{gather} \label{eq:esfss_thm_tip_r_u}
                |- r \partial_u r (u, v) - C_u(p) | \leq D \cdot r, 
                \\[0.5em] \label{eq:esfss_thm_tip_r_v}
                |- r \partial_v r (u, v) - C_v (p) | \leq D \cdot r,
                \\[0.5em] \label{eq:esfss_thm_tip_phi_u}
                |r^2 \partial_u \phi (u, v) - C_u(p) \cdot \Psi_{\infty} (p) | \leq D \cdot r,
                \\[0.5em] \label{eq:esfss_thm_tip_phi_v}
                |r^2 \partial_v \phi (u, v) - C_v(p) \cdot \Psi_{\infty} (p) | \leq D \cdot r,
                \\[0.5em] \label{eq:esfss_thm_tip_phi}
                \left| \phi(u, v) - \Psi_{\infty}(p) \cdot \log \left( \frac{r_0}{r} \right) - \Xi_{\infty} (p)\right| \leq D \cdot r \log r^{-1},
                \\[0.5em] \label{eq:esfss_thm_tip_lapse}
                \left| \Omega^2 - \frac{4 C_u(p) C_v(p)}{r_0^2} \cdot \mathfrak{M}_{\infty}^{-1}(p) \cdot \left (\frac{r}{r_0}\right)^{\Psi_{\infty}(p)^2 - 1} \right| \leq D \cdot r^{\Psi_{\infty}(p)^2} \log r^{-1}.
            \end{gather}
        \item \label{item:esfss_thm_v}
            \underline{Precise blow-up estimates:}
            For $p \in \mathcal{S}$ and $J^-(p)$ as in (\ref{item:esfss_thm_iv}), we have the following blow-up rates for the Hawking mass and the Kretschmann scalar:
            \begin{gather} \label{eq:esfss_thm_tip_hawkingmass}
                \left| \, 2m(u, v) - \mathfrak{M}_{\infty}(p) r_0 \cdot \left( \frac{r_0}{r} \right)^{\Psi_{\infty}^2} \right| \leq D \cdot r^{- \Psi_{\infty}(p)^2 + 1} \log r^{-1},
                \\[0.5em] \label{eq:esfss_thm_tip_kretschmann}
                \left| \, \mathrm{Riem}_{\alpha \beta \gamma \delta} \mathrm{Riem}^{\alpha \beta \gamma \delta} - \mathfrak{M}_{\infty}(p)^2 \cdot \left( \frac{r_0}{r} \right)^{2 (\Psi_{\infty}(p)^2 + 3)} \cdot \frac{4 ( 3 - 2 \Psi_{\infty}(p)^2 + 2 \Psi_{\infty}(p)^4) }{r_0^4} \right| \hspace{20em} \nonumber \\[-0.1em] \hspace{25em}\leq D \cdot r^{- 2 \Psi_{\infty}(p)^2 - 5} \log r^{-1}.
            \end{gather}
            For the scalar field, we have the following gauge-invariant blow-up rate:
            \begin{equation} \label{eq:esfss_thm_tip_phi_blow}
                \left| \nabla_{\alpha} \phi \nabla^{\alpha} \phi + \frac{\mathfrak{M}_{\infty} (p) \Psi_{\infty}(p)^2 }{r_0^2} \cdot \left( \frac{r_0}{r} \right)^{\Psi_{\infty}(p)^2 + 3}\right| \leq D \cdot r^{- \Psi_{\infty}(p)^2 - 2} \log r^{-1}.
            \end{equation}
    \end{enumerate}
\end{theorem}

\begin{theorem} \label{thm:emsfss}
    Consider characteristic initial data for the ESFSS system on the bifurcate null hypersurface $C_0 \cup \underline{C}_0$, obeying the gauge condition (\ref{eq:gauge1}). Suppose a priori that the maximal development $\mathcal{D}$ produced by this data terminates at an $r = 0$ spacelike singularity, and moreover that (\ref{eq:cond_r}) and (\ref{eq:cond_phi}) hold in $\mathcal{D}$. Then for $\beta = \alpha / 2$, 
    \begin{enumerate}[(I)]
        \item \label{item:emsfss_thm_i}
            \underline{Asymptotics for $r$:} 
            The squared area-radius $r^2(u, v)$ is a $C^{1, \beta}$ function in $\mathcal{D}$. Building upon (\ref{eq:cond_r}), we find an additional upper bound: for $(u, v) \in \mathcal{D}$, there exists some $R_1 > R_0$ such that:
            \begin{equation} \label{eq:cond_pre_r_}
                R_0 \leq - r \partial_u r (u, v) \leq R_1, \qquad R_0 \leq -r \partial_v r (u, v) \leq R_1.
            \end{equation}
            Consequently, we can identify the $r = 0$ singularity as a $C^{1, \beta}$ curve, denoted $\mathcal{S}$, which we parameterize by $\mathcal{S} = \{ (u, v_*(u)): u_0 \leq u \leq v_*^{-1}(v_1) \}$. We then define $\bar{\mathcal{D}}$, $C_u$ and $C_v$ as in (\ref{item:emsfss_thm_i}) of Theorem \ref{thm:esfss}.
        \item \label{item:emsfss_thm_ii}
            \underline{Asymptotics for $\phi$:}
            There exists a $C^{0, \beta}$ function $\Psi(u, v)$ and a $C^{0, \beta, \log}$ function $\Xi(u, v)$ such that for $(u, v) \in \mathcal{D}$, \blue{$\Psi^2(u, v) \geq 1 + \alpha$ and}
            \begin{equation} \label{eq:emsfss_thm_phi}
                \phi(u, v) = \Psi(u, v) \cdot \log \left( \frac{r_0}{r} \right) + \Xi(u, v).
            \end{equation}
            We may extend $\Psi$ and $\Xi$ continuously to the singular boundary $\mathcal{S}$, and denote by $\Psi_{\infty}(p)$ and $\Xi_{\infty}(p)$ their values at $p \in \mathcal{S}$.
        \item \label{item:emsfss_thm_iii}
            \underline{Asymptotics for $\Omega^2$:} There exists a $C^{0, \beta, \log}$ function $\mathfrak{M}(u, v)$, continuously extendible as a positive function on $\bar{\mathcal{D}}$, such that we have the following asymptotics for the gauge-invariant quantity $4 \Omega^{-2} \partial_u r \partial_v r$:
            \begin{equation} \label{eq:emsfss_thm_lapse}
                4 \Omega^{-2} \partial_u r \partial_v r (u, v) = \mathfrak{M}(u, v) \cdot \left( \frac{r}{r_0} \right)^{-(\Psi^2 + 1)}.
            \end{equation}
            We denote by $\mathfrak{M}_{\infty}(p)$ the value of $\mathfrak{M}$ at $p \in \mathcal{S}$.
            
        \item \label{item:emsfss_thm_iv}
            \underline{Estimates localized at a TIP}:
            Let $p = (u_p, v_*(u_p)) \in \mathcal{S}$, and consider its causal past $J^-(p) = \{ (u, v) \in \mathcal{D} : u \leq u_p, v \leq v_*(u_p) \}$. Then there exists a constant $D > 0$ such that for $(u, v) \in J^-(p)$:
            \begin{gather} \label{eq:emsfss_thm_tip_r_u}
                |- r \partial_u r (u, v) - C_u(p) | \leq D \cdot r^{\alpha}, 
                \\[0.5em] \label{eq:emsfss_thm_tip_r_v}
                |- r \partial_v r (u, v) - C_v (p) | \leq D \cdot r^{\alpha},
                \\[0.5em] \label{eq:emsfss_thm_tip_phi_u}
                |r^2 \partial_u \phi (u, v) - C_u(p) \cdot \Psi_{\infty} (p) | \leq D \cdot r^{\alpha},
                \\[0.5em] \label{eq:emsfss_thm_tip_phi_v}
                |r^2 \partial_v \phi (u, v) - C_v(p) \cdot \Psi_{\infty} (p) | \leq D \cdot r^{\alpha},
                \\[0.5em] \label{eq:emsfss_thm_tip_phi}
                \left| \phi(u, v) - \Psi_{\infty}(p) \cdot \log \left( \frac{r_0}{r} \right) - \Xi_{\infty} (p)\right| \leq D \cdot r^{\alpha} \log r^{-1},
                \\[0.5em] \label{eq:emsfss_thm_tip_lapse}
                \left| \Omega^2 - \frac{4 C_u(p) C_v(p)}{r_0^2} \cdot \mathfrak{M}_{\infty}^{-1}(p) \cdot \left (\frac{r}{r_0}\right)^{\Psi_{\infty}(p)^2 - 1} \right| \leq D \cdot r^{\Psi_{\infty}(p)^2 - 1 + \alpha} \log r^{-1}.
            \end{gather}
        \item \label{item:emsfss_thm_v}
            \underline{Precise blow-up estimates:}
            For $p \in \mathcal{S}$ and $J^-(p)$ as in (\ref{item:esfss_thm_iv}), we have the following blow-up rates for the Hawking mass and the Kretschmann scalar:
            \begin{gather} \label{eq:emsfss_thm_tip_hawkingmass}
                \left| \, 2m(u, v) - \frac{\mathfrak{M}_{\infty}(p) r_0}{2} \cdot \left( \frac{r_0}{r} \right)^{\Psi_{\infty}^2} \right| \leq D \cdot r^{- \Psi_{\infty}(p)^2 + \alpha} \log r^{-1},
                \\[0.5em] \label{eq:emsfss_thm_tip_kretschmann}
                \left| \, \mathrm{Riem}_{\alpha \beta \gamma \delta} \mathrm{Riem}^{\alpha \beta \gamma \delta} - \mathfrak{M}_{\infty}(p)^2 \cdot \left( \frac{r_0}{r} \right)^{2 (\Psi_{\infty}(p)^2 + 3)} \cdot \frac{4 ( 3 - 2 \Psi_{\infty}(p)^2 + 2 \Psi_{\infty}(p)^4) }{r_0^4} \right| \hspace{20em} \nonumber \\[-0.1em] \hspace{25em}\leq D \cdot r^{- 2 \Psi_{\infty}(p)^2 - 6 + \alpha} \log r^{-1}.
            \end{gather}
            We treat also the matter fields. For the scalar field $\phi$, we have the following gauge-invariant blow-up rate:
            \begin{equation} \label{eq:emsfss_thm_tip_phi_blow}
                \left| \nabla_{\alpha} \phi \nabla^{\alpha} \phi + \frac{\mathfrak{M}_{\infty} (p) \Psi_{\infty}(p)^2 }{r_0^2} \cdot \left( \frac{r_0}{r} \right)^{\Psi_{\infty}(p)^2 + 3}\right| \leq D \cdot r^{- \Psi_{\infty}(p)^2 - 3 + \alpha} \log r^{-1},
            \end{equation}
            while the electromagnetic $2$-form $F$ satisfies:
            \begin{equation} \label{eq:emsfss_thm_tip_f}
                F_{\mu \nu} F^{\mu \nu} = - \frac{2 Q^2}{r^4}.
            \end{equation}
    \end{enumerate}
\end{theorem}

\subsection{Corollary~\ref{cor:bkl}: A BKL-like expansion} \label{statement.bkl}

We now state Corollary~\ref{cor:bkl}, which interprets the strongly singular spacetimes of Theorem~\ref{thm:esfss} and Theorem~\ref{thm:emsfss} in the context of the BKL ansatz. In particular, we define a foliation using the level sets of a time function $\tau$, given by \eqref{eq:tau}, and show that this foliation obeys several BKL-like properties. We refer the reader to \cite{RingstromSilentGeometry} to justify how Parts~\ref{item:bkl_1} to \ref{item:bkl_5} indeed detail a correspondence between our foliation and the BKL picture.

Note that throughout Corollary~\ref{cor:bkl}, we make reference to error terms of the form $O(r^{\alpha} \log r^{-1})$. To see that these are genuine lower order terms, note that via \eqref{eq:tau}, one has $O(r^{\alpha} \log r^{-1} ) = O( \tau^{\frac{2 \alpha}{\Psi^2 + 3}} \log \tau^{-1} )$. Hence these terms may indeed be interpreted as error terms in the BKL picture.

\begin{corollary} \label{cor:bkl}
    Let $(\mathcal{M}, g)$ be a smooth, spherically symmetric, strongly singular spacetime solving the EMSFSS system, and described by either Theorem~\ref{thm:esfss} or Theorem~\ref{thm:emsfss}. Let $\Psi$, $\Xi$ and $\mathfrak{M}$ be described by (\ref{item:esfss_thm_ii}) and (\ref{item:esfss_thm_iii}) in these theorems. Now define the function $\tau: \mathcal{D} \to \R$ by
    \begin{equation} \label{eq:tau}
        \tau \coloneqq \frac{2 r_0}{\Psi^2 + 3} \cdot \mathfrak{M}^{-1/2} \cdot \left( \frac{r}{r_0} \right)^{\frac{\Psi^2 + 3}{2}}.
    \end{equation}
    Now, foliate the spacetime $(\mathcal{M}, g)$ by hypersurfaces of constant $\tau$, denoted $S_{\tau}$. \blue{Let $e_0 = n^{-1} \nabla \tau$ be the (future-directed) normal to the foliation, where $n = (-g(\nabla \tau, \nabla \tau))^{1/2}$ is the ADM lapse, and let $k_{ij}$ be the associated second fundamental form, with $k(X, Y) = g(\nabla_X e_0, Y)$.} Then, there exists some $\tau_0$ such that:
    \begin{enumerate}[1.]
        \item \label{item:bkl_1}
            The spacetime $\mathcal{M} \cap \{ \tau \leq \tau_0 \}$ is globally hyperbolic and foliated by the spacelike hypersurfaces $S_{\tau}$. Furthermore, the time-function $\tau$ is \underline{asymptotically normalized} in the sense that
            \begin{equation} \label{eq:propertime}
                \left| g( \nabla \tau, \nabla \tau ) + 1 \right| \leq D \, r^{\alpha} \log r^{-1}.
            \end{equation}

        \item \label{item:bkl_2}
            The foliation is also \underline{asymptotically CMC (constant mean curvature)} in the sense that the mean curvature \blue{$\tr k$ of the leaves $S_{\tau}$ of} the foliation obeys
            \begin{equation} \label{eq:cmc}
                \left| \tau \, \blue{\tr k \, +}\, 1 \right| \leq D \, r^{\alpha} \log r^{-1}.
            \end{equation}

        \item \label{item:bkl_3}
            Consider an orthonormal frame as follows: let $e_0 = n^{-1} \nabla \tau$ \blue{be as above}, let $e_2 = \frac{1}{r} \partial_{\theta}$ and $e_3 = \frac{1}{r \sin \theta}\partial_{\blue{\varphi}}$ be tangent to the $2$-spheres of symmetry, and finally choose $e_1$ to complete the orthonormal basis with direction chosen such that $g(e_1, \partial_v )< 0$. 

            Then the frame is approximately parallel transported in \blue{the sense that there exist functions $c_i: \mathcal{D} \to \R$, with $\tau c_i$ vanishing at the rate $|\tau c_i| \leq D \, r^{\alpha} \log r^{-1}$, such that}
            \begin{equation} \label{eq:parallel}
                \nabla_{e_0} e_i = c_i e_0. 
            \end{equation}

        \item \label{item:bkl_4}
            For $k_{ij}$ the second fundamental form \blue{as defined above}, let \blue{$\mathcal{K}_i^{\phantom{i}j} \coloneqq (\tr k)^{-1} k_i^{\phantom{i}j} $} denote the respective expansion-normalized Weingarten map. Then the eigenvectors of $\mathcal{K}_i^{\phantom{i}j}$ are exactly the frame vectors $e_1$, $e_2$, $e_3$, and their respective eigenvalues $p_1$, $p_2$ and $p_3$ satisfy
            \begin{equation} \label{eq:weingarten}
                p_1 = \frac{\Psi^2 - 1}{\Psi^2 + 3} + O(r^{\alpha}\log r^{-1}), \quad p_2 = p_3 = \frac{2}{\Psi^2 + 3} + O(r^{\alpha} \log r^{-1}).
            \end{equation}

        \item \label{item:bkl_5}
            Finally, the scalar field has the following logarithmic behavior:
            \begin{equation} \label{eq:pphi}
                \phi = \frac{2 \Psi}{\Psi^2 + 3} \cdot \log \tau^{-1} + O(1), \qquad - \tau e_0 \phi = \frac{2 \Psi}{\Psi^2 + 3} + O(r^{\alpha} \log r^{-1}).
            \end{equation}
    \end{enumerate}
\end{corollary}


%% file: upperbounds.tex
\section{Upper bounds} \label{upperbounds}

In this section, we recover the upper bounds of \cite{AnZhang}, and also generalize their results to the EMSFSS model with $Q \neq 0$. Throughout Section~\ref{upperbounds}, $D$ will represent a constant depending on the initial data at $C_0 \cup \underline{C}_0$, as well as the constants $R_0, \alpha > 0$ in (\ref{eq:cond_r}), (\ref{eq:cond_phi}) in the case $Q \neq 0$.

\subsection{Initial bounds on \texorpdfstring{$\Omega^2$}{Ω2} using Raychaudhuri} \label{upperbounds.lapse}

We begin with an upper bound for $\Omega^2$ that will be crucial for integrability purposes. Note that at this point, we must already assume (\ref{eq:cond_r}) and (\ref{eq:cond_phi}) when $Q \neq 0$. One may compare this to Section 3 of \cite{AnZhang}.

\begin{lemma} \label{lem:upperbounds_lapse}
    For $(u, v) \in \mathcal{D}$, one has the following upper bounds for $\Omega^2$ and related quantities:
    \begin{enumerate}[(A)]
        \item
            For the ESFSS system, one has
            \begin{equation} \label{eq:upperbound_lapse_uv}
                \max \left \{ \frac{\Omega^2}{- \partial_u r}(u, v), \frac{\Omega^2}{- \partial_v r}(u, v) \right \} \leq D
            \end{equation}
            as well as
            \begin{equation} \label{eq:upperbound_lapse}
                \Omega^2(u, v) \leq \frac{D}{r(u, v)}.
            \end{equation}
        \item \label{lapse.maxwell}
            Now consider the EMSFSS system with $Q \neq 0$. Assuming (\ref{eq:cond_phi}), one has the stronger bound
            \begin{equation} \label{eq:upperbound_lapse_q_uv}
                \max \left \{ \frac{\Omega^2}{- \partial_u r}(u, v), \frac{\Omega^2}{- \partial_v r}(u, v) \right \} \leq D \, r(u, v)^{1 + \alpha}.
            \end{equation}
            Assuming also (\ref{eq:cond_r}), one further finds, with $r_0 = r(u_0, v_0)$:
            \begin{equation} \label{eq:upperbound_lapse_q}
                \Omega^2 (u, v) \leq D \, r(u, v)^{\alpha} \leq D \, r_0^{\alpha}.
            \end{equation}
    \end{enumerate}
\end{lemma}

\begin{proof}
    (\ref{eq:upperbound_lapse_uv}) follows trivially from the Raychaudhuri equations (\ref{eq:raych_u}), (\ref{eq:raych_v}). Indeed, these imply that $- \Omega^{-2} \partial_u r$ and $- \Omega^{-2} \partial_v r$ are nondecreasing in $u$ and $v$ respectively, so the quantities on the left hand side of (\ref{eq:upperbound_lapse_uv}) are bounded by their respective maxima on $C_0 \cup \underline{C}_0$. To get (\ref{eq:upperbound_lapse}), simply combine (\ref{eq:upperbound_lapse_uv}) with the lower bounds on $-r \partial_u r, -r \partial_v r$ from Lemma~\ref{lem:setup_esfss}.

    Moving on to the case (\ref{lapse.maxwell}) with $Q \neq 0$, we recast the Raychaudhuri equations as transport equations for the following logarithmic quantities:
    \begin{equation} \label{eq:raych_u_log}
        \partial_u \log \left( \frac{\Omega^2}{ - \partial_u r} \right) = \frac{\partial_u r}{r} \cdot \left( \frac{r \partial_u \phi}{\partial_u r} \right)^2 = \frac{\partial_u r}{r} \cdot (r^2 \underline{L} \phi)^2.
    \end{equation}
    \begin{equation} \label{eq:raych_v_log}
        \partial_v \log \left( \frac{\Omega^2}{ - \partial_v r} \right) = \frac{\partial_v r}{r} \cdot \left( \frac{r \partial_v \phi}{\partial_v r} \right)^2 = \frac{\partial_v r}{r} \cdot (r^2 L \phi)^2.
    \end{equation}
    
    Now assume the lower bound (\ref{eq:cond_phi}). Since $\partial_u r, \partial_v r < 0$, (\ref{eq:cond_phi}) applied to (\ref{eq:raych_u_log}) and (\ref{eq:raych_v_log}) yields
    \begin{equation*}
        \partial_u \log \left( \frac{\Omega^2}{ - \partial_u r} \right) \leq \partial_u \log r \cdot (1 + \alpha),
    \end{equation*}
    \begin{equation*}
        \partial_v \log \left( \frac{\Omega^2}{ - \partial_v r} \right) \leq \partial_v \log r \cdot (1 + \alpha).
    \end{equation*}

    From this, the expression $\frac{\Omega^2}{- \partial_u r} \cdot r^{-(1 + \alpha)}$ is also monotonically decreasing in $u$. This observation, alongside the analogous observation with $v$ in place of $u$, asserts that
    \begin{equation*}
        \max \left \{ \frac{\Omega^2}{- \partial_u r}(u, v) \cdot r(u, v)^{- (1+ \alpha)} , \frac{\Omega^2}{- \partial_v r}(u, v) \cdot r(u, v)^{-(1 + \alpha)} \right \} \leq D.
    \end{equation*}
    This is exactly (\ref{eq:upperbound_lapse_q_uv}), and the final estimate (\ref{eq:upperbound_lapse_q}) follows from combining (\ref{eq:upperbound_lapse_q_uv}) and the lower bound (\ref{eq:cond_r}) on $-r \partial_u r$ and $ - r \partial_v r$.
\end{proof}

\subsection{Upper bounds for \texorpdfstring{$r^2 L \phi$}{r2L𝞍} and \texorpdfstring{$r^2 {\protect \underline{L}} \phi$}{r2L͟𝞍}} \label{upperbounds.phi}

We now apply the upper bounds on $\Omega^2$ found in Section \ref{upperbounds.lapse} to bound appropriately weighted derivatives of the scalar field $\phi$. We use a similar Gr\"onwall argument to that appearing in Section 6 of \cite{AnZhang}, though we instead opt to estimate the gauge-independent quantities $L \phi$ and $\underline{L} \phi$.

As in \cite{AnZhang}, we foliate our spacetime $\mathcal{D}$ using spacelike constant-$r$ hypersurfaces. Define, for $\tilde{r} \leq r_0 = r(u_0, v_0)$, the hypersurface
\begin{equation*}
    \Sigma_{\tilde{r}} \coloneqq \{ (u, v) \in \mathcal{D}: r(u,v) = \tilde{r} \}.
\end{equation*}

\begin{figure}[h]
    \centering
    \begin{tikzpicture}[scale=0.9]
        \path[fill=lightgray, opacity=0.5] (0, -4) -- (-4, 0) -- (-2.5, 1.5)
            .. controls (-0.9, 2) and (0.9, 1) .. (3, 1)
            -- (4, 0) -- (0, -4);

        \node (p) at (0, -4) [circle, draw, inner sep=0.5mm, fill=black] {};
        \node (r) at (4, 0) [circle, draw, inner sep=0.5mm, fill=black] {};
        \node (l) at (-4, 0) [circle, draw, inner sep=0.5mm, fill=black] {};
        \node (rs) at (3, 1) [circle, draw, inner sep=0.5mm] {};
        \node (ls) at (-2.5, 1.5) [circle, draw, inner sep=0.5mm] {};

        \node at (p) [below] {$(u_0, v_0)$};
        \node at (r) [below right] {$(u_0, v_1)$};
        \node at (l) [below left] {$(u_1, v_0)$};

        \draw [thick] (p) -- (r)
            node [midway, below right] {$C_0$};
        \draw [thick] (p) -- (l)
            node [midway, below left] {$\underline{C}_0$};
        \draw [dashed] (ls) .. controls (-0.9, 2) and (0.9, 1) .. (rs)
            node [midway, above=0.5mm] {$r = 0$};
        \draw [thick, dotted] (-3, 1) .. controls (-1.5, 1) and (1.5, 0) .. (3.5, 0.5);
        \draw [thick, dotted] (-4, 0) .. controls (-1.7, 0) and (1.7, -1.0) .. (4, 0)
            node [midway, below] {\footnotesize $\Sigma_{\tilde{r}}$};
        \draw [thick, dotted] (-3, -1) .. controls (1, -1.3) .. (3, -1);
        \draw [thick, dotted] (-2, -2) .. controls (0.5, -2.2) .. (2, -2);
        \draw [thick, dotted] (-1, -3) .. controls (0.2, -3.1) .. (1, -3);
    \end{tikzpicture}

    \captionsetup{justification = centering}
    \caption{Foliation of $\mathcal{D}$ by constant-$r$ hypersurfaces $\Sigma_r$, used in a Gr\"onwall argument}
    \label{fig:char_ivp_foliation}
\end{figure}
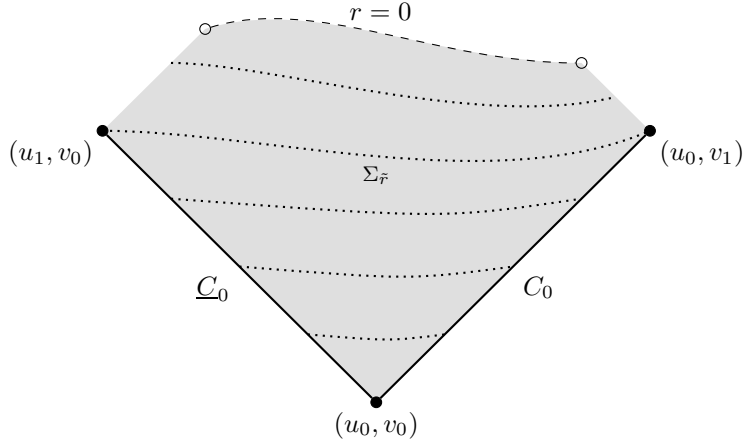

We next use the EMSFSS system to determine null transport equations for the weighted quantities $r L \phi$ and $r \underline{L} \phi$.
From the equations (\ref{eq:wave_phi_u}) and (\ref{eq:wave_r_u}), one finds the following propagation equations:
\begin{equation} \label{eq:wave_l_u}
    \partial_u ( r L \phi ) = \frac{ \partial_u \phi} { r } - \frac{\Omega^2}{-4 r \partial_v r} \left( 1 - \frac{Q^2}{r^2} \right) \cdot r L \phi,
\end{equation}
\begin{equation} \label{eq:wave_l_v}
    \partial_v ( r \underline{L} \phi ) = \frac{ \partial_v \phi} { r } - \frac{\Omega^2}{-4 r \partial_u r} \left( 1 - \frac{Q^2}{r^2} \right) \cdot r \underline{L} \phi.
\end{equation}
The following proposition uses these to produce the required upper bounds for $r^2 L \phi$ and $r^2 \underline{L} \phi$.

\begin{proposition} \label{prop:upperbounds_phi}
    Let $(u, v) \in \mathcal{D}$ for either (A) the ESFSS system with general initial data, or (B) the EMSFSS system supplemented with (\ref{eq:cond_r}) and(\ref{eq:cond_phi}), one has
    \begin{equation} \label{eq:upperbounds_phi}
        |r^2 L \phi (u, v)| \leq D, \qquad |r^2 \underline{L} \phi (u, v)| \leq D.
    \end{equation}
\end{proposition}

\begin{proof}
    Define the function $A(r)$ to be 
    \begin{equation*}
        A(\tilde{r}) = \sup_{\tilde{r} \leq \hat{r} \leq r_0} \max \, \left \{ \max_{(u, v) \in \Sigma_{\hat{r}}} |r L \phi(u, v) |, \max_{(u, v) \in \Sigma_{\hat{r}}} |r \underline{L} \phi (u, v)| \, \right \}.
    \end{equation*}
    The objective is to show that $A(r) \leq D / r$ for some constant $D$ depending only on data, $R_0$ and $\alpha$.

    We integrate (\ref{eq:wave_l_u}) in the interval $[u_0, u]$; one finds after a change of variables $\tilde{u} \mapsto r(\tilde{u}, v)$ that
    \begin{align*}
        r L \phi (u, v) 
        &= \int_{u_0}^{u} \left [ \frac{ \partial_u \phi} { r } (\tilde{u}, v) - \frac{\Omega^2}{-4 r \partial_v r} \left( 1 - \frac{Q^2}{r^2} \right) \cdot r L \phi (\tilde{u}, v) \right] \, d\tilde{u} + r L \phi (u_0, v), \\[1em]
        &= \int_{r(u_0, v)}^{r(u, v)} \left [ - \frac{1}{r} \cdot r \underline{L} \phi - \frac{\Omega^2}{- 4 r \partial_u r \partial_v r} \left( 1 - \frac{Q^2}{r^2} \right) \cdot r L \phi \right ] \, d r(\tilde{u}, v) + r L \phi (u_0, v).
    \end{align*}
    One deduces from this the integral inequality
    \begin{equation*}
        |r L \phi(u, v)| \leq \int^{r_0}_{r(u, v)} \left [ \frac{1}{\hat{r}} + \sup_{(\tilde{u}, \tilde{v}) \in \Sigma_{\hat{r}}} \frac{\Omega^2}{- 4 r \partial_u r \partial_v r} \cdot \left| \, 1 - \frac{Q^2}{r^2} \, \right| (\tilde{u}, \tilde{v}) \right] \cdot A(\hat{r}) \, d \hat{r} + \sup_{\tilde{v} \in [v_0, v_1]} |r L \phi (u_0, \tilde{v})|.
    \end{equation*}
    
    A similar integral inequality can be found for $|r \underline{L} \phi(u, v)|$. Taking the maximum over $(u, v) \in \Sigma_{\tilde{r}}$, we find
    \begin{equation} \label{eq:a_gronwall}
        A(\tilde{r}) \leq \int_{\tilde{r}}^{r_0} \left [ \frac{1}{\hat{r}} +  \sup_{(\tilde{u}, \tilde{v}) \in \Sigma_{\hat{r}}} \frac{\Omega^2}{- 4 r \partial_u r \partial_v r} \cdot \left| \, 1 - \frac{Q^2}{r^2} \, \right| (\tilde{u}, \tilde{v}) \right] \cdot A(\hat{r}) \, d \hat{r} + \sup_{C_0 \cup \underline{C}_0} \max \{ |r L \phi|, |r \underline{L} \phi| \}.
    \end{equation}
    The last term is controlled by data, and we now apply Gr\"onwall's inequality. The $1/\hat{r}$ appearing in this inequality is responsible for the $1/\tilde{r}$ growth of $A(\tilde{r})$, while the term involving $\Omega^2$ will be controlled by Lemma \ref{lem:upperbounds_lapse}:
    \begin{equation} \label{eq:upperbounds_lapse_gronwall}
        \frac{\Omega^2}{- 4 r \partial_u r \partial_v r} \cdot \left| \, 1 - \frac{Q^2}{r^2} \, \right| \lesssim
        \begin{cases}
            1 & \text{ if } Q = 0, \\
            r^{\alpha - 1} & \text{ if } Q \neq 0 \text{ but (\ref{eq:cond_r}) holds}.
        \end{cases}
    \end{equation}

    Since $\alpha > 0$, this expression is integrable in $r$ as $r \to 0$, and there exists some $\tilde{D}$, such that
    \begin{equation} \label{eq:a_gronwall2}
        \int^{r_0}_{\tilde{r}} \left [ \frac{1}{\hat{r}} +  \sup_{(\tilde{u}, \tilde{v}) \in \Sigma_{\hat{r}}} \frac{\Omega^2}{- 4 r \partial_u r \partial_v r} \cdot \left| \, 1 - \frac{Q^2}{r^2} \, \right| (\tilde{u}, \tilde{v}) \right] \, d\hat{r} \leq \log \left( \frac{r_0}{\tilde{r}} \right) + \tilde{D}.
    \end{equation}

    Hence applying Gr\"onwall to (\ref{eq:a_gronwall}) and using this estimate (\ref{eq:a_gronwall2}), we find that
    \begin{equation}
        A(\tilde{r}) \leq e^{\tilde{D}} \cdot \sup_{C_0 \cup \underline{C}_0} \max \{ |r L \phi|, |r \underline{L} \phi| \} \cdot \frac{r_0}{\tilde{r}}.
    \end{equation}
    Multiplying both sides by $\tilde{r}$, we have a uniform bound on $|r^2 L \phi(u, v)|$ and $|r^2 \underline{L} \phi(u, v)|$ as required.
\end{proof}

\begin{remark}
    Proposition \ref{prop:upperbounds_phi} provides upper bounds for the scalar field $\phi$ appearing in the ESFSS or EMSFSS systems. However, since only the wave equation \eqref{eq:wave_phi} is used in the Gr\"onwall argument, it follows from the proof that a similar estimate would \blue{hold} if $\phi$ were any spherically symmetric scalar function \blue{solving the wave equation \eqref{eq:eom_scalar_field}} in the \textit{fixed background} $(\mathcal{M}, g)$. 

    Note that the $r^{-2}$ growth of $L \phi, \underline{L} \phi$, and the associated logarithmic growth of $\phi$ itself, is ubiquitous in the study of scalar waves in backgrounds involving a Kasner-like spacelike singularity, see for instance \cite{FournodavlosSbierski, AlhoFournodavlosFranzen}.
    \blue{In the sequel, it will be useful} to control not only solutions of the scalar wave equation but also solutions of \textit{inhomogeneous} wave equations in the background $(\mathcal{M}, g)$, where the inhomogeneity has a prescribed blow-up rate. This will be the aim of the following Section \ref{upperbounds.wave}.
\end{remark}

\subsection{Upper bound estimates for general inhomogeneous wave equations} \label{upperbounds.wave}

We prove a result, analogous to Proposition \ref{prop:upperbounds_phi}, producing sharp upper bounds for scalar quantities $\psi$ satisfying a (possibly inhomogeneous) wave equation in the region $\mathcal{D}$. This will be useful when commuting the wave equation with various vector fields, as will be necessary in Section~\ref{upperbound.higher} and Section~\ref{scalarfield}.

\begin{proposition} \label{prop:upperbounds_wave}
    Let $\psi$ be a function on $\mathcal{D} \subset \R^2$ obeying the following inhomogeneous wave equation:
    \begin{equation} \label{eq:wave_psi}
        \partial_u \partial_v \psi = - \frac{\partial_u r \cdot \partial_v \psi}{r} - \frac{\partial_v r \cdot \partial_u \psi}{r} + F.
    \end{equation}
    Here $F$ is a function of $(u, v) \in \mathcal{D}$, such that for some $\gamma \in \R$ and $F_0 > 0$ fixed, $F$ satisfies:
    \begin{equation} \label{eq:wave_psi_inhomog}
        \left| \frac{F(u, v)}{ \partial_u r \, \partial_v r} \right| \leq F_0 \, r^{\gamma}(u, v).
    \end{equation}

    Then there exists a constant $C$, depending only upon the initial data for $(\Omega^2, r, \phi, Q)$ at $C_0 \cup \underline{C}_0$, the exponent $\gamma$, and the quantities $R_0$ and $\alpha$ in (\ref{eq:cond_r}) and (\ref{eq:cond_phi}) when $Q \neq 0$, such that we have the following estimates for the first derivatives $L\psi$ and $\underline{L} \psi$:
    \begin{equation} \label{eq:upperbounds_wave_l}
        |r^2 L \psi(u, v)| + |r^2 \underline{L} \psi (u, v)| \leq C \sup_{C_0 \cup \underline{C}_0} \max \{|r L \psi|, |r \underline{L} \psi| \} + {C F_0}\cdot 
        \begin{cases}
            1 & \text{ if } \gamma > - 2, \\
            \log \left( \frac{r_0}{r(u, v)} \right) & \text{ if } \gamma = - 2, \\
            r(u, v)^{\gamma + 2} & \text{ if } \gamma < - 2.
        \end{cases}
    \end{equation}

    Furthermore, the quantity $\psi$ itself obeys the following estimate:
    \begin{equation} \label{eq:upperbounds_wave_qty}
        |\psi (u, v)| \leq |\psi (u_0, v_0)| +  C \sup_{C_0 \cup \underline{C}_0} \max \{|r L \psi|, |r \underline{L} \psi| \} \log \left( \frac{r_0}{r(u, v)} \right) + {C F_0}\cdot 
        \begin{cases}
            \log \left( \frac{r_0}{r(u, v)} \right) & \text{ if } \gamma > - 2, \\
            \log^2 \left( \frac{r_0}{r(u, v)} \right) & \text{ if } \gamma = - 2, \\
            r(u, v)^{\gamma + 2} & \text{ if } \gamma < - 2.
        \end{cases}
    \end{equation}
\end{proposition}

\begin{proof}
    The proof follows that of Proposition \ref{prop:upperbounds_phi}. We first rewrite (\ref{eq:wave_psi}) as propagation equations for the quantities $r L \psi$ and $r \underline{L} \psi$. One finds
    \begin{equation} \label{eq:wave_psi_l_u}
        \partial_u (r L \psi) = + \frac{\partial_u \psi}{r} - \frac{\Omega^2}{- 4 \partial_v r} \left( 1 - \frac{Q^2}{r^2} \right) \cdot r L \psi + \frac{F}{- \partial_v r},
    \end{equation}
    \begin{equation} \label{eq:wave_psi_l_v}
        \partial_v (r \underline{L} \psi) = + \frac{\partial_v \psi}{r} - \frac{\Omega^2}{- 4 \partial_u r} \left( 1 - \frac{Q^2}{r^2} \right) \cdot r \underline{L} \psi + \frac{F}{- \partial_u r}.
    \end{equation}

    One now integrates \eqref{eq:wave_psi_l_u} and \eqref{eq:wave_psi_l_v} from $C_0$ and $\underline{C}_0$ respectively. Integrating (\ref{eq:wave_psi_l_u}) and changing variables to $r (\tilde{u}, v)$ as before yields:
    \begin{multline} \label{eq:psi_integrated}
        r L \psi (u, v) 
        = \int_{r(u_0, v)}^{r(u, v)} \left [ - \frac{1}{r} \cdot r \underline{L} \psi - \frac{\Omega^2}{- 4 r \partial_u r \partial_v r} \left( 1 - \frac{Q^2}{r^2} \right) \cdot r L \psi \right ] \, d r(\tilde{u}, v) \\ + r L \psi (u_0, v)
        + \int_{r(u_0, v)}^{r(u, v)} \frac{F }{ - \partial_u r \partial_v r} \, dr(\tilde{u}, v).
    \end{multline}

    We wish to find an integral inequality similar to (\ref{eq:a_gronwall}). To simplify notation, we define 
    \begin{equation} \label{eq:beta_r}
        \beta(\tilde{r}) = \frac{1}{\tilde{r}} + \sup_{(\tilde{u}, \tilde{v}) \in \Sigma_{\tilde{r}}} \frac{\Omega^2}{-4 r \partial_u r \partial_v r} \cdot \left| 1 - \frac{Q^2}{r^2} \right| (\tilde{u}, \tilde{v}).
    \end{equation}
    Due to (\ref{eq:upperbounds_lapse_gronwall}), there exist constants $D_0 < D_1$ depending only on the data $(r, \Omega^2, \phi, Q)$ at $C_0 \cup \underline{C}_0$ as before, such that we have the following bounds:
    \begin{equation} \label{eq:beta_bounds}
        \frac{D_0}{\tilde{r}} \leq \beta(r) \leq \frac{D_1}{\tilde{r}},
    \end{equation}
    \begin{equation} \label{eq:beta_integral_bounds}
        \frac{D_0 r_0}{\tilde{r}} \leq \gamma(\tilde{r}) \coloneqq \exp(\int^{r_0}_{\tilde{r}} \beta(\hat{r}) \, d \hat{r}) \leq \frac{D_1 r_0}{\tilde{r}}.
    \end{equation}

    We also use (\ref{eq:wave_psi_inhomog}) to bound the rightmost term appearing in (\ref{eq:psi_integrated}). To avoid having to distinguish cases further based on whether $\gamma \geq -1$, we write
    \begin{equation} \label{eq:F_bounds}
        \int_{r(u_0, v)}^{r(u, v)} \frac{F }{ - \partial_u r \partial_v r} \, dr(\tilde{u}, v) \leq C F_0 \, r(u, v)^{\min\{ \gamma + 1, - 1 / 2 \}}.
    \end{equation}

    We can now write down the inequality analogous to (\ref{eq:a_gronwall}) in Proposition~\ref{prop:upperbounds_phi}. Define
    \begin{equation*}
        A_{\psi}(\tilde{r}) = \sup_{\tilde{r} \leq \hat{r} \leq r_0} \max \, \left \{ \max_{(u, v) \in \Sigma_{\hat{r}}} |r L \psi(u, v) |, \max_{(u, v) \in \Sigma_{\hat{r}}} |r \underline{L} \psi (u, v)| \, \right \}.
    \end{equation*}
    Then using (\ref{eq:psi_integrated}) and the corresponding equation for $r \underline{L} \psi$, as well as (\ref{eq:F_bounds}), one finds
    \begin{equation} \label{eq:a_psi_gronwall}
        A_{\psi}(\tilde{r}) \leq \int^{r_0}_{\tilde{r}} \beta(\hat{r}) A_{\psi}(\hat{r}) \, d \hat{r} + P_0 + C F_0 \, \tilde{r}^{\min\{\gamma + 1, - 1/2 \}}.
    \end{equation}
    Here $P_0 = \sup_{C_0 \cup \underline{C}_0} \max \{ |r L \psi|, |r \underline{L} \psi| \}$ is the term arising from the initial data for $\psi$.

    Since the usual integral Gr\"onwall inequality applied to (\ref{eq:a_psi_gronwall}) would produce a weaker estimate than the required (\ref{eq:upperbounds_wave_l}), we instead use (\ref{eq:a_psi_gronwall}) to produce a differential inequality for the quantity
    \begin{equation*}
        B_{\psi}(\tilde{r}) = \int^{r_0}_{\tilde{r}} \beta(\hat{r}) A_{\psi}(\hat{r}) \, d \hat{r}.
    \end{equation*}
    From (\ref{eq:a_psi_gronwall}), we find the following inequality for $B_{\psi}(\tilde{r})$:
    \begin{equation*}
        \frac{d}{d\tilde{r}} B_{\psi}(\tilde{r}) = - \beta(\tilde{r}) A_{\psi}(\tilde{r}) \geq - \beta (\tilde{r}) B_{\psi}(\tilde{r}) - P_0 \, \blue{\beta(\tilde{r})} - C F_0 \, \beta(\tilde{r}) \, \tilde{r}^{\min\{\gamma + 1, - 1/2\}}.
    \end{equation*}

    Defining $\gamma(\tilde{r})$ as in (\ref{eq:beta_integral_bounds}), one thus finds
    \begin{equation} \label{eq:b_gamma}
        \frac{d}{d\tilde{r}} (\gamma^{-1}(\tilde{r}) B_{\psi}(\tilde{r})) \geq - P_0 \gamma^{-1}(\tilde{r}) \blue{\beta(\tilde{r})} - C F_0 \gamma^{-1}(\tilde{r}) \beta(\tilde{r}) \, \tilde{r}^{\min \{\gamma + 1, - 1/2 \}}.
    \end{equation}
    \blue{Next, we integrate \eqref{eq:b_gamma} backwards in $\tilde{r}$, and use that, from (\ref{eq:beta_bounds}) and (\ref{eq:beta_integral_bounds}), we have $\gamma^{-1}(\tilde{r}) \sim \tilde{r}$ and $\beta(\tilde{r}) \sim \tilde{r}^{-1}$, to get
    \[
        \tilde{r} B_{\psi}(\tilde{r}) \leq C \left( P_0 r_0 + \int^{r_0}_r F_0 \tilde{r}^{\min \{\gamma + 1, - 1/2\} } \, d \tilde{r} \right).
    \]
Thus, distinguishing cases based on the integrability of $r^{\min \{ \gamma+1, -1/2\}}$, we find that}
    \begin{equation} \label{eq:b_gronwall_estimate}
        \tilde{r} B_{\psi}(\tilde{r}) \leq 
        C P_0 \blue{r_0} + {C F_0}\cdot 
        \begin{cases}
            \blue{r_0^{\min\{\gamma+2, 1/2\}}} & \text{ if } \gamma > - 2, \\
            \log \left( \frac{r_0}{\tilde{r}} \right) & \text{ if } \gamma = - 2, \\
            \tilde{r}^{\gamma + 2} & \text{ if } \gamma < - 2.
        \end{cases}
    \end{equation}

    Inserting (\ref{eq:b_gronwall_estimate}) \blue{into} (\ref{eq:a_psi_gronwall}), \blue{it is evident that the} $\tilde{r} B_{\psi}(\tilde{r})$ in the left hand side of (\ref{eq:b_gronwall_estimate}) can be replaced by $\tilde{r} A_{\psi}(\tilde{r})$, at the cost of increasing the constant $C$. \blue{Upon also absorbing any powers of $r_0$ into this constant $C$}, this estimate is exactly the desired (\ref{eq:upperbounds_wave_l}).

    To find the final estimate (\ref{eq:upperbounds_wave_qty}) for $\psi$ itself, note that the vector fields $r L$ and $r \underline{L}$ are morally \blue{$\frac{d}{dr}$}-derivatives, taken in null directions. To make this concrete, given $G(r)$ any real function of $r$ and $\zeta: \mathcal{D} \subset \R^2$ a $C^1$ function satisfying $|r L \zeta (u, v)|, |r \underline{L}\zeta (u, v) | \leq G(r(u, v))$, then
    \begin{align}
        |\zeta(u, v)|
        &\leq |\zeta(u_0, v_0)| + |\zeta(u_0, v) - \zeta(u_0, v_0)| + |\zeta(u, v) - \zeta(u_0, v)|, \nonumber \\[0.5em]
        &\leq |\zeta(u_0, v_0)| + \int^v_{v_0} |\partial_v \zeta (u_0, \tilde{v})| \, d \tilde{v} + \int^u_{u_0} |\partial_u \zeta (\tilde{u}, v) | \, d \tilde{u}, \nonumber \\[0.5em]
        &= |\zeta(u_0, v_0)| + \int^{r(u_0, v)}_{r_0} |r L \zeta (u_0, \tilde{v})| \, dr (u_0,\tilde{v}) + \int^{r(u, v)}_{r(u_0, v)} |r \underline{L} \zeta (\tilde{u}, v)| \, dr(\tilde{u}, v), \nonumber \\[0.5em]
        &\leq |\zeta(u_0, v_0)| + \int^{r_0}_{r(u, v)} G(\tilde{r}) \, d \tilde{r}. \label{eq:zll}
    \end{align}
    Using (\ref{eq:upperbounds_wave_l}) and applying (\ref{eq:zll}) with $\zeta = \psi$, we deduce (\ref{eq:upperbounds_wave_qty}) as claimed.
\end{proof}

\subsection{Higher order estimates} \label{upperbound.higher}

As in Section 7 of \cite{AnZhang}, it will be important in later sections to control higher derivatives of the quantities $(r, \Omega^2, \phi)$. We start with the geometric quantities $r$ and $\Omega^2$:

\begin{proposition} \label{prop:upperbounds_higher_geometry}
    We have the following estimates for derivatives of $\Omega^2(u, v)$:
    \begin{enumerate}[(A)]
        \item
            For the ESFSS system, there exists a constant $D$ such that
            \begin{equation} \label{eq:upperbounds_lapse_der}
                |\partial_u \Omega^2| \leq \frac{D}{r^3(u, v)}, \qquad
                |\partial_v \Omega^2| \leq \frac{D}{r^3(u, v)}.
            \end{equation}
        \item
            For the EMSFSS system supplemented with \ref{eq:cond_r}) and (\ref{eq:cond_phi}), we have instead that:
            \begin{equation} \label{eq:upperbounds_lapse_der_em}
                |\partial_u \Omega^2| \leq \frac{D}{r^{4 - \alpha} (u, v)}, \qquad
                |\partial_v \Omega^2| \leq \frac{D}{r^{4 - \alpha} (u, v)}.
            \end{equation}
    \end{enumerate}

    Moving onto estimates for derivatives of $r^2(u, v)$, we have
    \begin{enumerate}[(A)]
        \item
            For the ESFSS system, with $D$ as above, we have
            \begin{equation} \label{eq:upperbounds_rr_der}
                |\partial_u ( - r \partial_v r)| = |\partial_v ( - r \partial_u r)| \leq \frac{D}{r(u, v)},
            \end{equation}
            \begin{equation} \label{eq:upperbounds_rr_der2}
                |\partial_u ( - r \partial_u r)| \leq \frac{D}{r(u, v)}, \qquad
                |\partial_v ( - r \partial_v r)| \leq \frac{D}{r(u, v)}.
            \end{equation}
        \item
            For the EMSFSS system, with $D$ as above, we have
            \begin{equation} \label{eq:upperbounds_rr_der_em}
                |\partial_u ( - r \partial_v r)| = |\partial_v ( - r \partial_u r)| \leq \frac{D}{r^{2 - \alpha}(u, v)},
            \end{equation}
            \begin{equation} \label{eq:upperbounds_rr_der2_em}
                |\partial_u ( - r \partial_u r)| \leq \frac{D}{r^{2 - \alpha}(u, v)}, \qquad
                |\partial_v ( - r \partial_v r)| \leq \frac{D}{r^{2 - \alpha}(u, v)}.
            \end{equation}
        \item
            In both cases, one also has
            \begin{equation} \label{eq:upperbounds_r_der}
                |\partial_u \partial_u r| \leq \frac{D}{r^3(u, v)}, \qquad
                |\partial_u \partial_v r| \leq \frac{D}{r^3(u, v)}, \qquad
                |\partial_v \partial_v r| \leq \frac{D}{r^3(u, v)}.
            \end{equation}
    \end{enumerate}
\end{proposition}

\begin{remark}
    Regarding (\ref{eq:upperbounds_lapse_der})--(\ref{eq:upperbounds_r_der}), due to \eqref{eq:esfss_cond} it is straightforward to derive similar estimates except with the derivatives $\partial_u, \partial_v$ replaced by the gauge-invariant derivatives $\underline{L}, L$ respectively.
\end{remark}

\begin{proof}
    Consider the evolution equation (\ref{eq:wave_omega}) for $\log \Omega^2$. Using Lemma~\ref{lem:upperbounds_lapse} and Proposition~\ref{prop:upperbounds_phi}, we find from this equation the following estimate:
    \begin{equation} \label{eq:upperbounds_wave_omega}
        |\partial_u \partial_v \log \Omega^2| \lesssim \frac{1}{r^4}.
    \end{equation}
    We now integrate \eqref{eq:upperbounds_wave_omega} in the $u$ direction; since $- r \partial_u r \geq R_0 \gtrsim 1$ we have
    \begin{align*}
        |\partial_v \log \Omega^2|(u, v)
        &\lesssim |\partial_v \log \Omega^2|(u_0, v) + \int_{u_0}^{u} \frac{1}{r^4} (\tilde{u}, v) \, d \tilde{u}, \\[0.5em]
        &\lesssim |\partial_v \log \Omega^2|(u_0, v) + \int_{r(u, v)}^{r(u_0, v)} \frac{1}{r^3} (\tilde{u}, v) \, d r(\tilde{u}, v) \lesssim \frac{1}{r^2(u, v)},
    \end{align*}
    noting that the first term on the right hand side arises from data and can be ignored. 

    Writing $|\partial_v \Omega^2| = \Omega^2 \cdot |\partial_v \log \Omega^2|$ and using Lemma~\ref{lem:upperbounds_lapse} to bound $\Omega^2$, one deduces the $\partial_v \Omega^2$ estimates in (\ref{eq:upperbounds_lapse_der}) and (\ref{eq:upperbounds_lapse_der_em}). The $\partial_u \Omega^2$ estimates follow analogously.

    We turn to the second derivative estimates for $r^2(u, v)$. (\ref{eq:upperbounds_rr_der}) and (\ref{eq:upperbounds_rr_der_em}) follow easily from the evolution equations (\ref{eq:wave_r_u}) and (\ref{eq:wave_r_v}), where Lemma~\ref{lem:upperbounds_lapse} is used to bound the $\Omega^2$ arising in these equations. On the other hand, to deal with $\partial_u ( - r \partial_u r )$ for instance, we differentiate (\ref{eq:wave_r_v}) in $u$ to find:
    \begin{equation}
        \partial_v \partial_u ( - r \partial_u r ) = \frac{\partial_u \Omega^2}{4} \left( 1 - \frac{Q^2}{r^2} \right) + \frac{\Omega^2}{4} \frac{Q^2 \partial_u r}{r^3}.
    \end{equation}

    Using (\ref{eq:upperbounds_lapse_der}) or (\ref{eq:upperbounds_lapse_der_em}) alongside Lemma~\ref{lem:upperbounds_lapse}, one therefore finds:
    \begin{equation} \label{eq:vuur}
        |\partial_v \partial_u ( - r \partial_u r )| \leq
        \begin{cases}
            D \, r(u, v)^{-3} & \text{ for ESFSS}, \\
            D \, r(u, v)^{-4 + \alpha} & \text{ for EMSFSS, given (\ref{eq:cond_r}), (\ref{eq:cond_phi}).}
        \end{cases}
    \end{equation}

    Integrating \eqref{eq:vuur} and changing the integration variable from $v$ to $r$ in a familiar fashion, one gets
    \begin{equation}
        |\partial_u ( - r \partial_u r )| \leq
        \begin{cases}
            D \, r(u, v)^{-1} & \text{ for ESFSS}, \\
            D \, r(u, v)^{-2 + \alpha} & \text{ for EMSFSS, given (\ref{eq:cond_r}), (\ref{eq:cond_phi}).}
        \end{cases}
    \end{equation}
    An analogous estimate holds for $\partial_v ( -r \partial_v r)$, yielding exactly (\ref{eq:upperbounds_rr_der2}) and (\ref{eq:upperbounds_rr_der2_em}). The final estimate (\ref{eq:upperbounds_r_der}) is straightforward from combining all previous estimates, e.g.\
    \begin{equation*}
        |\partial_u \partial_u r| = \frac{1}{r} \left| (\partial_u r)^2 + \partial_u ( - r \partial_u r) \right| \lesssim \frac{1}{r^3}.
    \end{equation*}
    This completes the proof of Proposition~\ref{prop:upperbounds_higher_geometry}.
\end{proof}

We now use Propositions~\ref{prop:upperbounds_higher_geometry} and \ref{prop:upperbounds_wave} to derive estimates for second order derivatives of $\phi(u, v)$. 

\begin{proposition} \label{prop:upperbounds_higher_phi}
    For either the ESFSS system or EMSFSS supplemented with (\ref{eq:cond_r}), (\ref{eq:cond_phi}), there exists a constant $D$ such that the second order derivatives of $\phi$ satisfy
    \begin{equation} \label{eq:upperbounds_phi_der}
        |\partial_u \partial_u \phi| \leq \frac{D}{r^4(u, v)}, \qquad
        |\partial_u \partial_v \phi| \leq \frac{D}{r^4(u, v)}, \qquad
        |\partial_v \partial_v \phi| \leq \frac{D}{r^4(u, v)}.
    \end{equation}
    As in Proposition \ref{prop:upperbounds_higher_geometry}, we may find a similar estimate to (\ref{eq:upperbounds_r_der}) with $\partial_u, \partial_v$ replaced by $\underline{L}, L$ respectively.
\end{proposition}

\begin{proof}
    Commute the wave equation (\ref{eq:wave_phi}) for $\phi$ with the coordinate derivative $\partial_u$. One finds:
    \begin{equation} \label{eq:wave_phi_commute_u}
        \partial_u \partial_v (\partial_u \phi) = - \frac{\partial_u r \cdot \partial_v (\partial_u \phi)}{r} - \frac{\partial_v r \cdot \partial_u ( \partial_u \phi )}{r} + F_u.
    \end{equation}
    Here the inhomogeneity $F$ is given by:
    \begin{equation} \label{eq:wave_phi_commute_u_inhomog}
        F_u = \left( \frac{( \partial_u r)^2}{r^2} - \frac{\partial_u \partial_u r}{r} \right) \partial_v \phi + \left( \frac{\partial_u r \partial_v r}{r^2} - \frac{\partial_u \partial_v r}{r} \right) \partial_u \phi.
    \end{equation}

    Using Proposition~\ref{prop:upperbounds_higher_geometry}, particularly (\ref{eq:upperbounds_r_der}), and Proposition~\ref{prop:upperbounds_phi}, one estimates
    \begin{equation*}
        \left| \frac{F_u}{ \partial_u r \partial_v r} \right| \lesssim \frac{1}{r^4}.
    \end{equation*}
    In light of (\ref{eq:wave_phi_commute_u}), one now applies Proposition~\ref{prop:upperbounds_wave} with $\psi = \partial_u \phi$ and $\gamma = - 4$ to find that
    \begin{equation*}
        |r^2 L \partial_u \phi| + |r^2 \underline{L} \partial_u \phi| \lesssim \frac{1}{r^2}.
    \end{equation*}

    Since $L$ and $\underline{L}$ are uniformly equivalent to $\partial_v$ and $\partial_u$ respectively, one therefore deduces the $\partial_u \partial_u \phi$ and $\partial_u \partial_v \phi$ estimate in (\ref{eq:upperbounds_phi_der}). The $\partial_v \partial_v \phi$ estimate follows by commuting with $\partial_v$ instead of $\partial_u$.
\end{proof}

%% file: scalarfield.tex
\section{Precise asymptotics for the scalar field} \label{scalarfield}

The aim of the present section will be to prove the parts of Theorems~\ref{thm:esfss} and \ref{thm:emsfss} concerning the scalar field $\phi$. The key step will be to show that the weighted derivatives $r^2 L \phi$, $r^2 \underline{L} \phi$ can be extended in a $C^{0, \beta}$ fashion to the $r = 0$ boundary. 

As discussed in Section~\ref{intro.proof}, the key ingredient will be the following: treating instead $r^2 \partial_v \phi = - r \partial_v r \cdot r^2 L \phi$ in this preliminary discussion, one has from \eqref{eq:wave_phi_u} the equation
\begin{equation*}
    \partial_u (r^2 \partial_v \phi) = - r \partial_v r \cdot \partial_u \phi + r \partial_u r \cdot \partial_v \phi.
\end{equation*}
The results of Section \ref{upperbounds}, particularly Proposition \ref{prop:upperbounds_phi}, imply that the right-hand side is $O(r^{-2})$, which is borderline non-integrable in the variable $u$ as $r \to 0$. 

Thankfully, the operator $\tilde{X} = - r \partial_v r \cdot \partial_v + r \partial_u r \cdot \partial_u$ turns out to be a better derivative than the standard coordinate derivatives $\partial_u$, $\partial_v$. Since $\tilde{X} r = 0$, the vector field $\tilde{X}$ is adapted to the constant-$r$ spacelike hypersurfaces $\Sigma_{\tilde{r}}$, see Figure \ref{fig:char_ivp_foliation}. Hence one expects that such a ``spatial'' derivative behaves better. 

Indeed, we shall see from Proposition \ref{prop:scalarfield_xphi} that $\tilde{X} \phi$ is integrable up to the $r = 0$ singularity, allowing us to extend $r^2 \partial_u \phi$ and $r^2 \partial_v \phi$ to the singularity in a continuous manner. As in Section \ref{upperbounds}, the argument is streamlined upon replacing $\partial_u$ and $\partial_v$ by their gauge-independent counterparts $\underline{L}$ and $L$. In particular, instead of treating $\tilde{X}$ as above we will study the spatial vector field $X = L - \underline{L}$.

\subsection{The vector field \texorpdfstring{$X= L - {\protect \underline{L}}$}{X=L-L͟}} \label{scalarfield.x}

Define the vector field $X$ by
\begin{equation} \label{eq:scalarfield_x}
    X = L - \underline{L} = \frac{1}{- r \partial_v r} \partial_v - \frac{1}{- r \partial_u r} \partial_u.
\end{equation}
From \eqref{eq:scalarfield_x}, it immediately follows that $X r = 0$, hence $X$ is tangent to the constant-$r$ hypersurfaces $\Sigma_{\tilde{r}}$. In particular, $X$ is spacelike, and we can further compute
\begin{equation} \label{eq:scalarfield_x_length}
    g(X, X) = - 2 g (L, \underline{L}) = \frac{\Omega^2}{r^2 \partial_u r \, \partial_v r} > 0.
\end{equation}

We wish to show that $X \phi$ obeys better estimates than those of $L \phi$, $\underline{L} \phi$ found in Proposition \ref{prop:upperbounds_phi}. For this purpose, we commute the wave equation (\ref{eq:wave_phi}) with $X$, and apply Proposition \ref{prop:upperbounds_wave}. We therefore first produce some preliminary estimates regarding the various commutator terms that arise.

\begin{lemma} \label{lem:scalarfield_x_commutator_1}
    For $X$ as defined in (\ref{eq:scalarfield_x}), decompose the commutator vector field $[\partial_u, X]$ with respect to the coordinate basis $\partial_u, \partial_v$ in the usual way:
    \begin{equation*}
        [\partial_u, X] = [\partial_u, X]^u \partial_u + [\partial_u, X]^v \partial_v,
    \end{equation*}
    and similarly for $[\partial_v, X]$ and $[L, X] = - [L, \underline{L}] = [\underline{L}, X]$. Then there exists a constant $D > 0$, depending on the data at $C_0 \cup \underline{C}_0$, such that
    \begin{multline} \label{eq:x_commutator_1}
        |[\partial_v, X]^u| + |[\partial_u, X]^v| + |[\partial_v, X]^u| + |[\partial_v, X]^v| + |[L, X]^u| + |[L, X]^v|
        \\ \leq
        \begin{cases}
            D \, r(u, v)^{-1} & \text{ for ESFSS}, \\
            D \, r(u, v)^{-2 + \alpha} & \text{ for EMSFSS, given (\ref{eq:cond_r}), (\ref{eq:cond_phi}).}
        \end{cases}
    \end{multline}

    Alternatively, one may decompose the commutator vector fields with respect to the basis $L, \underline{L}$:
    \begin{equation*}
        [\partial_u, X] = [\partial_u, X]^{\underline{L}} \underline{L} + [\partial_v, X]^{L} L,
    \end{equation*}
    and similarly for $[\partial_v, X]$ and $[L, X]$. Then (\ref{eq:x_commutator_1}) holds with the $u$, $v$ indices replaced by $\underline{L}, L$ indices.
\end{lemma}

\begin{proof}
    This will follow from the higher derivative estimates of Section \ref{upperbound.higher}. Using (\ref{eq:scalarfield_x}), we have
    \begin{align*}
        [\partial_u, X] 
        &= \partial_u \left( \frac{1}{- r \partial_v r} \right) \partial_v - \partial_u \left( \frac{1}{- r \partial_u r} \right) \partial_u, \\[1em]
        &= - \frac{ \partial_u ( - r \partial_v r)}{(- r \partial_v r )^2} \partial_v + \frac{ \partial_u ( - r \partial_u r)}{(- r \partial_u r)^2} \partial_u.
    \end{align*}

    Hence we may read off $[\partial_u, X]^v$ and $[\partial_v, X]^u$, and bound these using Proposition \ref{prop:upperbounds_higher_geometry} and the lower bounds for $- r \partial_u r$ and $- r \partial_v r$. The commutators $[\partial_v, X]$ and $[L, X]$ are dealt with similarly.

    For the final statement of the lemma, note that changing basis from $\partial_u$, $\partial_v$ to $\underline{L}$, $L$ simply involves multiplying and dividing by $- r \partial_u r$ and $- r \partial_v r$. Since these are upper and lower bounded (see Section \ref{setup.sing} and (\ref{eq:cond_r})), the estimate (\ref{eq:x_commutator_1}) still holds when computing components of the various commutator vector fields with respect to the $L$, $\underline{L}$ basis.
\end{proof}

We also need to consider commutation with the second order wave operator $\partial_u \partial_v$:
\begin{lemma} \label{lem:scalarfield_x_commutator_2}
    We write the second order operator $[\partial_u \partial_v, X]$ in the form
    \begin{equation} \label{eq:xuv}
        [\partial_u \partial_v, X] = \mathcal{X}_{uu} \partial_u \partial_u + \mathcal{X}_{uv} \partial_u \partial_v + \mathcal{X}_{vv} \partial_v \partial_v + \mathcal{X}_{u} \partial_u + \mathcal{X}_{v} \partial_v,
    \end{equation}
    where $\mathcal{X}_{uu}, \mathcal{X}_{uv}, \mathcal{X}_{vv}, \mathcal{X}_u, \mathcal{X}_v$ are each scalar functions of $(u, v) \in \mathcal{D}$. Then we have the following bounds:
    \begin{equation} \label{eq:x_commutator_2}
        |\mathcal{X}_{uu}| + |\mathcal{X}_{uv}| + |\mathcal{X}_{vv}| \leq
        \begin{cases}
            D \, r(u, v)^{-1} & \text{ for ESFSS}, \\
            D \, r(u, v)^{-2 + \alpha} & \text{ for EMSFSS, given (\ref{eq:cond_r}), (\ref{eq:cond_phi}),}
        \end{cases}
    \end{equation}
    \begin{equation} \label{eq:x_commutator_3}
        |\mathcal{X}_{u}| + |\mathcal{X}_{v}| \leq
        \begin{cases}
            D \, r(u, v)^{-3} & \text{ for ESFSS}, \\
            D \, r(u, v)^{-4 + \alpha} & \text{ for EMSFSS, given (\ref{eq:cond_r}), (\ref{eq:cond_phi}).}
        \end{cases}
    \end{equation}
\end{lemma}

\begin{proof}
    Once again, we may write explicit formulae for each of $\mathcal{X}_{uu}, \mathcal{X}_{uv}, \mathcal{X}_{vv}, \mathcal{X}_u, \mathcal{X}_v$, then appeal to (\ref{prop:upperbounds_higher_geometry}). Using $[\partial_u \partial_v, X] = \partial_u \circ [\partial_v, X] + [\partial_u, X] \circ \partial_v$, we find that:
    \begin{gather*}
        \mathcal{X}_{uu} = [\partial_v, X]^{u} = \partial_v \left( \frac{-1}{- r \partial_u r} \right), \\[1em]
        \mathcal{X}_{uv} = [\partial_v, X]^{v} + [\partial_u, X]^u = \partial_v \left(\frac{1}{- r \partial_v r} \right) + \partial_u \left( \frac{-1}{- r \partial_u r} \right), \\[1em]
        \mathcal{X}_{vv} = [\partial_u, X]^{v} = \partial_u \left( \frac{1}{- r \partial_v r} \right), \\[1em]
        \mathcal{X}_u = \partial_u ( [\partial_v, X]^{u} ) = \partial_u \partial_v \left( \frac{-1}{-r \partial_u r} \right), \\[1em]
        \mathcal{X}_v = \partial_u ( [\partial_v, X]^{v} ) = \partial_u \partial_v \left( \frac{1}{-r \partial_v r} \right).
    \end{gather*}

    In particular, (\ref{eq:x_commutator_2}) follows immediately from (\ref{eq:x_commutator_1}), while for (\ref{eq:x_commutator_3}), we first use the equations (\ref{eq:wave_r_u}) and (\ref{eq:wave_r_v}). For instance, we rewrite $\mathcal{X}_u$ using (\ref{eq:wave_r_v}) as
    \begin{equation*}
        \mathcal{X}_u = \partial_u \left( \frac{1}{(- r \partial_u r)^2} \cdot \frac{\Omega^2}{4} \left( 1 - \frac{Q^2}{r^2} \right) \right).
    \end{equation*}
    The $\partial_u$ derivative could hit any of $( - r \partial_u r)^{-1}$, $\Omega^2$ or $r^{-1}$. However, we see from Lemma \ref{lem:upperbounds_lapse} and Proposition \ref{prop:upperbounds_higher_geometry} that in any of these cases, the result remains bounded by the right hand side of (\ref{eq:x_commutator_3}). The estimate for $\mathcal{X}_v$ is analogous, concluding the proof of Lemma \ref{lem:scalarfield_x_commutator_2}.
\end{proof}

\subsection{Estimates on \texorpdfstring{$X \phi$}{X𝞍}} \label{scalarfield.xphi}

In this section we prove Proposition \ref{prop:scalarfield_xphi}, producing upper bounds on $X \phi$ and its derivatives. One achieves this by commuting (\ref{eq:wave_phi}) with the vector field $X$, then appealing to Proposition \ref{prop:upperbounds_wave}. For convenience, we record (\ref{eq:wave_phi}) again here:
\begin{equation} \label{eq:wave_phi_copy}
    \partial_u \partial_v \phi = - \frac{\partial_u r \cdot \partial_v \phi}{r} - \frac{\partial_v r \cdot \partial_u \phi}{r}.
\end{equation}

\begin{proposition} \label{prop:scalarfield_xphi}
    For the ESFSS system or the EMSFSS system supplemented with the additional assumptions (\ref{eq:cond_r}), (\ref{eq:cond_phi}), we have the following upper bounds for $X \phi$ and its weighted derivatives:
    \begin{equation} \label{eq:scalarfield_xphi}
        |X \phi (u, v)| + | r^2 L X \phi (u, v) | + | r^2 \underline{L} X \phi (u, v) | \leq
        \begin{cases}
            D \, r^{-1}(u, v) & \text{ for ESFSS},\\
            D \, r^{-2 + \alpha}(u, v) & \text{ for EMSFSS, given (\ref{eq:cond_r}), (\ref{eq:cond_phi})}.
        \end{cases}
    \end{equation}
\end{proposition}

\begin{proof}
    In light of Proposition \ref{prop:upperbounds_wave}, we seek an equation of the form (\ref{eq:wave_psi}), with $\psi = X \phi$. We apply $X$ to equation (\ref{eq:wave_phi_copy}); since $X r = 0$ we find
    \begin{equation} \label{eq:wave_xphi}
        \partial_u \partial_v X \phi = - \frac{\partial_u r \cdot \partial_v X \phi}{r} - \frac{\partial_v r \cdot \partial_u X \phi}{r} + F_X,
    \end{equation}
    where the inhomogeneous term $F_X$ is given exactly by \blue{
    \begin{equation} \label{eq:wave_xphi_inhomog}
        F_X = [\partial_u \partial_v, X] \phi + \frac{\partial_u r}{r} \cdot [\partial_v, X] \phi + \frac{\partial_v r}{r} \cdot [\partial_u, X] \phi + \frac{[\partial_u, X] r}{r} \cdot \partial_v \phi  + \frac{[\partial_v, X] r}{r} \cdot \partial_u \phi.
\end{equation}}
\indent It simply remains to estimate this inhomogeneous term $F_X$ using the results of Section \ref{scalarfield.x}. 
    We introduce some schematic notation; we write $\mathcal{X}_{**}$ to represent some element of $\{ [\partial_u, X]^u, [\partial_u, X]^v, [\partial_v, X]^u, [\partial_v, X]^v\}$ or any linear combination thereof. Note that $\mathcal{X}_{uu}, \mathcal{X}_{uv}, \mathcal{X}_{vv}$ are included in $\mathcal{X}_{**}$, see the proof of Lemma \ref{lem:scalarfield_x_commutator_2}. Likewise, denote $\mathcal{X}_{*}$ to represent one of $\mathcal{X}_u$ or $\mathcal{X}_v$ as in Lemma \ref{lem:scalarfield_x_commutator_2}. Finally, denote $\partial_*$ to be one of the two null coordinate derivatives $\partial_u, \partial_v$.

    Using this notation and Lemma \ref{lem:scalarfield_x_commutator_2}, $F_X$ can be written in the schematic form:
    \begin{equation} \label{eq:wave_xphi_inhomog_2}
        F_X = \mathcal{X}_* \cdot \partial_* \phi + \mathcal{X}_{**} \cdot \partial_* \partial_* \phi + \mathcal{X}_{**} \cdot \frac{\partial_* r}{r} \cdot \partial_* \phi.
    \end{equation}
    Now, from Lemmas \ref{lem:scalarfield_x_commutator_1} and \ref{lem:scalarfield_x_commutator_2}, we may deduce that
    \begin{equation*}
        |r^2 \mathcal{X}_*| + |\mathcal{X}_{**}| \lesssim 
        \begin{cases}
            r^{-1}(u, v) & \text{ for ESFSS},\\
            r^{-2 + \alpha}(u, v) & \text{ for EMSFSS, given (\ref{eq:cond_r}), (\ref{eq:cond_phi})}.
        \end{cases}
    \end{equation*}

    We also need to control $\partial_* r$, $\partial_* \phi$ and $\partial_* \partial_* \phi$. The upper bounds for $ - r \partial_* r$ give that $ - \partial_* r \lesssim r^{-1}$, while derivatives of $\phi$ are controlled using Proposition~\ref{prop:upperbounds_phi} for $\partial_* \phi$ and Proposition~\ref{prop:upperbounds_higher_phi} for $\partial_* \partial_* \phi$. These yield
    \begin{equation*}
        |\partial_* \phi| \lesssim r^{-2}, \qquad |\partial_* \partial_* \phi| \lesssim r^{-4}.
    \end{equation*}
    Applying all of these bounds to (\ref{eq:wave_xphi_inhomog_2}), we find that
    \begin{equation*}
        |F_X| \lesssim 
        \begin{cases} 
            r^{-5}(u, v) & \text{ for ESFSS},\\
            r^{-6 + \alpha}(u, v) & \text{ for EMSFSS, given (\ref{eq:cond_r}), (\ref{eq:cond_phi})}.
        \end{cases}
    \end{equation*}

    Using again lower bounds for $- r \partial_* r$, we therefore have that
    \begin{equation} \label{eq:fxuv}
        \left| \frac{F_X}{\partial_u r \, \partial_v r} \right|
        \leq F_{X,0} \cdot r^{\gamma}(u, v).
    \end{equation}
    Here $F_{X, 0}$ is some constant depending only on data, while $\gamma = - 3$ and $\gamma = - 4 + \alpha$ in the context of the ESFSS and the EMSFSS system respectively. We are now in a position to apply Proposition \ref{prop:upperbounds_wave} to (\ref{eq:wave_xphi}).

    Since we have $\gamma < -2$ in both cases, the final term on the right hand sides of (\ref{eq:upperbounds_wave_l}) and (\ref{eq:upperbounds_wave_qty}) arising from the inhomogeneity dominate any terms arising from the data for $\psi = X \phi$ at $C_0 \cup \underline{C}_0$. 
    Therefore the estimates (\ref{eq:upperbounds_wave_l}) and (\ref{eq:upperbounds_wave_qty}) applied to $\psi = X \phi$ yield (\ref{eq:scalarfield_xphi}), as required.
\end{proof}

\subsection{Asymptotics for the scalar field at the singularity} \label{scalarfield.asymptotics}

Armed with Proposition \ref{prop:scalarfield_xphi}, 
we now prove the asymptotics for $\phi$ given in Theorems~\ref{thm:esfss} and \ref{thm:emsfss}, including the appropriate H\"older continuity for the quantities $\Psi$ and $\Xi$.

\begin{proposition} \label{prop:scalarfield_asymptotics}
    Let $(r, \Omega^2, \phi, Q)$ be a strongly singular solution to the EMSFSS system with an $r=0$ spacelike singularity as in Section \ref{setup.sing}, where we assume (\ref{eq:cond_r}), (\ref{eq:cond_phi}) to hold if $Q \neq 0$. Define the functions $\Psi: \mathcal{D} \to \R$ and $\Xi: \mathcal{D} \to \R$ to be:
    \begin{gather} \label{eq:Psi}
        \Psi(u, v) \coloneqq r^2 \cdot \frac{( L + \underline{L} ) \phi (u, v)}{2},
        \\[1em] \label{eq:Xi}
        \Xi(u, v) \coloneqq \phi(u, v) - \Psi (u, v) \cdot \log \left( \frac{r_0}{r(u, v)} \right).
    \end{gather}
    Then $\Psi$ is $C^{0, \beta}$ in $\mathcal{D}$, with H\"older exponent $\beta = 1/2$ for the ESFSS system and $\beta = \alpha / 2$ for the EMSFSS system with (\ref{eq:cond_r}), (\ref{eq:cond_phi}). $\Xi$, on the other hand, is $C^{0, \beta, \log}$ (see \eqref{eq:logholder} for a definition of log-H\"older continuity).
    In fact, $\| \Psi \|_{C^{0, \beta}}$ and $\| \Xi \|_{C^{0, \beta, \log}}$ are bounded by a constant $D$ depending on the data:
    \begin{equation} \label{eq:psixi}
        \| \Psi \|_{C^{0, \beta}} + \| \Xi \|_{C^{0, \beta, \log}} \leq D.
    \end{equation}
\end{proposition}

\blue{
\begin{remark}
    To explain the definitions \eqref{eq:Psi}--\eqref{eq:Xi}, recall from Theorem~\ref{thm:esfss}(\ref{item:esfss_thm_ii}) and Theorem~\ref{thm:emsfss}(\ref{item:emsfss_thm_ii}) that we wish to write $\phi = \Psi \cdot \log( \frac{r_0}{r} ) + \Xi$. This immediately justifies \eqref{eq:Xi}, while to justify \eqref{eq:Psi}, we recall that $L$ and $\underline{L}$ are, in some sense, derivatives of the form $- r^{-1} \frac{d}{dr}$ adapted to null hypersurfaces. As a result, $r^2 \cdot \frac{L + \underline{L}}{2}$ can be viewed as a $- r \frac{d}{dr}$ derivative, and thus identifies the coefficient $\Psi$ in front of $\log(\frac{r_0}{r})$.
\end{remark}
}

\begin{proof}
    We study the derivatives of the quantities $\Psi$ and $\Xi$. Using the wave equation (\ref{eq:wave_phi_copy}), we produce some null transport equations for $r^2 L \phi$ and $r^2 \underline{L} \phi$. 
    Using (\ref{eq:wave_phi_u}) as well as (\ref{eq:wave_r_u}), one finds
    \begin{equation} \label{eq:wave_rlphi_u}
        \underline{L} (r^2 L \phi) = - X \phi + \frac{\Omega^2}{4 r^2 \partial_u r \partial_v r} \left( \frac{Q^2}{r^2} - 1 \right) \cdot r^2 L \phi.
    \end{equation}
    One finds an analogous equation for $r^2 \underline{L} \phi$:
    \begin{equation} \label{eq:wave_rlphi_v}
        L (r^2 \underline{L} \phi) = X \phi + \frac{\Omega^2}{4 r^2 \partial_u r \partial_v r} \left( \frac{Q^2}{r^2} - 1 \right) \cdot r^2 \underline{L} \phi.
    \end{equation}

    As $\Psi$ contains both $L$ and $\underline{L}$ derivatives, it is useful to also include the following null transport equations for $r^2 X \phi$. Using that $L r^2 = \underline{L} r^2 = -2$, we have:
    \begin{gather} \label{eq:wave_xphi_u}
        \underline{L}(r^2 X \phi) = r^2 \underline{L} X \phi - 2 X \phi,
        \\[1em] \label{eq:wave_xphi_v}
        L(r^2 X \phi) = r^2 L X \phi - 2 X \phi.
    \end{gather}

    By Proposition \ref{prop:scalarfield_xphi} and Lemma \ref{lem:upperbounds_lapse}, all the terms on the right hand sides of (\ref{eq:wave_rlphi_u}), (\ref{eq:wave_rlphi_v}), (\ref{eq:wave_xphi_u}), (\ref{eq:wave_xphi_v}) are $O(r^{-1})$ when $Q = 0$ and $O(r^{-2 + \alpha})$ when $Q \neq 0$. This implies that these equations can be integrated up to $r = 0$. These equations, along with the definition (\ref{eq:Psi}) for $\Psi$, yield
    \begin{gather} \label{eq:wave_psi_u}
        \underline{L} \Psi = - \frac{1}{2} r^2 \underline{L} X \phi + \frac{\Omega^2}{4 r^2 \partial_u r \, \partial_v r} \left( \frac{Q^2}{r^2} - 1 \right) \cdot r^2 \underline{L} \phi.
        \\[1em] \label{eq:wave_psi_v}
        L \Psi = \frac{1}{2} r^2 L X \phi + \frac{\Omega^2}{4 r^2 \partial_u r \, \partial_v r} \left( \frac{Q^2}{r^2} - 1 \right) \cdot r^2 L \phi.
    \end{gather}

    Applying Proposition \ref{prop:scalarfield_xphi}, there exists a constant $D$, depending on the data, $R_0$ and $\alpha$, such that
    \begin{equation} \label{eq:llpsi}
        |\underline{L} \Psi| + |L \Psi| \leq D \, r(u, v)^{-2 + 2 \beta}.
    \end{equation}
    Here $\beta$ is as defined in the statement of Proposition \ref{prop:scalarfield_asymptotics}. We now use this to prove the desired $C^{0, \beta}$ regularity: we first control the H\"older seminorm in the $u$-direction:
    \begin{align}
        |\Psi(\tilde{u}_1, v) - \Psi(\tilde{u}_0, v)|
        &= \left| \int^{\tilde{u}_1}_{\tilde{u}_0} \partial_u \Psi (\tilde{u}, v) \, d \tilde{u} \right|, \nonumber \\[1em]
        &\leq \int^{\tilde{u}_1}_{\tilde{u}_0} | r \partial_u r \cdot \underline{L} \Psi (\tilde{u}, v) | \, d \tilde{u}, \nonumber \\[1em]
        &\leq D \int^{r(\tilde{u}_0, v)}_{r(\tilde{u}_1, v)} \tilde{r}^{-1 + 2 \beta} \, d \tilde{r} = \frac{D}{2 \beta} \cdot ( r(\tilde{u}_0, v)^{2 \beta} - r(\tilde{u}_1, v)^{2 \beta}). \nonumber
    \end{align}
    We next use the inequality $x^{\beta} - y^{\beta} \leq (x - y)^{\beta}$ for $x > y > 0$ to get
    \begin{equation} \label{eq:psi_u_final}
        |\Psi(\tilde{u}_1, v) - \Psi(\tilde{u}_0, v)|
        \leq \frac{D}{2 \beta} \cdot |r^2(\tilde{u}_0, v) - r^2(\tilde{u}_1, v)|^{\beta}.
    \end{equation}

    Since the H\"older seminorm is defined with respect to $u$ and not $r$, we need to again use the upper bounded for $- r \partial_u r$; we have $|r^2(\tilde{u}_0, v) - r^2(\tilde{u}_1, v)| \leq |\tilde{u}_1 - \tilde{u}_0| \cdot \sup_{\tilde{u}_0 \leq \tilde{u} \leq \tilde{u}_1} ( - 2 r \partial_u r (\tilde{u}, v))$. Since $-r \partial_u r \leq R_1$ is uniformly bounded by Lemma~\ref{lem:setup_esfss}, we deduce from (\ref{eq:psi_u_final}) that
    \begin{equation}
        \sup_{ \substack{
                [\tilde{u}_0, \tilde{u}_1] \times \{v\} \\ \subset \mathcal{D}
        }}\frac{|\Psi(\tilde{u}_1, v) - \Psi(\tilde{u}_0, v)|}{|\tilde{u}_1 -\tilde{u}_0|^{\beta}} \leq D.
    \end{equation}
    The $\beta$-H\"older seminorm in the $v$-direction is similar, while the boundedness of $|\Psi|$ follows immediately from the seminorm bounds and the fact that $\mathcal{D}$ is a bounded domain. This completes the estimates for $\Psi$.

    We now move onto the quantity $\Xi$ defined in (\ref{eq:Xi}), describing the $O(1)$ term in the asymptotic expansion of $\phi$. From (\ref{eq:Xi}), (\ref{eq:Psi}), and $L \log r = \underline{L} \log r = - \frac{1}{r^2}$, one finds
    \begin{gather} \label{eq:wave_xi_u}
        \underline{L} \Xi = \underline{L} \phi - \frac{\Psi}{r^2} - \underline{L} \Psi \cdot \log \left( \frac{r_0}{r} \right) = - \frac{1}{2} X \phi - \underline{L} \Psi \cdot \log \left( \frac{r_0}{r} \right),
        \\[1em] \label{eq:wave_xi_v}
        L \Xi = L \phi - \frac{\Psi}{r^2} - L \Psi \cdot \log \left( \frac{r_0}{r} \right) = + \frac{1}{2} X \phi - L \Psi \cdot \log \left( \frac{r_0}{r} \right).
    \end{gather}

    From Proposition \ref{prop:scalarfield_xphi} that $|X \phi| \lesssim r^{-2 + 2 \beta}$, while the terms involving $L \Psi$ and $\underline{L} \Psi$ on the right hand side are bounded using (\ref{eq:llpsi}). Unfortunately, the presence of the $\log(r_0 / r)$ term on the right hand side produces some loss on the exponent, and we therefore only have
    \begin{equation} \label{eq:llxi}
        |L \Xi| + |\underline{L} \Xi| \leq \tilde{D} \, r^{-2 + 2\beta} \log r^{-1}.
    \end{equation}
    Given (\ref{eq:llxi}), the remainder of the proof follows the same strategy as for $\Psi$, with appropriate modifications due to the $\log r^{-1}$ term. This concludes the proof of the proposition.
\end{proof}

\begin{remark}
    The important steps of this proof are (\ref{eq:llpsi}) and (\ref{eq:llxi}). Once these estimates are found, there are other ways to characterise the regularity of $\Psi$ and $\Xi$. For instance, one can instead find uniform Sobolev regularity on null hypersurfaces of the form:
    \begin{equation*}
        \Psi \in L^{\infty}_u W^{1, p_{\beta}^-}_v \cap L^{\infty}_v W^{1, p_{\beta}^-}_u,
        \quad \Xi \in L^{\infty}_u W^{1, p_{\beta}^-}_v \cap L^{\infty}_v W^{1, p_{\beta}^-}_u,
        \quad p_{\beta} = \frac{1}{1 - \beta}.
    \end{equation*}
%
\end{remark}


%% file: bkl.tex
\section{Proofs of the main Theorems} \label{bklasymp}

In this section, we conclude the proofs of Theorems~\ref{thm:esfss} and \ref{thm:emsfss}, as well as the proof of Corollary~\ref{cor:bkl}.

\subsection{Asymptotics for \texorpdfstring{$\Omega^2$}{Ω2}} \label{bklasymp.geometry}

We firstly provide the precise asymptotics for the null lapse $\Omega^2(u, v)$, and related geometric quantities.

\begin{proposition} \label{prop:bklasymp_lapse}
    Consider a strongly singular solution to the EMSFSS system as in Section \ref{setup.sing}. Recalling the function $\Psi: \mathcal{D} \to \R$ from \eqref{eq:Psi} in Proposition \ref{prop:scalarfield_asymptotics}, there exists $D > 0$ depending on data, $R_0$ and $\alpha$, such that
    \begin{equation} \label{eq:bklasymp_lapse_upper}
        \left| \log \left( \Omega^2 \cdot \left( \frac{r_0}{r} \right)^{- \Psi^2 + 1} \right) \right| \leq D.
    \end{equation}

    In fact we find the following precise asymptotics for the \textit{gauge-invariant} quantity $g(\nabla r, \nabla r) = 4 \Omega^{-2} \partial_u r \, \partial_v r$. There exists a positive function $\mathfrak{M}: \mathcal{D} \to \R$, $C^{0, \beta, \log}$ up to the $\{ r = 0 \}$ boundary, with
    \begin{equation} \label{eq:mfrak}
        4 \Omega^{-2} \partial_u r \, \partial_v r \eqqcolon \mathfrak{M} \cdot \left( \frac{r}{r_0} \right)^{-(\Psi^2 + 1)},
    \end{equation}
    Here the H\"older exponent $\beta$ is as in Proposition \ref{prop:scalarfield_asymptotics}, and one finds the following bounds on $\mathfrak{M}$:
    \begin{equation} \label{eq:mfrak_est}
        \| \log \mathfrak{M} \|_{C^{0,\beta,\log}} + \| \mathfrak{M} \|_{C^{0,\beta,\log}} \leq D.
    \end{equation}
\end{proposition}

\begin{proof}
    We use the Raychaudhuri equation (\ref{eq:raych_u}), which we record again for convenience:
    \begin{equation*}
        \partial_u ( - \Omega^{-2} \partial_u r ) = \Omega^{-2} r |\partial_u \phi|^2,
    \end{equation*}
    We rewrite this equation as a null transport equation for the following logarithmic quantity involving $\Omega^2$, \blue{where we introduce a harmless factor of $4$ in order to be consistent with the left hand side of \eqref{eq:mfrak}}:
    \begin{equation*}
        \partial_u \log \left( \frac{\Omega^2}{- 4 \partial_u r} \right) = \frac{r}{\partial_u r} \cdot |\partial_u \phi|^2.
    \end{equation*}
    
    Converting to the gauge-invariant derivative $\underline{L}$, one finds:
    \begin{equation} \label{eq:raych_u_l}
        \underline{L} \log \left( \frac{\Omega^2}{- 4 \partial_u r} \right) = - \frac{1}{r^2} \cdot | r^2 \underline{L} \phi|^2.
    \end{equation}
    Of course, there is also the analogous equation:
    \begin{equation} \label{eq:raych_v_l}
        L \log \left( \frac{\Omega^2}{- 4 \partial_v r} \right) = - \frac{1}{r^2} \cdot | r^2 L \phi|^2.
    \end{equation}

    By Proposition \ref{prop:upperbounds_phi}, the right hand sides of (\ref{eq:raych_u_l}) and (\ref{eq:raych_v_l}) are $O(r^{-2})$, and thus borderline non-integrable as $r \to 0$. We solve this issue by subtracting from (\ref{eq:raych_u_l}) and (\ref{eq:raych_v_l}) the leading order term, which is exactly $\Psi^2 \cdot r^{-2}$. Since $L \log r = \underline{L} \log r = - \frac{1}{r^2}$, we therefore find:
    \begin{equation*}
        \underline{L} \log \left( \frac{\Omega^2}{- 4 \partial_u r} \right) - \underline{L} \left( \Psi^2 \log \left( \frac{r}{r_0} \right) \right) 
        = - \frac{1}{r^2} \cdot ( |r^2 \underline{L} \phi|^2 - \Psi^2 ) - 2 \Psi (\underline{L} \Psi) \log \left( \frac{r}{r_0} \right)
    \end{equation*}

    Taking the difference of two squares on the right hand side, and combining the terms on the left, we get:
    \begin{equation} \label{eq:lapse_u_pre}
        \underline{L} \log \left( \frac{\Omega^2}{- 4 \partial_u r} \cdot \left( \frac{r}{r_0} \right)^{- \Psi^2} \right) 
        = \frac{1}{2} (X \phi) (r^2 \underline{L} \phi + \Psi) - 2 \Psi (\underline{L} \Psi) \log \left( \frac{r}{r_0} \right).
    \end{equation}
    We now follow a similar method to the proof of Proposition \ref{prop:scalarfield_asymptotics}. By (\ref{eq:scalarfield_xphi}) and (\ref{eq:llpsi}), the right hand side of (\ref{eq:lapse_u_pre}) is bounded by $r^{-2 + 2 \beta} \log r^{-1}$:
    \begin{equation}
        \left| \underline{L} \log \left( \frac{\Omega^2}{- 4 \partial_u r} \cdot \left( \frac{r}{r_0} \right)^{- \Psi^2} \right) \right|
        \leq D \cdot r^{-2 + 2 \beta} \log \left( \frac{r_0}{r} \right),
    \end{equation}

    Hence integrating in the $\underline{L}$ direction, where as usual this involves multiplying by $r$ then integrating in $r$, we observe that the right-hand side is integrable towards $r = 0$, and 
    \begin{equation*}
        \left| \log \left( \frac{\Omega^2}{- 4 \partial_u r} \cdot \left( \frac{r}{r_0} \right)^{- \Psi^2} \right) \right| \lesssim 1.
    \end{equation*}
    We finally exponentiate this expression, and use that $- r \partial_u r \lesssim 1$ to get the sharp upper bound (\ref{eq:bklasymp_lapse_upper}) for $\Omega^2$ as a polynomial power of $r$. 

    To find the sharp asymptotics for $\mathfrak{M}$ as defined in (\ref{eq:mfrak}), the additional ingredient is to use (\ref{eq:wave_r_u}) to get:
    \begin{equation} \label{eq:rvr_u}
        \underline{L} \log ( - r \partial_v r ) = \frac{\Omega^2}{4 \partial_u r \partial_v r} \cdot \left( 1 - \frac{Q^2}{r^2} \right).
    \end{equation}
    Combining (\ref{eq:rvr_u}) with (\ref{eq:lapse_u_pre}), we therefore find that
    \begin{equation} \label{eq:mfrak_u}
        \underline{L} \log \mathfrak{M} = \frac{\Omega^2}{4 \partial_u r \partial_v r} \cdot \left( 1 - \frac{Q^2}{r^2} \right) - \frac{1}{2} (X \phi)( r^2 \underline{L} \phi + \Psi) + 2 \Psi (\underline{L} \Psi) \log \left( \frac{r}{r_0} \right).
    \end{equation}
    Similarly, in the $v$-direction we have
    \begin{equation} \label{eq:mfrak_v}
        L \log \mathfrak{M} = \frac{\Omega^2}{4 \partial_u r \partial_v r} \cdot \left( 1 - \frac{Q^2}{r^2} \right) + \frac{1}{2} (X \phi)( r^2 L \phi + \Psi) + 2 \Psi (L \Psi) \log \left( \frac{r}{r_0} \right).
    \end{equation}

    Now we may proceed in the same manner we used to \blue{derive the} H\"older continuity for $\Psi$ and $\Xi$ in Proposition \ref{prop:scalarfield_asymptotics}. Using Proposition \ref{prop:scalarfield_xphi}, Proposition \ref{prop:scalarfield_asymptotics}, and Lemma \ref{lem:upperbounds_lapse}, one has
    \begin{equation} \label{eq:mfrakll}
        |L \log \mathfrak{M}| + |\underline{L} \log \mathfrak{M}| \leq D \cdot r^{-2 + 2 \beta} \log \left( \frac{r_0}{r} \right).
    \end{equation}
    The remainder of the proof then follows exactly as in the proof of Proposition \ref{prop:scalarfield_asymptotics}.
\end{proof}

\begin{remark}
    Taking the difference of (\ref{eq:mfrak_v}) and (\ref{eq:mfrak_u}) yields the equation:
    \begin{equation} \label{eq:mfrakx}
        X \log \mathfrak{M} = 2 (X \phi) \Psi + 2 \Psi (X \Psi) \log \left( \frac{r}{r_0} \right) = 2 \Psi \cdot X \Xi.
    \end{equation}
    This can be viewed as a momentum constraint relating the quantities $ \log \mathfrak{M}$, $\Psi$ and $\Xi$ on any constant-$r$ hypersurfaces $\Sigma_{\tilde{r}}$. 

    Since $\mathfrak{M}$, $\Psi$ and $\Xi$ extend continuously to the $\{ r = 0 \}$ singularity, one would also like to view (\ref{eq:mfrakx}) as an asymptotic momentum constraint at the singularity, see \eqref{eq:asymp_momentum}.  
    However, in order for this to be true, we need $\log \mathfrak{M}$ and $\Xi$ to be regular enough to make sense of the spatial derivative $\tilde{\nabla}$ along the singularity, which is unfortunately not achieved by Theorems~\ref{thm:esfss} and \ref{thm:emsfss}.

    In order to prove an identity of the form \eqref{eq:asymp_momentum}, it seems we need to at least be able to upgrade the H\"older regularity of $\Psi_{\infty}$, $\Xi_{\infty}$ and $\mathfrak{M}_{\infty}$ to Lipschitz regularity.
\end{remark}

\subsection{Proof of Theorems \ref{thm:esfss} and \ref{thm:emsfss}} \label{bklasymp.proof}

We finally combine the results of Section \ref{upperbounds}, Section \ref{scalarfield} and Proposition~\ref{prop:bklasymp_lapse} above to complete the proofs of Theorems~\ref{thm:esfss} and \ref{thm:emsfss}. We shall prove both theorems simultaneously using the now familiar notation
\begin{equation} \label{eq:beta}
    \beta = \begin{cases} 1/2 & \text{ if } Q = 0, \\ \alpha / 2 & \text{ if } Q \neq 0 \text{ and \eqref{eq:cond_r}, \eqref{eq:cond_phi} hold.} \end{cases}
\end{equation}
With brevity in mind, when we refer to equations in Theorems~\ref{thm:esfss} and \ref{thm:emsfss} we generally only mention those of the latter, e.g.~equation (\ref{eq:emsfss_thm_tip_r_u}) will refer to (\ref{eq:esfss_thm_tip_r_u}) also.

\begin{proof}[Proof of (\ref{item:esfss_thm_i})]
    The estimate (\ref{eq:cond_pre_r}) holds in the ESFSS case by Lemma~\ref{lem:setup_esfss}, while the analogous (\ref{eq:cond_pre_r_}) in the EMSFSS model follows, assuming (\ref{eq:cond_r}), by the same lemma. To prove that $r^2(u, v)$ is $C^{1, \beta}$, we use Proposition~\ref{prop:upperbounds_higher_geometry}, specifically estimates (\ref{eq:upperbounds_rr_der}), (\ref{eq:upperbounds_rr_der2}), (\ref{eq:upperbounds_rr_der_em}) and (\ref{eq:upperbounds_rr_der2_em}).

    These estimates, combined with (\ref{eq:cond_pre_r}), give that
    \begin{equation} \label{eq:rurll}
        | L ( -r \partial_u r) | + | \underline{L} ( - r \partial_u r) | + | L ( - r \partial_v r)| + | \underline{L} ( - r \partial_v r )| \lesssim \frac{1}{r^{-2 + 2\beta}}.
    \end{equation}
    From here, a familiar argument (e.g.\ from the proof of Proposition~\ref{prop:scalarfield_asymptotics}) yields that the expressions $- r \partial_u r(u, v)$ and $- r \partial_v r (u, v)$ are $C^{0, \beta}$ in $\bar{\mathcal{D}} = \mathcal{D} \cup \mathcal{S}$. These expressions are exact multiples of the first derivatives of $r^2(u, v)$, hence $r^2$ is $C^{1, \beta}$ in $\mathcal{D}$ with respect to the double-null coordinates $(u, v)$. 

    Furthermore, by (\ref{eq:cond_pre_r}) the gradient of $r^2(u, v)$ is non-vanishing in $\bar{\mathcal{D}}$, so the level sets of $r^2$, including the limiting level set $\mathcal{S} = \{ r(u, v) = 0 \} $, are $C^{1, \beta}$ curves. It is then clear that one may parameterize using the $C^{1, \beta}$ function $v_*(u)$ as described in the theorem.
    \renewcommand{\qedsymbol}{}
\end{proof}

\begin{proof}[Proof of (\ref{item:esfss_thm_ii})]
    This follows immediately from Proposition~\ref{prop:scalarfield_asymptotics}. \blue{We make the additional comment that in the context of Theorem~\ref{thm:emsfss}, where \eqref{eq:QTS} holds in $\mathcal{D}$, Proposition~\ref{prop:scalarfield_asymptotics} together with \eqref{eq:QTS} implies that $r^2 L \phi$ and $r^2 \underline{L} \phi$ have the same sign, and therefore, that $\Psi$ in \eqref{eq:Psi} obeys $\Psi^2 \geq 1 + \alpha$.}
    \renewcommand{\qedsymbol}{}
\end{proof}

\begin{proof}[Proof of (\ref{item:esfss_thm_iii})]
    This is exactly equation (\ref{eq:mfrak}) of Proposition~\ref{prop:bklasymp_lapse}, and the following estimate (\ref{eq:mfrak_est}).
    \renewcommand{\qedsymbol}{}
\end{proof}

\begin{proof}[Proof of (\ref{item:emsfss_thm_iv})]
    \renewcommand{\qedsymbol}{}
    We pick now a point $p \in \mathcal{S}$, and consider its causal past $J^-(p)$. We shall prove the various estimates of (\ref{item:emsfss_thm_iv}) by integrating ``backwards'' from $p \in \mathcal{S}$, given that we already know that the limiting values $C_u(p) > 0$, $C_v(p) > 0$, $\Psi_{\infty}(p)$, $\Xi_{\infty}(p)$ and $\mathfrak{M}_{\infty}(p) > 0$ exist.

    We illustrate this by giving a complete proof of (\ref{eq:emsfss_thm_tip_r_u}) in Theorem~\ref{thm:emsfss}. Let $p = (u_p, v_*(u_p))$, and suppose $(u, v) \in J^-(p)$. Then one has:
    \begin{align*}
        {|} - {r} \partial_u r (u, v) - C_u(p) | 
        &\leq | - {r} \partial_u r (u, v) + r \partial_u r (u_p, v) | + | - r \partial_u r (u_p, v) + r \partial_u r (u_p, v_*(u_p)) |, \\[0.5em]
        &\leq \int_{u}^{u_p} |\partial_u ( - r \partial_u r )(\tilde{u}, v)| \, d\tilde{u} + \int_{v}^{v_*(u_p)} |\partial_v ( -r \partial_u r ) (u_p, \tilde{v})| \, d \tilde{v}, \\[0.5em]
        &= \int_{r(u_p, v)}^{r(u, v)} | r \underline{L} ( - r \partial_u r) | \, dr + \int_{0}^{r(u_p, v)} | r L ( - r \partial_u r) | \, dr \lesssim r^{2 \beta}(u, v).
    \end{align*}
    Note that the final inequality followed by using (\ref{eq:rurll}) above, and evaluating the integral $\int_0^r r^{- 1 + 2 \beta} \, dr = \frac{1}{2\beta} r^{2 \beta}$. This proves (\ref{eq:emsfss_thm_tip_r_u}), and (\ref{eq:emsfss_thm_tip_r_v}) is similar.

    For the scalar field estimates, note that a similar procedure with $\Psi$ or $\Xi$ in place of $- r \partial_u r$, gives
    \begin{gather} \label{eq:thm_tip_psi}
        |\Psi(u, v) - \Psi_{\infty}(p)| \lesssim r^{2 \beta},
        \\[0.5em] \label{eq:thm_tip_xi}
        |\Xi(u, v) - \Xi_{\infty}(p)| \lesssim r^{2 \beta} \log r^{-1}.
    \end{gather}
    Here we used (\ref{eq:llpsi}) and (\ref{eq:llxi}) to estimate the appropriate null derivatives. The definition of $\Xi$ in (\ref{eq:Xi}) then yields (\ref{eq:emsfss_thm_tip_phi}).

    To translate the estimates on gauge-invariant quantities such as $\Psi$ to estimates on the coordinate derivatives $\partial_u \phi$ and $\partial_v \phi$, we use the definition of $\Psi$ in (\ref{eq:Psi}), as well as the definition of the vector fields $L$, $\underline{L}$ and $X$. For instance, we write
    \begin{equation*}
        r^2 \partial_u \phi = ( - r \partial_u r ) \cdot \left( \Psi - \tfrac{1}{2} r^2 X \phi \right).
    \end{equation*}
    Using (\ref{eq:emsfss_thm_tip_r_u}) and (\ref{eq:thm_tip_psi}) above, as well as (\ref{eq:scalarfield_xphi}) to estimate the remaining term $r^2 X \phi$, one straightforwardly deduces (\ref{eq:emsfss_thm_tip_phi_u}). (\ref{eq:emsfss_thm_tip_phi_v}) follows analogously.

    Finally, (\ref{eq:emsfss_thm_tip_lapse}) follows by integrating (\ref{eq:mfrakll}) backwards from $p$, exactly as before, to find that
    \begin{equation*}
        |\log \mathfrak{M}(u, v) - \log \mathfrak{M}_{\infty}(p)| \lesssim r^{2 \beta} \log r^{-1}.
    \end{equation*}
    Hence by the definition (\ref{eq:mfrak}) of $\mathfrak{M}$, we find
    \begin{equation*}
        \left| \frac{\Omega^2}{4 \partial_u r \partial_v r} (u, v) - \mathfrak{M}_{\infty}^{-1} \cdot \left( \frac{r}{r_0} \right)^{\Psi_{\infty}(p)^2 + 1} \right| \lesssim r^{\Psi_{\infty}(p)^2 + 1 + 2 \beta} \log r^{-1}.
    \end{equation*}
    (\ref{eq:emsfss_thm_tip_lapse}) then follows upon multiplying by $4 \partial_u r \partial_v r$, and using (\ref{eq:emsfss_thm_tip_r_u}), (\ref{eq:emsfss_thm_tip_r_v}). \blue{We comment that, since $\Psi_{\infty}^2 \geq 1 + \alpha$ when $Q \neq 0$ but \eqref{eq:SKE} applies, the estimate \eqref{eq:emsfss_thm_tip_lapse} is an enhancement of the upper bound estimate \eqref{eq:upperbound_lapse_q} found in Lemma~\ref{lem:upperbounds_lapse}, inside the light cone $J^-(p)$.}
\end{proof}

\begin{proof}[Proof of (\ref{item:esfss_thm_v})]
    The estimate (\ref{eq:emsfss_thm_tip_hawkingmass}) for the Hawking mass $m(u, v)$ is immediate from (\ref{eq:emsfss_thm_tip_lapse}), (\ref{eq:emsfss_thm_tip_r_u}), (\ref{eq:emsfss_thm_tip_r_v}), in light of the following expression for $m(u, v)$:
    \begin{equation*}
        m(u, v) = \frac{r}{2} \cdot ( 1 + 4 \Omega^{-2} \partial_u r \partial_v r) = \frac{r}{2} + \frac{2}{r \Omega^2} (- r  \partial_u r) (-r \partial_v r).
    \end{equation*}

    To estimate the Kretschmann scalar, we list all the non-vanishing components of the Riemann tensor, whose components are evaluated with respect to the following null frame:
    \begin{equation*}
        L = \frac{1}{- r \partial_v r} \partial_v, \quad
        \underline{L} = \frac{1}{- r \partial_u r} \partial_u, \quad
        e_1 = r^{-1} \partial_{\theta}, \quad e_2 = (r \sin \theta)^{-1} \partial_{\varphi}.
    \end{equation*}
    For the EMSFSS system, one computes that, modulo the symmetries of the Riemann tensor: $\mbox{Riem}(V_1, V_2, V_3, V_4) = \mbox{Riem}(V_3, V_4, V_1, V_2) = -\mbox{Riem}(V_2, V_1, V_3, V_4)$, the only non-vanishing components are (where $A, B \in \{1, 2\}$):
    \begin{gather*}
        \mbox{Riem}(L, e_A, L, e_B) = \frac{1}{r^4} |r^2 L \phi|^2 \,  \delta_{AB}, \\[0.5em]
        \mbox{Riem}(\underline{L}, e_A, \underline{L}, e_B) = \frac{1}{r^4} |r^2 \underline{L} \phi|^2 \, \delta_{AB}, \\[0.5em]
        \mbox{Riem}(L, e_A, \underline{L}, e_B) = \frac{1}{r^4} \left[ 1 + \frac{\Omega^2}{4 \partial_u r \partial_v r} \left( 1 - \frac{Q^2}{r^2} \right) \right] \delta_{AB}, \\[0.5em]
        \mbox{Riem}(e_1, e_2, e_1, e_2) = \frac{1}{r^2} \left( 1 + 4 \Omega^{-2} \partial_u r \partial_v r \right), \\[0.5em]
        \mbox{Riem}(L, \underline{L}, L, \underline{L}) = - \frac{\Omega^2}{4 \partial_u r \partial_v r} \cdot \frac{1}{r^6} \left[ 1 - (r^2 L \phi) ( r^2 \underline{L} \phi ) + \frac{\Omega^2}{4 \partial_u r \partial_v r} \left( 1 - \frac{2 Q^2}{r^2} \right) \right].
    \end{gather*}

    Using these, alongside the estimates of (\ref{item:esfss_thm_iv}), one can deduce the estimate (\ref{eq:emsfss_thm_tip_kretschmann}), describing the leading order polynomial blow-up of the Kretschmann scalar.

    Finally, (\ref{eq:emsfss_thm_tip_phi_blow}) can be deduced using (\ref{eq:emsfss_thm_tip_phi_u}), (\ref{eq:emsfss_thm_tip_phi_v}) and (\ref{eq:emsfss_thm_tip_lapse}), while (\ref{eq:emsfss_thm_tip_f}), which is an exact equality, can be computed using the form (\ref{eq:em_doublenull}) of the electromagnetic $2$-form $F_{\mu \nu}$. This completes the proof of Theorem~\ref{thm:esfss} and Theorem~\ref{thm:emsfss}.
\end{proof}

\subsection{Several higher order estimates} \label{bklasymp.higher}

We move now to Corollary~\ref{cor:bkl}. In order to make precise statements about the foliation using the time function $\tau$, defined in \eqref{eq:tau}, such as statements regarding the second fundamental form of the embedding $S_{\tau} \hookrightarrow \mathcal{M}$, we must first understand the second order derivatives of $\tau$. \blue{As before, to allow the results of Section~\ref{bklasymp.higher} and Section~\ref{bklasymp.kasner} to apply to both the ESFSS system and the EMSFSS system, we again use $\beta$ from \eqref{eq:beta}.}

As the quantity $\Psi$ appears in \eqref{eq:tau}, this means that we must now understand third-order derivatives of $\phi$. Hence we now assume the initial data to be smooth (in fact $C^3$ is sufficient), and use the following:

\begin{proposition} \label{prop:phi3}
    Consider a strongly singular solution as described in Theorems~\ref{thm:esfss} and \ref{thm:emsfss}. Assume moreover that $(r^2, \Omega^2, \phi)$ are smooth in $\mathcal{D}$. Then there exists $D$, depending only on data, such that
    \begin{equation} \label{eq:phi3}
        |LLX\phi| + |L \underline{L} X \phi| + |\underline{L} L X \phi| + |\underline{L} \underline{L} X \phi| \leq D \, r^{-6 + \blue{\beta}}.
    \end{equation}
    Furthermore, the quantities $\Psi$ and $\mathfrak{M}$ from Theorems~\ref{thm:esfss} and \ref{thm:emsfss} obey:
    \begin{equation} \label{eq:Psi3}
        |LL\Psi| + |L \underline{L} \Psi| + |\underline{L} L \Psi| + |\underline{L} \underline{L} \Psi| \leq D \, r^{-4 + \blue{\beta}},
    \end{equation}
    \begin{equation} \label{eq:M3}
        |LL\log\mathfrak{M}| + |L \underline{L} \log \mathfrak{M}| + |\underline{L} L \log \mathfrak{M}| + |\underline{L} \underline{L} \log \mathfrak{M}| \leq D \, r^{-4 + \blue{\beta}} \log r^{-1}.
    \end{equation}
\end{proposition}

\begin{proof}
    We start with \eqref{eq:phi3}. We first rewrite \eqref{eq:wave_xphi} in terms of $L$ and $\underline{L}$; one gets
    \begin{equation} \label{eq:wave_llxphi}
        \underline{L} L X \phi = \frac{1}{r^2} L X \phi + \frac{1}{r^2} \underline{L} X \phi - \frac{\Omega^2}{4 r^2 \partial_ur \partial_v r} \cdot \left( 1 - \frac{Q^2}{r^2} \right) L X \phi + \frac{F_X}{(- r \partial_u r)(- r \partial_v r)}.
    \end{equation}
    
    We use this directly to get the $|\underline{L} L \phi|$ estimate in \eqref{eq:phi3}. We estimate the first three terms on the right hand side of \eqref{eq:wave_llxphi} using \eqref{eq:llpsi}, as well as \eqref{eq:mfrak} to deal with the $\Omega^2$ term. The $F_X$ term is then estimated using \eqref{eq:fxuv}. Combining all these allows us to conclude that $|\underline{L} L X \phi| \lesssim r^{-6 + \blue{\beta}}$, and $|L \underline{L} X \phi| \lesssim r^{-6 + \blue{\beta}}$ follows similarly.

    To get the $LLX \phi$ estimate, we commute the equation \eqref{eq:wave_llxphi} with the vector field $L$. We find:
    \begin{equation} \label{eq:transport_llxphi}
        \underline{L} L L X \phi = K_1 (u, v) \cdot L L X \phi + K_2(u, v),
    \end{equation}
    where $K_1, K_2: \mathcal{D} \to \R$ are given by the following expressions:
    \begin{equation} \label{eq:transport_llxphi_k1}
        K_1 = \frac{1}{r^2} + [\underline{L}, L]^L - \frac{\Omega^2}{4 r^2 \partial_ur \partial_v r} \cdot \left( 1 - \frac{Q^2}{r^2} \right).
    \end{equation}%
    \begin{multline} \label{eq:transport_llxphi_k2}
        K_2 = [\underline{L}, L]^{\underline{L}} \cdot \underline{L} L X \phi + \frac{2}{r^4} L X \phi + \frac{2}{r^4} \underline{L} X \phi +\frac{1}{r^2} L \underline{L} X \phi \\[0.5em]
        - L \left[ \frac{\Omega^2}{4 r^2 \partial_u r \partial_v r} \cdot \left( 1 - \frac{Q^2}{r^2} \right) \right] \cdot L X \phi
        + L \left[ \frac{F_X}{(- r \partial_u r) ( -r \partial_v r) }\right].
    \end{multline}

    By Lemma~\ref{lem:scalarfield_x_commutator_1} and Proposition~\ref{prop:bklasymp_lapse}, we have that $K_1 = \frac{1}{r^2} + O(r^{-2 + \blue{\beta}})$. Furthermore, combining all the results of Sections~\ref{scalarfield} and \ref{bklasymp}, we can also deduce that $|K_2| \lesssim r^{-8 + \blue{\beta}}$. We mention here that to estimate the final term appearing in \eqref{eq:transport_llxphi_k2}, we go back to the expression \eqref{eq:wave_xphi_inhomog_2} for $F_X$. Taking an $L$-derivative of $F_X$ therefore entails taking a further derivative of the commutator terms appearing in \eqref{eq:wave_xphi_inhomog_2}. But we have explicit expressions for the commutator terms in Lemmas~\ref{lem:scalarfield_x_commutator_1} and \ref{lem:scalarfield_x_commutator_2}, and one crucial ingredient will be to estimate the third-order derivatives $\partial_*^3 (r^2)$, where $\partial_*$ is either of $\partial_u$ or $\partial_v$.

    This would require pushing the results of Section~\ref{upperbound.higher} to one derivative further, which is straightforward, and we refer the reader to Proposition 7.4 and Proposition 7.5 of \cite{AnZhang}. The strategy in that article applies also to the $Q \neq 0$ case supplemented with \eqref{eq:cond_r}, \eqref{eq:cond_phi}, and the eventual estimate is that $\partial_*^3 (r^2) \lesssim r^{-4 + \blue{\beta}}$. The remainder of the argument deducing the estimate $|K_2| \lesssim r^{-8 + \blue{\beta}}$ is left to the reader. The upshot is that \eqref{eq:transport_llxphi} can now be written as:
    \begin{equation*}
        \underline{L} L L X \phi = \left( \frac{1}{r^2} + O(r^{-2 + \blue{\beta}}) \right) \cdot L L X \phi + O(r^{-8 + \blue{\beta}}).
    \end{equation*}
    We then integrate this expression using a familiar Gr\"onwall argument to deduce $|LLX \phi| \lesssim r^{-6 + \blue{\beta}}$. The estimate for $|\underline{L} \underline{L} X \phi|$ is similar, completing our proof of \eqref{eq:phi3}.

    To go from \eqref{eq:phi3} to \eqref{eq:Psi3}, we recall for instance \eqref{eq:wave_psi_v}, and apply $L$, to get
    \begin{equation*}
        LL \Psi = \frac{1}{2} r^2 LLX\phi - LX \phi + L \left[ \frac{\Omega^2}{4 r^2 \partial_u r \partial_v r} \cdot \left( \frac{Q^2}{r^2} - 1 \right) \cdot r^2 L \phi \right].
    \end{equation*}
    This, combined with \eqref{eq:phi3} and the estimates of Theorems~\ref{thm:esfss} and \ref{thm:emsfss}, yields $|LL\Psi| \lesssim r^{-4 + \blue{\beta}}$. The remaining estimates of \eqref{eq:Psi3} follow similarly. The final estimate \eqref{eq:M3} then follows upon applying the vector fields $L$ and $\underline{L}$ to \eqref{eq:mfrak_u} and \eqref{eq:mfrak_v}. We leave the details to the reader.
\end{proof}

\begin{remark}
    It is apparent from the proof of Proposition~\ref{prop:phi3} that similar estimates hold for higher derivatives, so long as the data has the required regularity to begin with. A general rule of thumb is that taking an $L$ or an $\underline{L}$ derivative (or equivalently $\partial_u$ or $\partial_v$) costs a factor of $r^{-2}$. Note, however, that we do not claim anything about improving such derivative estimates when differentiating in spatial directions.
\end{remark}

We now apply Proposition~\ref{prop:scalarfield_asymptotics}, Proposition \ref{prop:bklasymp_lapse}, and the above Proposition~\ref{prop:phi3} to analyze the first and second derivatives of the function $\tau$, defined in \eqref{eq:tau}. Recall for convenience that
\begin{equation*}
    \tau \coloneqq \frac{2 r_0}{\Psi^2 + 3} \cdot \mathfrak{M}^{-1/2} \cdot \left( \frac{r}{r_0} \right)^{\frac{\Psi^2 + 3}{2}}.
\end{equation*}

\begin{corollary} \label{cor:tau}
    The first and second derivatives of $\tau$ obey the following estimates: for $\tilde{L}_1, \tilde{L}_2 \in \{ L , \underline{L} \}$ one has
    \begin{equation} \label{eq:taul}
        \left | \tilde{L}_1 \tau + \mathfrak{M}^{-1/2} r_0^{-1} \cdot \left( \frac{r}{r_0} \right)^{\frac{\Psi^2- 1}{2}} \right | \leq D \, r^{\frac{\Psi^2 - 1}{2}} \cdot r^{\blue{\beta}} \log r^{-1},
    \end{equation}
    \begin{equation} \label{eq:taull}
        \left | \tilde{L}_2 (  \log |\tilde{L}_1 \tau |) + \frac{\Psi^2 - 1}{2} \cdot \frac{1}{r^2} \right| \leq D \, r^{-2 + \blue{\beta}} \log r^{-1}.
    \end{equation}
\end{corollary}

\begin{proof}
    We compute $\tilde{L}_1 \tau$ explicitly. One finds
    \begin{equation} \label{eq:taul_explicit}
        \tilde{L}_1 \tau = - \mathfrak{M}^{-1/2} r_0^{-1} \cdot \left( \frac{r}{r_0} \right)^{\frac{\Psi^2 - 1}{2}} \left[ 1 + \frac{{1}}{\Psi^2 + 3} \cdot \left( \frac{r}{r_0} \right)^2 \cdot \mathfrak{E}_{\tilde{L}_1} \right],
    \end{equation}
    where the error term $\mathfrak{E}_{\tilde{L}_1}$ is given by:
    \begin{equation} \label{eq:taul_e}
        \mathfrak{E}_{\tilde{L}_1} = \frac{4 \Psi}{\Psi^2 + 3} \cdot \tilde{L}_1 \Psi + 2 \Psi \cdot \tilde{L}_1 \Psi \cdot \log( \frac{r_0}{r} ) + \tilde{L}_1 \log \mathfrak{M}.
    \end{equation}

    By \eqref{eq:wave_psi_u}, \eqref{eq:wave_psi_v}, \eqref{eq:mfrak_u} and \eqref{eq:mfrak_v}, we find that $|\mathfrak{E}_{\tilde{L}_1}| \lesssim r^{-2 + \blue{\beta}} \log r^{-1}$. Inserting this into \eqref{eq:taul_explicit}, it is straightforward to deduce \eqref{eq:taul}. Furthermore, from Proposition~\ref{prop:phi3} and \eqref{eq:taul_e}, we find also $|\tilde{L}_2 \mathfrak{E}_{\tilde{L}_1}| \lesssim r^{-4 + \blue{\beta}}$. \blue{By taking} a logarithm of \eqref{eq:taul_explicit}, and then applying $\tilde{L}_2$, \blue{upon using $\tilde{L}_2 r = - \frac{1}{r}$ we get}
    \blue{
    \begin{align*}
        \tilde{L}_2(\log |\tilde{L}_1 \tau|) = - \frac{\Psi^2 - 1}{2} \cdot \frac{1}{r^2} + \mathfrak{E}_{\tilde{L}_1 \tilde{L}_2},
    \end{align*}
    where $|\mathfrak{E}_{\tilde{L}_1 \tilde{L}_2}| \leq C \cdot(| \tilde{L}_2 \log \mathfrak{M} | + |r^2 \tilde{L}_2 \Psi| \cdot |\mathfrak{E}_{\tilde{L}_1}| + r^2 |\tilde{L}_2 \mathfrak{E}_{\tilde{L}_1}|)$. Combining all previously mentioned estimates, \eqref{eq:taul_e} follows.
    }
\end{proof}

\subsection{Proof of Corollary~\ref{cor:bkl}} \label{bklasymp.kasner}

We conclude this section with the proof of Corollary~\ref{cor:bkl}, indicating how our strongly singular spacetimes fit into the BKL picture. To show Part \ref{item:bkl_1}, we first write
\begin{equation} \label{eq:nablatau}
    \nabla \tau = \frac{1}{g(L, \underline{L})} \left[ L \tau \cdot \underline{L} + \underline{L} \tau \cdot L \right].
\end{equation}

Therefore, using \eqref{eq:llbar} to determine $g(L,\underline{L})$,  we expand $g( \nabla \tau, \nabla \tau )$ to deduce
\begin{equation*}
    g( \nabla \tau, \nabla \tau ) = \frac{ 2 (L \tau) (\underline{L} \tau ) }{ g(L, \underline{L})} = - 4 \Omega^{-2} r^2 \partial_u r \partial_v r ( L \tau) (\underline{L} \tau) = r^2 \cdot \mathfrak{M} \left( \frac{r_0}{r} \right)^{\Psi^2 + 1} \cdot (L \tau) (\underline{L} \tau).
\end{equation*}
Then \eqref{eq:propertime} follows immediately from Corollary~\ref{cor:tau}, particularly \eqref{eq:taul}. So $\tau$ is asymptotically normalized.

We now write down the frame $e_0, e_1, e_2, e_3$ from Part~\ref{item:bkl_3} in terms of $L$, $\underline{L}$ and $\tau$:
\begin{equation*}
    e_0 = - \gamma \left( \frac{L}{L \tau} + \frac{\underline{L}}{\underline{L} \tau} \right),
    \quad e_1 = \gamma \left( \frac{L}{L \tau} - \frac{\underline{L}}{\underline{L} \tau} \right),
    \quad e_2 = r^{-1} \partial_{\theta}, \quad e_3 = r^{-1} \sin^{-1} \theta \, \partial_{\blue{\varphi}},
\end{equation*}
where $\gamma$ is chosen such that $e_0$ and $e_1$ are unit vectors. That is, $2 \gamma^2 g(L, \underline{L} )= - (L \tau) (\underline{L} \tau)$. In order to evaluate the second fundamental form, we must evaluate connection coefficients with respect to this frame. Standard observations allow us to write
\begin{equation*}
    \nabla_{\frac{L}{L\tau}} \left(\frac{L}{L \tau}\right) = \lambda_{LL} \frac{L}{L \tau}, \qquad
    \nabla_{\frac{\underline{L}}{\underline{L}\tau}} \left(\frac{L}{L \tau}\right) = \lambda_{\underline{L}L} \frac{L}{L \tau},
\end{equation*}
\begin{equation*}
    \nabla_{\frac{L}{L\tau}} \left(\frac{\underline{L}}{\underline{L} \tau}\right) = \lambda_{L\underline{L}} \frac{\underline{L}}{\underline{L} \tau}, \qquad
    \nabla_{\frac{\underline{L}}{\underline{L}\tau}} \left(\frac{\underline{L}}{\underline{L} \tau}\right) = \lambda_{\underline{L}\underline{L}} \frac{\underline{L}}{\underline{L} \tau}.
\end{equation*}

We wish to understand the leading order behavior of these coefficients $\lambda_{LL}, \lambda_{\underline{L}L}, \lambda_{L\underline{L}}, \lambda_{\underline{LL}}$. In fact an explicit computation using the EMSFSS system \eqref{eq:raych_u}--\eqref{eq:wave_phi} will yield, for $\tilde{L}_1, \tilde{L}_2 \in \{ L , \underline{L} \}$:
\begin{equation*}
    \lambda_{\tilde{L}_1 \tilde{L}_2} = \frac{1}{\tilde{L}_1 \tau} \cdot \begin{cases}
        \left[ - \tilde{L}_1 \log |\tilde{L}_2 \tau | + \displaystyle{\frac{1}{r^2}} ( 1 - |r^2 \tilde{L}_1 \phi|^2 ) \right] & \text{ if } \tilde{L}_1 = \tilde{L}_2, \\[1em]
        \left[ - \tilde{L}_1 \log |\tilde{L}_2 \tau | \right] & \text{ if } \tilde{L}_1 \neq \tilde{L}_2.
    \end{cases}
\end{equation*}

Applying Corollary~\ref{cor:tau}, particularly the estimate \eqref{eq:taull}, and Theorems~\ref{thm:esfss} and \ref{thm:emsfss}, one finds
\begin{equation} \label{eq:lambdal}
    \lambda_{\tilde{L}_1 \tilde{L}_2} = \frac{1}{\tilde{L}_1 \tau} \cdot \begin{cases} 
        \left[- \displaystyle{\frac{1}{r^2} \frac{\Psi^2 - 1}{2}} + O(r^{-2 + \blue{\beta}} \log r^{-1})\right] & \text{ if } \tilde{L}_1 = \tilde{L}_2,\\[1em]
        \left[ + \displaystyle{\frac{1}{r^2} \frac{\Psi^2 - 1}{2}} + O(r^{-2 + \blue{\beta}} \log r^{-1}) \right] & \text{ if } \tilde{L}_1 \neq \tilde{L}_2.
    \end{cases}
\end{equation}

We now use \eqref{eq:lambdal} to find the connection coefficients of $(\mathcal{M}, g)$ with respect to the frame $e_0, e_1, e_2, e_3$. We begin with the components of the second fundamental form, $k_{ij} = g(\nabla_{e_i} e_0, e_j)$. It follows from spherical symmetry that $k_{ij}$ is diagonal with respect to our choice of frame. To evaluate $k_{11}$, we use
\begin{align*}
    \blue{k_{11}} = g(\nabla_{e_1} e_0, e_1) 
    &= - \gamma^3 g \left( \nabla_{\frac{L}{L \tau} - \frac{\underline{L}}{\underline{L}\tau}} \left( \frac{L}{L \tau} + \frac{\underline{L}}{\underline{L}\tau} \right), \frac{L}{L \tau} - \frac{\underline{L}}{\underline{L} \tau} \right) \\[0.5em]
    &= - \gamma^3 \cdot ( - \lambda_{LL} + \lambda_{L \underline{L}} + \lambda_{\underline{L} L} - \lambda_{\underline{LL}} ) \cdot g \left( \frac{L}{L \tau}, \frac{\underline{L}}{\underline{L} \tau} \right) \\[0.5em]
    &= \frac{\gamma}{2} \cdot ( - \lambda_{LL} + \lambda_{L \underline{L}} + \lambda_{\underline{L} L} - \lambda_{\underline{LL}} ).
\end{align*}

Combining this with \eqref{eq:lambdal}, the estimate \eqref{eq:taull}, and Theorems~\ref{thm:esfss} and \ref{thm:esfss}, we find
\begin{align}
    \blue{k_{11}} = g(\nabla_{e_1} e_0, e_1) &= - (-2g(L, \underline{L}))^{-1/2} \cdot \left( \frac{\Psi^2 - 1}{r^2} + O(r^{-2 + \blue{\beta}} \log r^{-1} )\right) \nonumber \\[0.5em]
    &= - \mathfrak{M}^{1/2} \cdot \left( \frac{r}{r_0} \right)^{- \frac{\Psi^2 + 3}{2}} r_0^{-1} \cdot \left( \frac{\Psi^2 - 1}{2} + O(r^{\blue{\beta}} \log r^{-1} ) \right). \label{eq:kasner1}
\end{align}
For the $e_2$ and $e_3$ components of the second fundamental form, one also explicitly computes
\begin{equation} \label{eq:kasner2}
    \blue{k_{22}} = g(\nabla_{e_2} e_0, e_2) = \frac{\gamma}{r^2} \left( \frac{1}{L \tau} + \frac{1}{\underline{L} \tau } \right) = - \mathfrak{M}^{1/2} \cdot \left( \frac{r}{r_0} \right)^{- \frac{\Psi^2 + 3}{2}} r_0^{-1} \cdot \left( 1 + O(r^{\blue{\beta}}{ \log r^{-1}}) \right),
\end{equation}
\begin{equation} \label{eq:kasner3}
    \blue{k_{33}} = g(\nabla_{e_3} e_0, e_3) = g(\nabla_{e_2} e_0, e_2) = - \mathfrak{M}^{1/2} \cdot \left( \frac{r}{r_0} \right)^{- \frac{\Psi^2 + 3}{2}} r_0^{-1} \cdot \left( 1 + O(r^{\blue{\beta}}{ \log r^{-1}}) \right).
\end{equation}

From the expressions \eqref{eq:kasner1}, \eqref{eq:kasner2} and \eqref{eq:kasner3}, \blue{we may compute the mean curvature $\tr k = k_{11} + k_{22} + k_{33}$ of our foliation to be
\[ 
    \tr k = - \mathfrak{M}^{1/2} \cdot \left( \frac{r}{r_0} \right)^{- \frac{\Psi^2 + 3}{2}} r_0^{-1} \cdot \left( \frac{\Psi^2 + 3}{2} + O(r^{\blue{\beta}}{ \log r^{-1}}) \right),
\]
which immediately yields Part~\ref{item:bkl_2}. With respect to the frame $e_1, e_2, e_3$, it is straightforward to check that all other components of $k_{ij}$ vanish, and it is therefore clear that $e_1$, $e_2$, $e_3$ are eigenvectors of $\mathcal{K}_i^{\phantom{i}j}$, with eigenvalues
\[
    \frac{k_{11}}{\tr k} = \frac{\Psi^2 - 1}{\Psi^2 + 3} + O(r^{\beta} \log r^{-1}), \quad \frac{k_{22}}{\tr k} = \frac{2}{\Psi^2 + 3} + O(r^{\beta} \log r^{-1}), \quad \frac{k_{33}}{\tr k} = \frac{2}{\Psi^2 + 3} + O(r^{\beta} \log r^{-1}).
\]
This yields Part~\ref{item:bkl_4} of the corollary.}
For Part~\ref{item:bkl_3}, note that a similar computation to \eqref{eq:kasner1} yields
\begin{equation}
    g(\nabla_{e_0} e_1, e_0) = \frac{\gamma}{2}\cdot (\lambda_{LL} + \lambda_{\underline{L}L} - \lambda_{L \underline{L}} + \lambda_{\underline{LL}} ).
\end{equation}
However, now substituting \eqref{eq:lambdal} will lead to a cancellation at the leading order $O(r^{-(\Psi^2 + 3)/2})$, and we therefore find that $g(\nabla_{e_0}e_1,  e_0) = \tau^{-1} \cdot  O(r^{\blue{\beta}} \log r^{-1})$. Furthermore, we easily compute $g(\nabla_{e_0} e_2, e_0) = g(\nabla_{e_0}e_3, e_0) = 0$. Hence Part~\ref{item:bkl_3} also follows. 

Finally, we study the scalar field. Rewriting the expression \eqref{eq:emsfss_thm_phi} in terms of $\tau$, one gets
\begin{equation} \label{eq:kasnerp}
    \phi = \frac{2 \Psi}{\Psi^2 + 3} \left( \log \tau^{-1} - \frac{1}{2} \log \mathfrak{M} + \log \frac{2 r_0}{\Psi^2 + 3} \right) + \Xi.
\end{equation}
Therefore the first expression in \eqref{eq:pphi} is immediate, while the second is straightforward after combining \eqref{eq:kasnerp} with Proposition~\ref{prop:scalarfield_asymptotics}, Proposition~\ref{prop:bklasymp_lapse} and Proposition~\ref{prop:phi3} as per usual. We leave the details to the reader. This concludes our proof of Corollary~\ref{cor:bkl}.\qed

%% file: examples.tex
\section{Several important examples} \label{examples}

In this final section, we mention several explicit examples of spherically symmetric spacetimes possessing an $r=0$ spacelike singularity, mostly taken from \cite{BuonannoDamourVeneziano}. Each of these are solutions to the EMSFSS system (i.e.~\eqref{eq:einstein} coupled to \eqref{eq:scalar_wave} and \eqref{eq:maxwell}, in fact these examples are such that $F_{\mu\nu} \equiv 0$ aside from those of Section~\ref{examples.just}). We then consider each of these spacetimes in the context of Theorem~\ref{thm:esfss}.

To the best of our knowledge, every such explicit solution is such that the quantity $\Psi_{\infty}(p)$ arising in Theorem~\ref{thm:esfss} will be constant along $\mathcal{S}$. However, as illustrated in Section~\ref{examples.just}, we show every value of $\Psi_{\infty} \in \R$ can be achieved. Further, in Sections~\ref{examples.flrw} and \ref{examples.scaleinvariant}, we show that the remaining asymptotic quantities $\Xi_{\infty}(p)$ and $\mathcal{M}_{\infty}(p)$ need not be constant, though in all cases the asymptotic momentum constraint \eqref{eq:asymp_momentum} is satisfied.

\subsection{Spatially homogeneous cosmologies} \label{examples.just}

A spherically symmetric spacetime that is also spatially homogeneous\footnote{The metric \eqref{eq:static} has the additional Killing field $\mathbf{K} = \frac{\partial}{\partial x}$ -- note that in this setting any such $\mathbf{K}$ is also hypersurface orthogonal.}, known in the literature as a Kantowski-Sachs cosmology \cite{KantowskiSachs}, has a metric that we write in the following form:
\begin{equation} \label{eq:static}
    g = - h(s) ds^2 + h(s)^{-1} dx^2 + r^2(s) d \sigma_{\mathbb{S}^2}.
\end{equation}
The associated massless scalar field $\phi$ and the Maxwell field $F$ are given by:
\begin{equation} \label{eq:staticmatter}
    \phi = \phi(s), \qquad F = \frac{Q(s)}{r^2(s)} ds \wedge dx.
\end{equation}

There are explicit families of solutions to the system \eqref{eq:einstein}, \eqref{eq:scalar_wave}, \eqref{eq:maxwell} of the above form, known in the $F_{\mu\nu} \equiv 0$ case as the Just spacetimes \cite{Just, Just2, Just3}. When $Q=0$, the functions $r^2(s)$, $h(s)$ and $\phi(s)$ are given by:
\begin{equation} \label{eq:just}
    r^2(s) = s ( 1 - s ) h (s), \quad h(s) = \left( \frac{s}{1-s} \right)^{b}, \quad \phi(s) = \pm \frac{\sqrt{1 - b^2}}{2} \log \left( \frac{1-s}{s} \right).
\end{equation}
Here $b \in [-1, 1]$ is a freely chosen parameter, while $s$ lies in the open interval $s \in (0, 1)$. Note that the particular choice $b =  1$ is exactly the Schwarzschild interior metric \eqref{eq:schwarzschild}, with $2 M = 1$.

If, on the other hand, there is a non-trivial Maxwell field, then $Q(s) \equiv Q \neq 0$ is constant, and the remaining functions $r^2(s)$, $h(s)$ and $\phi(s)$ are given by
\begin{equation} \label{eq:justem}
    r^2(s) = s(1-s)h(s), \quad h(s) = \frac{Q^2}{b^2} \left[ \left( \frac{s}{1-s} \right)^{b/2} + \left( \frac{s}{1-s} \right)^{-b/2} \right]^2, \quad \phi(s) = \pm \frac{\sqrt{1-b^2}}{2} \log \left( \frac{1-s}{s} \right).
\end{equation}
As $h(s)$ is undefined at $b = 0$, we now constrain $b$ to lie in $(0, 1]$. It turns out that $b = 1$ now corresponds to a Reissner-Nordstr\"om interior spacetime.

We now write these spacetimes in a double null coordinate gauge \eqref{eq:metric}. Choosing $u$ and $v$ such that $2 du = - h(s) ds + dx$, $2 dv = - h(s) ds - dx$, one writes
\begin{equation*}
    r^2(u, v) = r^2(s), \quad \Omega^2(u, v) = \frac{4}{ h(s)}, \quad \text{where }s=s(u, v)\text{ is defined by } \int_0^{s(u, v)}h(\tilde{s}) \, d \tilde{s} = -u - v.
\end{equation*}
Note that $s$ is a past-directed timelike coordinate, and the boundary at $s = 0$ is now located at the line $u + v = 0$ in the $(u, v)$-plane. Since $r(u, v) = 0$ here, we interpret this as the spacelike singularity $\mathcal{S}$ in Theorems~\ref{thm:esfss} and \ref{thm:emsfss}. These theorems will be applicable only in the trapped region where $\frac{d r}{ds} > 0$, which includes an $s$-neighborhood of $0$.

By spatial homogeneity, the functions $\Psi(u, v)$, $\Xi(u, v)$ and $\mathfrak{M}(u, v)$ arising in these Theorems are now solely functions of $s$. We compute $\Psi(s)$ to be:
\begin{equation} \label{eq:justpsi}
    \Psi(s) = \begin{cases} 
        \displaystyle{\frac{\sqrt{1 - b^2}}{1 - 2s - b}}, & \text{ in the } Q = 0 \text{ case \eqref{eq:just},}\\[1.5em]
        \displaystyle{\frac{\sqrt{1 - b^2}}{1 - 2s + b \tanh \left( \frac{b}{2} \log \left( \frac{s}{1 - s} \right) \right)}}, & \text{ in the } Q \neq 0 \text{ case \eqref{eq:justem}.}
    \end{cases}
\end{equation}

In both cases, $\Psi_{\infty}$ is given by $\lim_{s\to0}\Psi(s) = \sqrt{ \frac{1 + b}{1 - b} }$. Therefore, as $b$ varies across $[-1, 1]$, $\Psi_{\infty}$ is allowed to take all real values in the $Q = 0$ case. However, in the $Q \neq 0$ case we must obey $\Psi_{\infty} > 1$, as only $b \in (0, 1]$ is admissible. This is not surprising, as $\Psi_{\infty} > 1$ is equivalent to the subcriticality condition \eqref{eq:SKE} holding in a neighborhood of $\mathcal{S} = \{ r = 0 \}$.

Alternatively, we can view the metrics corresponding to \eqref{eq:just} and \eqref{eq:justem} as being singular at $s = 1$ instead, with the caveat that the ``singularity'' will disappear when $b = \pm 1$. The $s = 1$ singularity is a past boundary rather than a future boundary. We compute $\Psi(s)$ to be \eqref{eq:justpsi} as before, but evaluating \eqref{eq:justpsi} as $s \to 1$ we have a different value for $\Psi_{\infty}$ depending on whether $Q = 0$ or $Q \neq 0$. Penrose diagrams for these spacetimes are given in Figure~\ref{fig:just}.

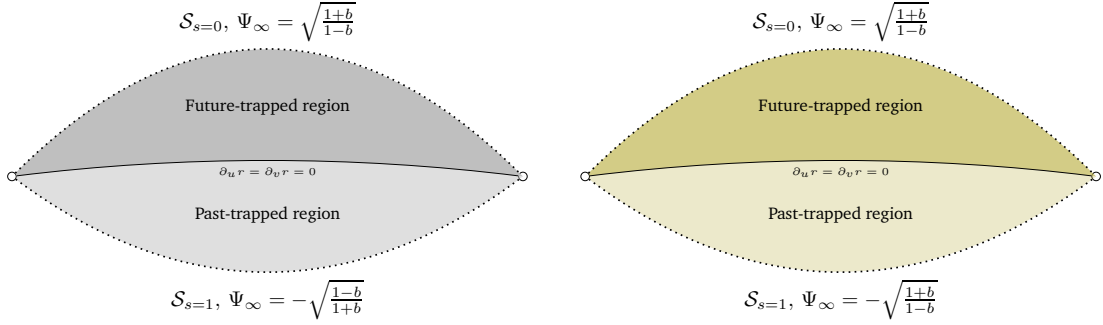
\begin{figure}[ht]
    \centering
    \begin{minipage}{.4\textwidth}\scalebox{0.8}{
    \begin{tikzpicture}[scale=0.7]
        \node (r) at (6, 0) [circle, draw, inner sep=0.5mm] {};
        \node (l) at (-6, 0) [circle, draw, inner sep=0.5mm] {};

        \path[fill=lightgray] (l) .. controls (-2, 4) and (2, 4) .. (6, 0)
            .. controls (2, 0.5) and (-2, 0.5) .. (l);
        \path[fill=lightgray, opacity=0.5] (l) .. controls (-2, -3) and (2, -3) .. (6, 0)
            .. controls (2, 0.5) and (-2, 0.5) .. (l);

        \node at (0, 1.65) {\footnotesize Future-trapped region};
        \node at (0, -0.9) {\footnotesize Past-trapped region};

        \draw [dotted, thick] (l) .. controls (-2, 4) and (2, 4) .. (r)
            node [midway, above] {$\mathcal{S}_{s = 0}, \,\Psi_{\infty} = \sqrt{ \frac{1 + b}{1 - b} }$};
        \draw [dotted, thick] (l) .. controls (-2, -3) and (2, -3) .. (r)
            node [midway, below] {$\mathcal{S}_{s = 1}, \,\Psi_{\infty} = - \sqrt{ \frac{1 - b}{1 + b} }$};
        \draw (l) .. controls (-2, 0.5) and (2, 0.5) .. (r)
            node [midway, below] {\tiny $\partial_u r = \partial_v r = 0$};
    \end{tikzpicture}}\end{minipage} \hspace{20pt}
    \begin{minipage}{.4\textwidth}\scalebox{0.8}{
    \begin{tikzpicture}[scale=0.7]
        \node (r) at (6, 0) [circle, draw, inner sep=0.5mm] {};
        \node (l) at (-6, 0) [circle, draw, inner sep=0.5mm] {};

        \path[fill=yellow!70!black, opacity=0.7] (l) .. controls (-2, 4) and (2, 4) .. (6, 0)
            .. controls (2, 0.5) and (-2, 0.5) .. (l);
        \path[fill=yellow!70!black, opacity=0.3] (l) .. controls (-2, -3) and (2, -3) .. (6, 0)
            .. controls (2, 0.5) and (-2, 0.5) .. (l);

        \node at (0, 1.65) {\footnotesize Future-trapped region};
        \node at (0, -0.9) {\footnotesize Past-trapped region};

        \draw [dotted, thick] (l) .. controls (-2, 4) and (2, 4) .. (r)
            node [midway, above] {$\mathcal{S}_{s = 0}, \,\Psi_{\infty} = \sqrt{ \frac{1 + b}{1 - b} }$};
        \draw [dotted, thick] (l) .. controls (-2, -3) and (2, -3) .. (r)
            node [midway, below] {$\mathcal{S}_{s = 1}, \,\Psi_{\infty} = - \sqrt{ \frac{1 + b}{1 - b} }$};
        \draw (l) .. controls (-2, 0.5) and (2, 0.5) .. (r)
            node [midway, below] {\tiny $\partial_u r = \partial_v r = 0$};
    \end{tikzpicture}}\end{minipage}
    \captionsetup{justification = centering}
    \caption{Penrose diagrams representing the Kantowski-Sachs cosmologies \eqref{eq:static}. Ignoring the borderline cases $b = \pm 1$, the spacetimes are bounded both in the future and in the past by a spacelike singularity. When $Q = 0$ (left), the value of $\Psi_{\infty}$ at either singularity takes all real values as $b$ varies in $[-1, 1]$. When $Q \neq 0$ (right), we must have $b \in (0, 1]$, and hence the value of $\Psi_{\infty}$ at both boundaries obeys $|\Psi_{\infty}|^2 > 1$. }
    \label{fig:just}
\end{figure}

As mentioned in the introduction, the existence of the $Q \neq 0$ example \eqref{eq:justem} is significant not only in showing that the non-emptiness of the set of spacetimes to which Theorem~\ref{thm:emsfss} applies, but also, via a straightforward gluing of some compact part of the right diagram in Figure~\ref{fig:just} to a two-ended asymptotically flat exterior, shows that there exist spacetimes with the right Penrose diagram in Figure~\ref{fig:twoended}, and where Theorem~\ref{thm:emsfss} applies to at least part of the singular boundary $\mathcal{S}$.

\subsection{The FLRW spacetimes} \label{examples.flrw}

The most famous ``cosmological'' solutions to Einstein's equation \eqref{eq:einstein} are the spatially homogeneous and isotropic \textit{Friedmann-Lema\^{i}tre-Robertson-Walker} (FLRW) spacetimes. These metrics are of the form:
\begin{equation} \label{eq:flrw}
    g_{FLRW} = - dt^2 + a^2(t) \, d \sigma_{\Sigma_0},
\end{equation}
where $\Sigma_0$ is a Riemannian manifold of constant curvature (here we take $\Sigma_0 = \mathbb{R}^3$ or $\Sigma_0 = \mathbb{H}^3$), and $d \sigma_{\Sigma_0}$ is the standard metric on $\Sigma_0$. The proper time variable $t$ takes values in $(0, + \infty)$, and $a(t) > 0$ is such that $\lim_{t \downarrow 0} a(t) = 0$.

We study FLRW metrics solving the Einstein-scalar field system (with $F_{\mu\nu} = 0$). Focusing on $\Sigma_0 = \mathbb{R}^3$ or $\mathbb{H}^3$, we write down the following explicit forms of the metric and scalar field in \textit{conformal coordinates}:
\begin{itemize}
    \item If $\Sigma_0 = \R^3$, then we have the following for $T > 0$:
        \begin{equation} \label{eq:flrwe}
            g_{FLRW_E} = T ( - dT^2 + dX^2 + dY^2 + dZ^2), \qquad \phi = \pm \frac{\sqrt{3}}{2} \log T.
        \end{equation}
    \item If $\Sigma_0 = \mathbb{H}^3$, then we denote $R^2 = X^2 + Y^2 + Z^2$, and the metric and scalar field can be written in the following form, valid for $T^2 - R^2 > 1$:
        \begin{equation} \label{eq:flrwh}
            g_{FLRW_H} = \left( 1 - \frac{1}{(T^2 - R^2)^2} \right) (-dT^2 + dX^2 + dY^2 + dZ^2), \quad \phi = \pm \frac{\sqrt{3}}{2} \log \left( \frac{T^2 - R^2 - 1}{T^2 - R^2 + 1} \right).
        \end{equation}
\end{itemize}

In both cases, we express the metric in double-null coordinates using $u = \frac{- T - R}{2}, v = \frac{- T + R}{2}$. Then the lapse $\Omega^2$, the area-radius $r$ and the scalar field $\phi$ can be written as follows.
\begin{itemize}
    \item
        In the Euclidean case $\Sigma_0 = \R^3$, we have for $(u, v) \in \mathcal{Q} = \{ (u, v): v \geq u,\, u + v < 0 \}$:
        \begin{equation} \label{eq:flrwe_dn}
            \Omega^2(u, v) = 4 (- u - v), \quad r^2(u, v) = (- u - v)(v- u)^2, \quad \phi(u, v) = \pm \frac{\sqrt{3}}{2} \log(- u - v).
        \end{equation}
    \item
        In the hyperbolic case $\Sigma_0 = \mathbb{H}^3$, we have for $(u, v) \in \mathcal{Q} = \{ (u, v): v \geq u, \, 4uv > 1, \, u < 0 \}$:
        \begin{equation} \label{eq:flrwh_dn}
            \Omega^2(u, v) = 4 \left(1 - \frac{1}{16u^2v^2}\right), \quad r^2(u, v) = \left(1 - \frac{1}{16u^2v^2}\right)(v- u)^2, \quad \phi(u, v) = \pm \frac{\sqrt{3}}{2} \log(\frac{4uv  - 1}{4uv + 1}).
        \end{equation}
\end{itemize}

It is straightforward to check that \eqref{eq:flrwe_dn} and \eqref{eq:flrwh_dn} solve the system \eqref{eq:raych_u}, \eqref{eq:raych_v}, \eqref{eq:wave_r}, \eqref{eq:wave_omega}, \eqref{eq:wave_phi} with $Q \equiv 0$. In both cases there are two parts of the boundary $\partial \mathcal{Q}$ where $r = 0$: the regular center $\Gamma$ corresponding to the fixed points of the $SO(3)$ isometry group, and the singular boundary $\mathcal{S}$ located at $u + v = 0$ and at $uv = 4$ in \eqref{eq:flrwe_dn} and \eqref{eq:flrwh_dn} respectively.

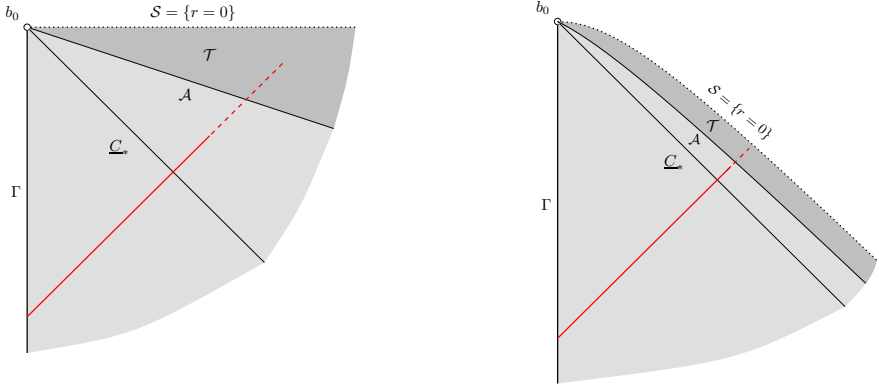
\begin{figure}[ht]
    \centering
    \begin{minipage}{.4\textwidth} \scalebox{0.6}{
    \begin{tikzpicture}[scale=0.8]
        \node (b0) at (0, 0) [circle, draw, inner sep=0.5mm] {};

        \path [fill = lightgray] (b0) -- (9, 0)
            .. controls (8.8, -1.4) .. (8.4, -2.8) -- (b0);
        \path [fill = lightgray, opacity=0.5] (b0) -- (8.4, -2.8)
            .. controls (7.5, -5) .. (6.5, -6.5) 
            .. controls (3, -8.5) .. (0, -9) -- (b0);

        \node [above left=-0.2mm of b0] {$b_0$};
        \node at (5, -0.75) {$\mathcal{T}$};

        \draw [thick] (b0) -- (0, -9) node [midway, left] {$\Gamma$};
        \draw [dotted, thick] (b0) -- (9, 0) node [midway, above] {$\mathcal{S} = \{ r = 0 \}$};
        \draw (b0) -- (6.5, -6.5) node [pos=0.45, below left] {$\underline{C}_*$};
        \draw (b0) -- (8.4, -2.8) node [pos=0.55, below left] {$\mathcal{A}$};
        \draw [color = red, thick] (0, -8) -- (5, -3);
        \draw [color = red, thick, dashed] (7, -1) -- (5, -3);
\end{tikzpicture}}
\end{minipage} \hspace{4pt}
\begin{minipage}{.4\textwidth} \scalebox{0.6}{
    \begin{tikzpicture}
        \node (b0) at (0, -2) [circle, draw, inner sep=0.5mm] {};

        \path [fill = lightgray] (b0)
            -- plot[domain = 0:7] ({\x}, {- sqrt(4 + \x^2)}) -- (7, -7.2801)
            .. controls (6.94, -7.5) .. (6.7626, -7.7946)
            -- plot[domain = 0.516:1] ({\x^-3 - \x}, {- \x - \x^-3});
        \path [fill = lightgray, opacity=0.5] (b0)
            -- plot[domain = 1:0.516] ({\x^-3 - \x}, {- \x - \x^-3})
            .. controls (6.55, -8.05) .. (6.3, -8.3)
            .. controls (4, -9.5) .. (0, -10); 

        \node [above left=-0.2mm of b0] {$b_0$};
        \node at (0, -11) {};
        \node at (3.4, -4.3) {$\mathcal{T}$};

        \draw [thick] (b0) -- (0, -10) node [midway, left] {$\Gamma$};
        \draw [dotted, thick] plot[domain = 0:7] ({\x}, {- sqrt(4 + \x * \x)})
            node [midway, above right] {};
        \node at (4, -4) [rotate=-40] {\small $\mathcal{S} = \{ r = 0 \}$};
        \node at (3, -4.6) {\small $\mathcal{A}$};
        \draw plot[domain = 1:0.516] ({\x^-3 - \x}, {- \x - \x^-3});
        \draw (b0) -- (6.3, -8.3) node [pos=0.45, below left=-0.5mm] {$\underline{C}_*$};
        \draw [color = red, thick] (0, -9) -- (3.8, -5.2);
        \draw [color = red, thick, dashed] (3.7, -5.3) -- (4.2, -4.8);
\end{tikzpicture}}
\end{minipage}
\captionsetup{justification = centering}
\caption{(Non-compactified) Penrose diagrams representing the FLRW spacetimes \eqref{eq:flrwe} (left) and \eqref{eq:flrwh} (right). Both spacetimes feature a spacelike singularity $\mathcal{S}$ at which $\Psi_{\infty} = \mp \sqrt{3}$. The singular boundary $\mathcal{S}$ is preceded by a trapped region $\mathcal{T}$ and an apparent horizon $\mathcal{A}$. In both cases, we can consider initial data given on a regular null cone emanating from $\Gamma$ which coincides exactly with the induced FLRW data on the thick red line above. Extending this initial cone to be asymptotically flat, the maximal development of this data will have a singular boundary, a portion of which exactly coincides with that of FLRW.}
\label{fig:flrw}
\end{figure}

In both cases, $\partial_u r < 0$ everywhere, and the singular boundary $\mathcal{S}$ is preceded by a trapped region $\mathcal{T} = \{(u, v): \partial_v r  < 0\}$, depicted in Figure~\ref{fig:flrw}. Therefore, we may apply Theorem~\ref{thm:esfss} to subregions which are strongly singular, and associate quantities $\Psi$, $\Xi$ and $\mathfrak{M}$ to all $(u, v) \in \mathcal{T} \cup \mathcal{S}$. In particular, we compute the quantities $r^2 L \phi$ and $r^2 \underline{L} \phi$ appearing in \eqref{eq:Psi} to be
\begin{equation*}
    r^2 L \phi = \mp \sqrt{3} \left[ 1 + \frac{2(u + v)}{-3u - v} \right], \quad
    r^2 \underline{L} \phi = \mp \sqrt{3} \left[ 1 - \frac{2(u+ v)}{u + 3v} \right] \qquad \text{for \eqref{eq:flrwe_dn}},
\end{equation*}
\begin{equation*}
    r^2 L \phi = \mp \sqrt{3} \left[ 1 + \frac{(4uv - 1)(4v^2+1)}{1 - 16uv^3} \right], \quad
    r^2 \underline{L} \phi = \mp \sqrt{3} \left[ 1 - \frac{(4uv - 1)(4u^2 + 1)}{16u^3v - 1} \right] \qquad \text{for \eqref{eq:flrwh_dn}}.
\end{equation*}

In both cases, it is apparent that we have $\Psi_{\infty} = \mp \sqrt{3}$. Using the correspondence \eqref{eq:Psikasner}, the Kasner exponents at $\mathcal{S}$ are therefore always $p_1 = p_2 = p_3 = \frac{1}{3}$. This is in accordance with the isotropy assumption i.e.~the assumption that all spatial directions should be equivalent. Furthermore, these exponents agree with the fact that $a(t) \sim t^{1/3}$ in \eqref{eq:flrw} for both of these spacetimes.

We compute also the values of $\Xi_{\infty}$ and $\mathfrak{M}_{\infty}$ for these spacetimes. Setting $r_0 = 1$ in the expressions \eqref{eq:esfss_thm_phi} and \eqref{eq:esfss_thm_lapse}, one finds:
\begin{equation*}
    \Xi_{\infty} = \mp \sqrt{3} \log( v - u ) = \mp \sqrt{3} \log R, \quad \mathfrak{M}_{\infty} = \frac{1}{4} (v - u)^6 = \frac{R^6}{4}, \quad \text{ for \eqref{eq:flrwe_dn}},
\end{equation*}
\begin{equation*}
    \Xi_{\infty} = \mp \sqrt{3} \log( \frac{v-2}{2} ) = \mp \sqrt{3} \log \frac{R}{2}, \quad \mathfrak{M}_{\infty} = 4 (v - u)^6 = 4 R^6, \quad \text{ for \eqref{eq:flrwh_dn}}.
\end{equation*}
Note, in particular, that these obey the asymptotic momentum constraint \eqref{eq:asymp_momentum}, with $\tilde{\nabla} = \frac{d}{dR}$.

Our analysis is valid for the whole spacelike singularity $\mathcal{S}$, excluding the first singularity at the center, $b_0 = \Gamma \cap \mathcal{S}$. As noted in Section~\ref{intro.related}, the point $b_0$ is excluded from our analysis as it is not preceded by trapped $2$-spheres. We note, however, that in the specific context of the FLRW spacetimes, upon reverting to the metric \eqref{eq:flrw} the point $b_0$ need not be distinguished from other points on the spacelike singularity, and still has Kasner exponents $p_1 = p_2 = p_3 = 1/3$.

Finally, one notes that the FLRW metrics \eqref{eq:flrwe} and \eqref{eq:flrwh} may be considered within the context of gravitational collapse. To see this, observe that if we consider a null cone emanating from $b \in \Gamma$, continuing it beyond $J^-(b_0)$ but not into the trapped region $\mathcal{T}$, and instead extending the cone to be asymptotically flat (see the red lines in Figure~\ref{fig:flrw}), the Penrose diagram resulting from the maximal future development of such data would be that of Figure~\ref{fig:christodoulou_collapse}. 

Moreover, the metric in a neighborhood of $b_0$ would be exactly FLRW, and have subcritical Kasner exponents of $p_1 = p_2 = p_3 = 1/3$. This \blue{construction} is particularly interesting in light of the fact that FLRW singularities are known to be stable outside of symmetry in some settings \cite{RodnianskiSpeck1, RodnianskiSpeck2, SpeckS3, FajmanUrban}.

\subsection{Scale-invariant collapsing spacetimes} \label{examples.scaleinvariant}

Our final explicit example of a spherically symmetric spacetime possessing a spacelike singularity are the scale-invariant solutions \cite{Wesson, Christodoulou_BV}. These are spherically symmetric solutions to \eqref{eq:einstein}, \eqref{eq:scalar_wave} possessing a conformal Killing vector field $\mathbf{S}$, known as the \textit{scaling vector field}, which obeys $\mathcal{L}_{\mathbf{S}} g = 2 g$.

In a suitable double-null gauge, where $\mathbf{S}$ can be written as $\mathbf{S} = u \frac{\partial}{\partial u} + v \frac{\partial}{\partial v}$, the metric has the following form, with $\Omega^2 \equiv 1$:
\begin{equation*}
    g = - du dv + r^2(u, v) d \sigma_{\mathbb{S}^2},
\end{equation*}
where the area-radius $r(u, v)$ and scalar field $\phi(u, v)$ are as follows: for $v < 0$,
\begin{equation*}
    r = \tfrac{1}{2}(v - u), \quad \phi \equiv 0,
\end{equation*}
while, for $v \geq 0$, one has for some $p \in \R$:
\begin{equation} \label{eq:scaleinvariant}
    r^2 = \tfrac{1}{4} [(1 - p) v - u ] [(1 + p) v - u ], \quad \phi = - \frac{1}{2} \log \left[ \frac{(1-p)v - u}{(1+p)v - u} \right].
\end{equation}

These solutions feature a loss of regularity across the ingoing null hypersurface $\{v = 0\}$; along any outgoing null cone $C_0 = \{ (u, v): v \geq u \}$ with $u < 0$, the quantity $\partial_v \phi$ will have a jump discontinuity at $v = 0$. Nonetheless, these solutions remain relevant in light of the fact that the Einstein-scalar field equations in spherical symmetry remain well-posed for $\partial_v \phi$ merely of bounded variation \cite{Christodoulou_BV}. In any case, for this article we are only concerned with the interior region $v > 0$, where all quantities are smooth.

This interior spacetime, whose area-radius $r$ and scalar field $\phi$ are given by \eqref{eq:scaleinvariant}, is singular if the parameter $p$ is such that $|p| > 1$. Assuming without loss of generality that $p > 1$, the singular boundary lies at $\mathcal{S} = \{ (u, v): u + (p-1) v = 0, v > 0 \}$. This boundary $\mathcal{S}$ is preceded by a trapped region $\mathcal{T}$, bounded to the past by an apparent horizon at $\mathcal{A} = \{ (u, v): u + (p^2 - 1) v = 0, v > 0 \}$.

Applying Theorem~\ref{thm:esfss}, we compute the quantity $\Psi(u, v)$ for $(u, v) \in \mathcal{T}$ to be:
\begin{equation} \label{eq:si_phi}
    \Psi(u, v) = \frac{1}{2} \left( \frac{pu}{(1-p^2)v - u} + \frac{pv}{v - u} \right).
\end{equation}
By considering \eqref{eq:si_phi} in the limit $u + (p-1)v \uparrow 0$, one finds that $\Psi_{\infty}(u, v) = 1$ for all $(u, v) \in \mathcal{S}$. Setting $r_0 = 1$ in \eqref{eq:esfss_thm_phi} and \eqref{eq:esfss_thm_lapse}, we also compute $\Xi_{\infty}$ and $\mathfrak{M}_{\infty}$ to be:
\begin{equation*}
    \Xi_{\infty} = \log \left[\frac{(1+p)v - u}{2}\right] = \log pv, \quad \mathfrak{M}_{\infty} = \tfrac{1}{4} (v-u)[(p^2-1)v + u ] = \tfrac{1}{4} p^2(p-1) v^2.
\end{equation*}
Here we used that $(p-1)v + u = 0$ on $\mathcal{S}$. In particular, it is straightforward to check that the asymptotic momentum constraint \eqref{eq:asymp_momentum} holds with spatial derivative given by $\tilde{\nabla} = \frac{\partial}{\partial v}$.

One should mention that the Penrose diagram corresponding to these spacetimes coincides with the left picture in Figure~\ref{fig:flrw}, though there are key differences in that the region to the past of $\underline{C}_*$ is now exactly Minkowski and that the value of $\Psi_{\infty}$ on $\mathcal{S}$ is now identically $1$ rather than $\sqrt{3}$.

This concludes our discussion of the scale-invariant spacetimes. Note that one could also apply this analysis to the $k$-self-similar spacetimes\footnote{Note that the purpose of \cite{Christodoulou_examples} was to obtain self-similar spacetimes containing a naked singularity. Nonetheless, there are also $k$-self-similar spacetimes in \cite{Christodoulou_examples} which contain a trapped region and spacelike singularity $\mathcal{S}$ in the interior, and we can apply Theorem~\ref{thm:esfss} to these spacetimes.} constructed by Christodoulou in \cite{Christodoulou_examples}. However, since these spacetimes are not fully explicit, and instead constructed via orbits of a $2$-dimensional dynamical system, we choose not to discuss them any further here. We instead refer the reader to \cite{Cicortas} for a detailed discussion in the specific case $k^2 = \frac{1}{3}$, and we note that it is expected that the methods of \cite{Cicortas} will generalize to all relevant values of $k$.